\documentclass[12pt]{article} 

\usepackage{amsfonts} 					
\usepackage{amssymb} 					
\usepackage{amsmath} 					
\usepackage{mathrsfs} 
\usepackage{geometry} 					
\usepackage{fancyhdr} 					
\usepackage{graphicx} 					
\usepackage{epsf} 						
\usepackage{epsfig} 					
\usepackage[bottom]{footmisc}			
\usepackage{natbib} 					
\usepackage{setspace}					
\usepackage{rotating}					
\usepackage{threeparttable} 	
\usepackage{booktabs}					
\usepackage{caption}					
\usepackage{array} 						
\usepackage{float} 						
\usepackage{forest}						
\usepackage{tikz} 						
\usepackage{pgfkeys}
\usepackage{comment}					
\usepackage{xcolor}						
\usepackage{multirow}					
\usepackage{mathtools}					
\usepackage{pdflscape} 					
\usepackage{subcaption}					

\usepackage[colorlinks, urlcolor=blue, linkcolor=blue,citecolor=blue,backref=page]{hyperref}

\usepackage{float}
\usepackage{amsmath,amsfonts,esint}
\usepackage{amsmath,amstext,rotating}
\usepackage{amsfonts,amssymb,graphics,xspace,endnotes}
\usepackage{lineno}
\usepackage{amsthm}
\usepackage{physics}
\usepackage{graphicx}
\usepackage{verbatim}
\usepackage{natbib}
\usepackage{comment}
\usepackage{algorithm}
\usepackage{algorithmic}
\usepackage[utf8]{inputenc} 
\usepackage[T1]{fontenc} 
\usepackage{pgfplots}
\pgfplotsset{compat=1.18}

\newcommand{\be}{\begin{eqnarray*}}
\newcommand{\ee}{\end{eqnarray*}}

\renewcommand{\theequation}{\thesection.\arabic{equation}}
\newtheorem{assumption}{Assumption}
\newtheorem{condition}{Condition}
\newtheorem{definition}{Definition}

\newtheorem{theorem}{Theorem}[section]
\newtheorem{lemma}{Lemma}[section]

\newcommand{\mcA}{{\mathcal A}}

\DeclareMathOperator*{\sign}{sign}
\newcommand{\PP}{{\mathbb P}}

\newcommand{\citeasnoun}{\citet}

\input{tcilatex}

\usetikzlibrary{patterns}

\usepackage{natbib}
\setcitestyle{sort&compress}

\AtBeginDocument{}

\hypersetup{
    colorlinks,
    linkcolor={blue!75!black},
    urlcolor={blue!75!black},
    citecolor={blue!75!black},
}

\title{Random Set Quantile  Estimation of  Partially Identified Discrete Response  Models\thanks{First Version: May 2022. This paper previously circulated under the title
``Sharp and Robust Estimation of
Partially Identified Discrete Response Models". We are grateful for helpful comments from  Patrick Guggenberger, Ilya Molchanov, Francesca Molinari, conference participants at the 2019 Asian ES meetings in Xiamen, the 2019 Chinese ES meetings in Guangzhou,  the 2019 Midwest Econometric Study Group in Columbus, Ohio, the 2022 Montreal Econometrics summer conference, the 2022 Advances in Econometrics conference at Toulouse School of Economics, the IPDC2024 in Orl\'{e}ans, the 2024 CEMMAP workshop, the 2024 $EC^2$ and seminar participants at Georgetown, UC Berkeley, UC Louvain, University of Bristol, University of Warwick, UC Riverside, Yale, University of Iowa, Boston University, University of Arizona, University of Miami. }}
 
\author{Shakeeb Khan\thanks{Department of Economics, Boston College.}  \and
Tatiana Komarova\thanks{Corresponding author. Faculty of Economics, University of Cambridge.  tk670@cam.ac.uk.}
\and 
Denis Nekipelov\thanks{Department of Economics, University of Virginia.}} 

\date{November 24, 2026 \\ \vspace{.5cm}}

\begin{document}
\maketitle

\spacing{1}


\begin{abstract} 
Semiparametric discrete choice models are widely used in economics, yet discrete covariates create a fundamental tension as regression coefficients that are  point identified under continuous regressors may become only partially identified. We show this can generate  estimation pathologies. Classical estimators, including e.g. the maximum score estimators of \citet{manski1985, manski1987} for binary models and a tie-indifferent
modification of the multinomial rank criterion of \citet{manski1975}, may have population maximizers that are {\em outer regions} of the identified set (\citet{komarova2013}) and converge to a random set supported on finitely many deterministic regions. Our central finding is that the limiting distribution of an inconsistent set estimator can nevertheless retain the information necessary to recover the closure of the sharp identified set, even when the population maximizer set strictly exceeds that closure. We characterize conditions under which this information can be recovered through quantiles of the limiting random set. Building on this result, we introduce the Random Set Quantile (RSQ) estimator and establish its consistency and local distributional stability. We illustrate the method through simulations and an application to the 2019 UK General Election.

\vspace{0.05in}

\begin{description}
\item[Keywords:] Maximum score estimation, Partial identification, Identified set, Robustness, Random Set Quantile, Panel data discrete choice, Multinomial choice.

\item[JEL Classification:] C14, C31, C25.

\end{description}
\end{abstract}
\spacing{1.2}
\newpage 

\section{Introduction}

Semiparametric discrete choice models specify the effects of observed regressors while leaving the distribution of unobservables unspecified. Foundational work by \citet{manski1975,manski1985,manski2} establishes conditions for point identification, including a conditional quantile restriction and suitable  variation  with a continuous regressor. When all regressors are discrete, point identification can fail (for details, see \citep{komarova2013}). Such covariates are common in applications involving income categories, education levels, geographic indicators, and treatment assignments. In this paper  we treat  discrete regressors  as an important  modeling feature.  

We demonstrate that depending on the true parameter value of the data generating process, the identified set may be either a singleton or a set of strictly positive affine dimension. We study what happens to estimation when the identified set changes its structure and measure across the parameter space. We make three main  contributions.

Our first contribution is formalizing the concepts of both an \textit{exhaustive characterization} and a
\textit{coarse characterization} of a semiparametric model. An exhaustive
characterization fully exploits all observable implications of the model and maps every observable functional of the conditional distribution of the
outcome to a constraint on the underlying parameter of interest. A coarse characterization uses a strict subset of these implications. To illustrate this distinction we show how maximum score type estimation procedures induce specific
coarse structures in three leading models. In cross-sectional binary
choice this is the maximum score estimator of
\citet{manski1985}, and in panel binary choice it is the
conditional maximum score estimator of \citet{manski1987}. 
Third, for  a multinomial model defined by the rank-order restrictions underlying \citet{manski1975}'s  criterion, we introduce a block score version of \citet{manski1975} rank criterion that treats alternatives with equal choice probabilities indifferently and coincides with the original criterion when choice probabilities are strictly ordered.\footnote{In this theoretical example, we take the maintained model's starting point to be the rank-order property linking conditional choice probabilities to utility indices.}
Using these coarse structures is essentially equivalent to using the
respective exhaustive ones when at least one regressor is continuous,
but this changes when all regressors are discrete.

Our second contribution is to show
 in full generality that a broad class of estimators based on
such coarse structures behave fundamentally differently with discrete
regressors than in settings with continuous covariates and point
identification. We establish five general theorems that characterize the weak limits of such estimators when coarse and exhaustive characterizations differ with positive probability. 

These results show that the estimators may be \textit{inconsistent}: instead of converging in probability to the identified set, they may converge weakly to a \textit{random set}. More precisely, they converge to  a draw from a finite collection of deterministic regions whose closures have union equal to the
closure of the population maximizer set.
Here the identified set is on the boundary of each region and  its closure being equal to the intersection of any strict majority of the corresponding closed regions.
These results apply to all three leading examples and, more broadly, to models satisfying the conditions of our general theorems.

In our leading examples, configurations 
in which inconsistencies of classical estimators occur   correspond to economically meaningful situations, such as indifference or threshold behavior. As a result, the failure of standard estimators is not a theoretical curiosity but a feature that can be expected in a wide range of empirical applications.
  Thus, what appears as a knife-edge phenomenon in models with a continuous regressor, becomes a first-order feature of the data once covariates are discrete, and therefore should not  be ignored in practice.

Our third contribution establishes a perhaps surprising property of the limiting distributions characterized in our general theorems. Although a classical estimator may be inconsistent and its population maximizer set may strictly contain the identified set, the distribution of the estimator can retain the information necessary to recover the closure of the identified set. Thus, identifying restrictions that are not exploited by the population objective function can be recovered from the distribution of its sample maximizers. This distinction between the identifying content of the population maximizer set and the information contained in the distribution of the corresponding estimator is central to our analysis.

Under conditions established in the paper, this information is captured by the membership probabilities of the limiting random set. In particular, every point in the closure of the identified set belongs to the limiting random set with probability one, whereas any point outside this closure belongs to it with probability at most one half. Consequently, any quantile of the limiting random set at a level strictly between one half and one coincides with the closure of the identified set. Building on this characterization, we propose a new class of estimators based on quantiles of random sets (\citet{molchanov2006book}), which we call Random Set Quantile (RSQ) estimators.

Using the majority intersection property and the limiting membership probability structure of the  random sets in our leading examples, we show the RSQ estimator recovers the closure of the identified set exactly in the limit, for any $\tau \in (1/2,1)$.\footnote{For the multinomial rank-order model, the most general version of the result requires an additional regularity condition on which rankings can arise. Without this condition, we establish the result for the linear rank score.}
We also characterize the estimators' limiting distributions under specified alternatives of 
$\alpha_n(h)=\alpha+h/\sqrt n$ where $\alpha$ is a fixed parameter, $h$ is a fixed admissible vector  and $n$ denotes the sample size.

We illustrate our methodology using results from the 2019 UK General Election, where Labour
secured only 202 of 650 constituencies, a historic low. We examine how constituency Leave voting is associated
with the index of seat retention and whether this
association differs between Labour-held and
non-Labour-held constituencies. The discrete support of these covariates in this application generates exactly the partial identification structure our theory analyzes. At the reported core plateau, the RSQ estimate is a proper subset of the closure of the full sample maximum score region and has strictly smaller affine dimension. It delivers  additional economically important restrictions on the fitted index ignored by the Maximum Score.

Our paper draws on rich prior literature on identification and inferences for point and partially identified semiparametric discrete choice models. This includes classic work on the maximum score estimator in \citet{manski1975},  \citet{manski1985}, \citet{manski2} with the distribution theory developed in \citet{kimpollard} and the smoothed version of the estimator \citet{horowitz-sms} yielding asymptotically normally distributed estimators.  This model was further studied in settings with discrete regressors\footnote{Important other work in the discrete regressor setting includes \citet{torgovitskyqe}, who develops a general
method for characterizing sharp identified sets of important reduced form parameters such as average treatment effects by extending
subdistributions, with applications to cross-sectional semiparametric binary
response models.
In other related work to the paper considered here, \citet{rosenura2022} considered finite sample properties of a related estimator that is based on moment inequalities but relies on distributional restrictions  different from the original \cite{manski1975,manski1985, manski-jasa} conditions on the unobservables (and different from the ones maintained here), with different identifying implications.
We discuss this distinction in Theoretical Example~1.} in \citet{komarova2013}. Our  set estimation problem based on a criterion is also related to
\citet{cht}.

The paper is also related to the partial identification literature, particularly the one using random set theory which was introduced to econometrics by \citet{beresteanumolinari2008} with the focus on using 
the Aumann expectation. For a broader treatment of random set methods in econometrics,
see \citet{molchanovMolinari2018}. Our approach is distinct in that we exploit the quantile of a random set rather than its expectation.

The rest of the paper is organized as follows. Section~\ref{sec:generalsetting} introduces the general model, the concepts of
exhaustive and coarse characterizations, and the three theoretical examples. 
Section~\ref{sec:main_estimators_coarse} formalizes the discrete setup and states
the five general theorems on the asymptotic behavior of 
estimators based on coarse structures, with applications to the three examples. 
Section~\ref{estimatorandcriteria} introduces 
consistency and local distributional stability as our two evaluation criteria for set estimators. 
Section~\ref{sec:quantile} develops the RSQ estimator and establishes its
consistency and distributional stability.  
Section~\ref{sec:monte_carlo} illustrates the relative performance of the RSQ and Maximum Score estimators through a simulation study that considers both point and partially identified designs. It also gives a  guide to feasible RSQ estimation in practice. 
Section~\ref{application} presents the empirical application. 
Section~\ref{sec:conclusion} concludes. All proofs are in the Appendix.

\section{General setting}
\label{sec:generalsetting}
\setcounter{equation}{0}

We consider a general semiparametric model indexed by an 
outcome $Y \in \mathbb{R}^p$, covariates $X \in \mathbb{R}^{d_X}$, and an unknown 
parameter $\alpha \in \mathcal{A} \subseteq \mathbb{R}^k$, where $\mathcal{A}$ is a compact convex set. It is  characterized by the following equivalence:
\begin{equation}
\label{eq:formulation}
\phi(F(Y|X=x)) \in C_m \quad \iff \quad \psi(x,\alpha) \in D_m, \quad m \in \mathcal{M},
\end{equation}
where $\phi: \mathcal{F} \to \mathbb{R}^{q_1}$ is a $q_1$-dimensional functional of the conditional c.d.f. $F(Y|X=x)$, and $\psi: \mathbb{R}^{d_X} \times \mathbb{R}^k \to \mathbb{R}^{q_2}$ is a $q_2$-dimensional map we will refer to as an \textit{index} map\footnote{Our leading examples will have $\psi(x,\alpha)=x'\alpha$ but we don't need this linear structure for general theory.} The sets $C_m$ belong to a collection $\mathcal{C} = \{C_m\}_{m \in \mathcal{M}}$ of subsets of $\mathbb{R}^{q_1}$, and the sets $D_m$ belong to a collection $\mathcal{D} = \{D_m\}_{m \in \mathcal{M}}$ of subsets of $\mathbb{R}^{q_2}$. We impose the following  on the model structure. 

\begin{condition}
\label{cond:disjointness}
The collections \(\mathcal C=\{C_m\}_{m\in\mathcal M}\) and \(\mathcal D=\{D_m\}_{m\in\mathcal M}\) consist of nonempty, pairwise disjoint sets. Thus, for any distinct \(m,m'\in\mathcal M\), $C_m\cap C_{m'}=\varnothing$, 
$D_m\cap D_{m'}=\varnothing$. 
The collection \(\mathcal C\) covers all admissible values of \(\phi(F(Y| X=x))\), and \(\mathcal D\) covers all admissible values of \(\psi(x,\alpha)\).
\end{condition}

Condition 1 formalizes the requirement that  distinct configurations of $\phi(F(Y|X=x))$ or $\psi(x,\alpha) $ 
are distinguishable and do not overlap.
Our next condition 
requires every regime to be attainable with positive probability under some admissible DGP.

\begin{condition}
\label{cond:positive_probability}
For every $C_m \in \mathcal{C}$, $m \in \mathcal{M}$, there exists a distribution of $X$ and a conditional c.d.f. $F(Y|X=x)$ such that $
P_X\left(
\left\{x:
\phi(F_{Y| X}(\cdot| x))\in C_m
\right\}
\right)>0$. Consequently, for the corresponding $D_m \in \mathcal{D}$, it follows that $P_X\left(
\left\{x: \psi(x,\alpha) \in D_m\right\}
\right)>0 $.
\end{condition}

 Although the formulation \eqref{eq:formulation} and both conditions allow an unrestricted regime index set $\mathcal M$, we assume henceforth that \textit{$\mathcal M$ is finite}, as in our three theoretical applications.

To formulate the next definition,  let \(\mathcal  P(a)\) denote the collection of observable distributions generated by the model with parameter \(a\in\mathcal A\) as the distribution of unobservables ranges over all specifications satisfying the maintained assumptions. Take the union  \(\mathcal P=\bigcup_{a\in\mathcal A}\mathcal P(a)\) and let \(m_{P}(x)\) denote its observable regime at \(x\) for \(P\in\mathcal P\) (that is, \(
\phi(F_P(\cdot| x))\in C_{m_P(x)}.
\)).

\begin{definition} 
\label{def:exhaustive} 
The collections \(\mathcal C\) and \(\mathcal D\) are \textbf{irreducible}  if there is no alternative valid characterization of the form \eqref{eq:formulation}, with collections $\widetilde{\mathcal C}
=\{\widetilde C_{\tilde m}\}_{\tilde m\in\widetilde{\mathcal M}}$, $\widetilde{\mathcal D}
=\{\widetilde D_{\tilde m}\}_{\tilde m\in\widetilde{\mathcal M}}$, 
covering the same admissible domains and satisfying the following
conditions, (i)–(iii): 
\begin{itemize}
    \item[(i)] Conditions \ref{cond:disjointness} and \ref{cond:positive_probability}.
    \item[(ii)] For every \(\tilde m\), there exists \(m\) such that $\widetilde C_{\tilde m}\subseteq C_m$, $\widetilde D_{\tilde m}\subseteq D_m$. 
    
    \item[(iii)] For some $\tilde{m} \in \widetilde{\mathcal{M}}$ and $m \in \mathcal{M}$, (a) $\widetilde{C}_{\tilde{m}}$ is a proper subset of $C_m$, (b) there exists a distribution of $X$ and a conditional c.d.f. $F(Y|X=x)$ such that $
P_X(
\{x:
 \phi(F(Y|X=x)) \in C_m \setminus \widetilde{C}_{\tilde{m}} 
\}) > 0$.
\end{itemize}
The characterization in \eqref{eq:formulation} is \textbf{exhaustive} if its collections are irreducible and
\[
\forall P\in\mathcal P \; \; \forall a\in\mathcal A: \quad \qquad P\in\mathcal P(a)
\quad\Longleftrightarrow\quad
\psi(x,a)\in D_{m_P(x)}
\quad\text{for }P_X\text{-almost every }x.
\]
\end{definition}

Irreducibility excludes further refinement of the regime collections within this class of characterizations. Exhaustiveness additionally requires their restrictions to be necessary and sufficient for observational compatibility and is, therefore, tied to the concept of the identified set given next.

\begin{definition} 
For a given observable distribution \(P\), the identified set is $\mathcal A_0(P):=\{a\in\mathcal A:P\in\mathcal P(a)\}.$ Under an exhaustive characterization of the model in equation \eqref{eq:formulation},
\[
\mathcal A_0(P)
=
\left\{
a\in\mathcal A:
\psi(x,a)\in D_{m_P(x)}
\text{ for }P_X\text{-almost every }x
\right\}.
\]
We suppress dependence on \(P\) when no confusion arises.
\end{definition}
 Hereafter, we fix the observable distribution \(P\) and write \(\mathcal A_0\) and \(m(x)\) for \(\mathcal A_0(P)\) and \(m_P(x)\).

In semiparametric discrete outcome models, it is standard to normalize one component of  $\alpha$ before identification analysis. We assume that such a normalization is already incorporated into the definition of $\mathcal{A}_0$ as well as the original parameter set $\mathcal{A}$.

We now formalize the notion of a coarse characterization. A coarse characterization replaces the tight 
constraint sets $D_m$ from the exhaustive characterization with 
supersets $D_m^*$,  while retaining the observable regime collection \(\mathcal C\), thereby potentially admitting additional 
parameter values as consistent with the observable data.

\begin{definition}
\label{def:coarse}
Suppose the characterization in equation \eqref{eq:formulation} is {exhaustive}. A one-directional characterization of the model, given by
\begin{equation}
\label{eq:coarse}
\phi(F(Y|X=x)) \in C_m \quad \implies \quad \psi(x,\alpha) \in D^*_m, \quad m \in \mathcal{M},
\end{equation}
is called \textbf{coarse} if the collection $\mathcal{D}^* = \{D^*_m\}_{m \in \mathcal{M}}$ satisfies the following:
\begin{itemize}
    \item[(i)] Condition \ref{cond:positive_probability}.
    \item[(ii)] For each $m \in \mathcal{M}$, the set $D^*_m$ is a superset of $D_m$, i.e., $D_m \subseteq D^*_m$.
    \item[(iii)] For some $m \in \mathcal{M}$ such that $D^*_m$ is a proper superset of $D_m$ ($D^*_m \supset D_m$), there exists a distribution of $X$ and a value of $\alpha$ such that $\PP(\psi(x,\alpha) \in D^*_m \setminus D_m) > 0$.
\end{itemize}
\end{definition}

Unlike the exhaustive sets, the coarse sets $D_m^*$ may overlap. They define the \textit{coarse set}
\[
\mathcal A^*
=\{a\in\mathcal A:\psi(x,a)\in D^*_{m(x)}
\text{ for }P_X\text{-almost every }x\},
\]
which contains $\mathcal A_0$ and can be strictly larger.
The examples in the subsection below show how classical score objectives induce these coarse restrictions.

Our interest in coarse characterizations stems from their prevalence in the literature. As argued in this paper, many significant estimators are built on coarse structures rather than exhaustive ones. This is primarily because these estimation approaches were designed with continuously distributed indices in mind, where, under suitable nondegeneracy conditions, population equality regimes occur with probability zero. With discrete covariates, these regimes can instead occur at support points of positive probability, so replacing their exhaustive  restrictions $D_m$ with the coarse restrictions $D_m^*$ can enlarge the population maximizer set and materially affect estimation.

\noindent \textbf{Objective functions and estimators.} Estimators based on coarse characterizations may take various forms, typically being sample versions of population criteria that reward compliance
with the coarse restriction $\psi(x,a)\in D^*_{m(x)}$. One example of such objective functions is based on a \emph{score}
function $s(\cdot,\cdot)$ whose first argument is an admissible value
$p$ of $\phi(F(Y| X=x))$ and whose second argument is
$t\in\mathbb R^{q_2}$.\footnote{We require the score to have the same
regime based structure for population and plug-in values of its first
argument.} In particular, for every admissible $p$ and every regime
$m$ relevant for the criterion under consideration,
\begin{align}
s(p,t)
&\text{ is constant in $t$ on }D_m^* \quad 
 \text{ when } p  \in C_m 
\label{eq:pointwise_separation1} \\ 
t\in D_m^*, & \; t'\notin D_m^* \; \quad\Longrightarrow\quad
s(p,t)>s(p,t') \quad 
\forall p \in C_m.
\label{eq:pointwise_separation2}
\end{align}
These properties apply both when
$p=\phi(F(Y|X=x))$ and when
$p=\phi(\widehat F(Y| X=x))$

With such a score function, we  can consider, for instance, the population objective function 
\begin{equation}
\label{ex_obective_population}
Q(a)=E\left[w(x)\, s\!\big(\phi(F(Y|X=x)), \psi(x,a)\big)\right], 
\end{equation} 
with weights $ w(x)\ge 0$, $\PP(w(x)>0)=1$. We assume that the relevant functions are measurable and that the population objective is finite for every candidate parameter. Under correct specification, \eqref{eq:pointwise_separation1}-\eqref{eq:pointwise_separation2} imply  that for each $x$ the
pointwise maximizers are those $a$ such that $\psi(x,a)\in D_{m(x)}^*$, and therefore
$\arg\max_a Q(a)=\mathcal A^*$. With data $\{(y_i,x_i)\}_{i=1}^n$ and  a plug-in estimator of $\phi(\widehat{F}(\cdot| X=x))$), the  direct sample criterion is 
\begin{equation}
\label{ex_obective_sample_discrete}
\widehat Q(a)
=
\sum_{x\in\mathcal X_n} w(x)
s\!\big(\phi(\widehat{F}(Y| X=x)),\psi(x,a)\big) \, \widehat \PP(X=x),
\end{equation}
where $\mathcal X_n$ is the sample support of $X$. 

One natural class of scores is based on indicator functions that reward satisfaction of the corresponding coarse restrictions, as illustrated by the binary choice models below.

Even though we introduced coarse structures first and estimators built on them second, that order was chosen for exposition only. In practice the dependence runs the other way with the  estimators already in use inducing the coarse structures through their objectives. The examples below show which structures can occur.

\subsection{Theoretical Example 1: Binary choice cross-sectional model \protect \citep{manski1975,manski1985,manski-jasa}} 

\label{sec:theoreticalMSbinary}

Consider the model\footnote{\cite{rosenura2022} work under weaker restrictions on the unobservable than those maintained in Manski’s original formulation and much of the subsequent maximum score literature. In particular, their framework does not impose absolute continuity of the conditional distribution of \(u\) or convexity of its support. Their approach therefore does not exploit their additional identifying content when the researcher is confident that they hold. This distinction matters for this theoretical example.}   
\begin{align} \label{maxscoredgp}Y & =1(X'\alpha- u\geq 0), \qquad \mathcal{Q}_{\gamma}(u|X=x)  =0,\\
\label{maxscoredgpc} & F_{u|X=x} (\cdot |x)  \text{ is continuous and strictly increasing in the  neighborhood of }0,
\end{align}
where $\mathcal{Q}_{\gamma}(\cdot|\cdot)$ denotes the conditional quantile operator, $\gamma\in(0,1)$, $X$ is a $k$-dimensional random vector and $\alpha$ is a $k$-dimensional unknown parameter vector of interest. This model is widely used in applications such as program participation or treatment choice, where the outcome reflects whether an individual crosses a latent utility threshold.

\noindent \textit{Exhaustive characterization.} To show that this model fits our setting of exhaustive characterization, denote 
$\phi(F(Y|X=x)) = \PP(Y=1|X=x)-\gamma$, $\psi(x,\alpha)=x'\alpha,$ 
and let the collections $\mathcal{C}$ and $\mathcal{D}$ in Definition \ref{def:exhaustive} consist of the same  three sets: 
\begin{align}
\label{eq:collectionC_MS} \mathcal{C}& =\{C_1=(0, +\infty), C_2=(-\infty,0), C_3=\{0\}\}, \\
\label{eq:collectionD_MS}\mathcal{D}&=\{D_1=(0, +\infty), D_2=(-\infty, 0), D_3=\{0\}\}.
\end{align}
With these definitions of $\mathcal{C}$ and  $\mathcal{D}$ the model can be equivalently written as in (\ref{eq:formulation}).

\noindent \textit{Coarse characterization induced by the maximum score estimator.} We focus on the coarse characterization induced by the maximum score estimator commonly used in this setting. Given a sample $\{(y_i,x_i)\}_{i=1}^n$, the maximum score estimator maximizes 
$ \sum_{i=1}^n (y_i-\gamma) \text{sgn}(x_i'a)$ 
over $a$ in the parameter space, 
which can  be rewritten as the maximization over $a$ of  
$$\sum_{x \in \mathcal{X}_n} (\widehat{\PP}(Y=1|x)-\gamma) \text{sgn}(x'a) \widehat{\PP}(X=x), $$
with $\mathcal{X}_n$ being the sample support , $\widehat{\PP}(Y=1|x)=\frac{1/n\sum_{i=1}^n y_i 1(x_i=x)}{1/n\sum_{i=1}^n 1(x_i=x)}$, $\widehat{\PP}(X=x)=1/n\sum_{i=1}^n 1(x_i=x)$, and $\text{sgn}(z)=1(z\geq 0)-1(z<0)$. The population version of this objective function is 
$E \left[({\PP}(Y=1|X)-\gamma)  \text{sgn}(X'a)\right]$ and the  
 set of maximizers of this population objective is 
\begin{equation} 
\label{Astar_MSbinary}
\{a \in \mathcal{A}: x'a \geq 0 \text{ if } {\PP}(Y=1|x)-\gamma >0, \text{ and } x'a <0  \text{ if } {\PP}(Y=1|x)-\gamma <0\}.  
\end{equation} 
This gives the following $\mathcal{D}^*$ for the coarse characterization induced by the MS estimation:  
\begin{equation}
\label{eq:collectionD*_MS}\mathcal{D}^*=\{D_1^*=[0, +\infty), D_2^*=(-\infty, 0), D_3^*=\mathbb{R}\}.
\end{equation} 
Thus, (\ref{Astar_MSbinary}) describes the coarse set $\mathcal{A}^*$ given by this induced  coarse characterization.

Note that maximizing the population and sample MS objective functions is equivalent to maximizing the criteria in \eqref{ex_obective_population} and \eqref{ex_obective_sample_discrete}, respectively, where $s(p,u)=\frac{|p|+p\operatorname{sgn}(u)}{2}$, $w(x)=1$, and, thus,   \begin{equation}
    \label{scorefunction_bc}
s(\phi(F(Y|X=x)),u)=|\PP(Y=1|x)-\gamma| \cdot \mathbf 1(u \in D_{m(x)}^*),  \quad w(x) = 1.
\end{equation} 
When $\PP(Y=1| x)=\gamma$, the exhaustive restriction $x'a=0$ is replaced by $D_3^*=\mathbb R$, so maximum score discards the equality restriction (see \citep{komarova2013}). Unlike the boundary change from $D_1$ to $D_1^*$, this can increase the affine dimension of the maximizer set. With discrete covariates, such equalities can occur on support points of positive probability.

This is not a minor issue. In a voting application with $\gamma=1/2$, a covariate group satisfying $x'\alpha=0$ is evenly split in its voting probabilities. In a closely contested election, a shift in support within such a group could be pivotal, and with discrete covariates these groups can make up a positive share of the electorate. Yet maximum score discards the equality restrictions that their choice probabilities reveal.

Before proceeding with further analysis, let us note that as is standard in the semiparametric literature, we normalize one nonzero coeffcient to $1$ or $-1$ (by its assumed sign), since identification is at best up to scale. This restriction is built into $\mathcal A$, $\mathcal A_0$, and $\mathcal A^*$, so their affine dimensions are at most $k-1$.

\noindent \textit{Simple example.} Let $X$ be discrete with two values, $x_1$ and $x_2$, such that $\PP(Y=1| X=x_1)>\gamma$ and $\PP(Y=1| X=x_2)=\gamma$. The exhaustive characterization implies that the identified set $\mathcal{A}_0$  consists of all $a\in\mathcal{A}$ satisfying $x_1'a>0$, $x_2'a=0.$ By contrast, the coarse set $\mathcal{A}^*$ induced by the maximum score estimator  consists of all $a\in\mathcal{A}$ such that $x_1'a\ge 0.$  If $\mathcal{A}$ has nonempty relative interior in the $(k-1)$-dimensional normalized space and is large enough,  then both $\mathcal{A}$ and $\mathcal{A}^*$ have affine dimension $k-1$, whereas $\mathcal{A}_0$ has affine dimension $k-2$.\footnote{This dimension statement assumes that the identified set contains a point interior to $\mathcal A$ in the normalized parameter space and that the equality  restriction is nontrivial on that space.}

To sum up, in this theoretical example in discrete designs with some $\PP(Y=1| X=x)-\gamma$ being in $C_3$, the coarse set  
$\mathcal{A}^*$ is typically a strict superset of $\mathcal{A}_0$ and typically has a larger dimension.

\subsection{Theoretical Example 2: Static panel data \protect \citep{manski1987}}
\label{sec:theoreticalMSpanel}

\citet{manski1987} proposed a semiparametric extension of the panel binary choice logit model of \citet{andersen70}:
\begin{equation}\label{discrete:choice:panel}
 Y_{is}={\bf 1}\left(c_i+{X}_{is}'{\alpha}+\epsilon_{is}>0\right)
\end{equation}
where $i=1,2,...n$ are the cross-sectional units and $ s=1,2$ are the time periods. The binary variable $Y_{is}$ and the $k$-dimensional regressor vector $X_{is}$ are each observed and the parameter of 
interest is the $k \times 1$ vector ${\alpha}$. The variables not observed in the data are $c_i$, and $\epsilon_{is}$, the former being the fixed effect. \citet{manski1987} imposed no  specific distributions  on unobservables.

Denote $\mathbf{X}_i=(X_{i1},X_{i2})$ and in the spirit  of \citeasnoun{manski1987}, suppose  that  
\begin{itemize}
    \item[(i)] $F_{\epsilon_{i1}|\mathbf{X}_i,c_i}(\cdot|\cdot)=F_{\epsilon_{i2}|\mathbf{X}_i,c_i}(\cdot|\cdot)$ for all  $(c_i,\mathbf{X}_i)$ in the joint support. 
    \item[(ii)] The support of $\epsilon_{it}\,\big|\,\mathbf{X}_i,c_i$ is $\mathbb{R}$.
    \item[(iii)] For every $\mathbf x$ in the support of $\mathbf X_i$
such that
$\PP(Y_{i2}=1\mid\mathbf X_i=\mathbf x)
=
\PP(Y_{i1}=1\mid\mathbf X_i=\mathbf x)$, we have $\PP(Y_{i1}\neq Y_{i2}\mid\mathbf X_i=\mathbf x)>0$.

\end{itemize}
Note that these conditions, in particular, imply that $\PP(Y_{is}=1|\mathbf X_i=\mathbf x) \in (0,1)$, $s=1,2$. 

\noindent \textit{Exhaustive characterization.} To obtain an exhaustive characterization (\ref{eq:formulation})  of this model under the stated assumptions,   define 
$$\phi(F(Y_{i1},Y_{i2}|\mathbf X_i=\mathbf x)) = \PP(Y_{i2}=1|\mathbf X_i=\mathbf x)- \PP(Y_{i1}=1|\mathbf X_i=\mathbf x), \quad \psi(\mathbf X_i,a)=X_{i2}'a-X_{i1}'a.$$ The exhaustive characterization is then obtained using these functionals and collections $\mathcal{C}$ and $\mathcal{D}$ in (\ref{eq:collectionC_MS}) and (\ref{eq:collectionD_MS}), respectively. 

\noindent \textit{Conditional maximum score estimator and the induced coarse structure.}  Given a sample $\{(y_{i1},y_{i2},x_{i1},x_{i2})\}_{i=1}^n$,  the classical estimation approach  in this literature is the maximization over index parameter $a$ of the {\it conditional maximum score} objective proposed by \citeasnoun{manski1987}: 
\begin{equation}\label{MS:objective:panel}
\frac{1}{n} \sum_{i=1}^n (y_{i2}-y_{i1})\text{sgn}\left((x_{i2}-x_{i1})'a\right).
\end{equation} 
This is equivalent to the maximization over $a$ of 
$$\sum_{\mathbf x \in \mathcal{X}_n} (\widehat{\PP}(Y_{2}=1|\mathbf x) -\widehat{\PP}(Y_{1}=1| \mathbf x))\text{sgn}\left((x_{2}-x_{1})'a\right) \widehat{\PP}(\mathbf X=\mathbf x),$$
where $\mathcal{X}_n$ is the sample support of $\mathbf X$, $\mathbf x=(x_1,x_2)$, $\widehat{\PP}(Y_s=1|\mathbf x)$  and $\widehat{\PP}(\mathbf X=\mathbf x)$ are the plug-in estimators for ${\PP}(Y_s=1|\mathbf x)$ and ${P}(\mathbf X=\mathbf x)$, respectively,  analogous to those  in Theoretical Example 1.  The corresponding population objective function takes the form 
$$E\left( ({\PP}(Y_{2}=1|\mathbf X) -{\PP}(Y_{1}=1| \mathbf X))\text{sgn}\left((X_{2}-X_{1})'a\right)\right).$$
Analogously to the cross-sectional binary choice in Theoretical Example 1, the optimization of this objective induces the coarse structure with the same collection $\{D_m^*\}$ in \eqref{eq:collectionD*_MS}.  The population and the sample objective functions that have the same maximizers as the conditional maximum score can be written in the form 
\eqref{ex_obective_population} and \eqref{ex_obective_sample_discrete} with the same $s(p,u)$ as in Theoretical Example 1, which gives  $$s(\phi(F(Y_1,Y_2|\mathbf X=\mathbf x)),u)
=
\left|
\PP(Y_2=1|\mathbf x)
-\PP(Y_1=1|\mathbf x)
\right|
\mathbf 1\{u\in D^*_{m(\mathbf x)}\},
\qquad w(\mathbf x)=1.$$ 

The coarse criterion drops an
equality restriction when the two conditional choice
probabilities coincide. With discrete covariates, this can
affect support points of positive probability.

Just as in Theoretical Example 1, we normalize one of slope coefficients in  $\alpha$ to 1 or $-1$ and assume that the parameter space $\mathcal{A}$ has a relative interior in the $(k-1)$-dimensional normalized space and is large enough to ensure that $\mathcal{A}^*$ also has a relative interior  in that space. Case $m(x)=3$ happening with a positive probability typicaly results in  $\mathcal{A}_0$ having a strictly smaller affine dimension than the coarse set $\mathcal{A}^*$.\footnote{This is provided that $(x_2-x_1)'a=0$ is not identically satisfied
on the normalized parameter space for $x$ with $m(x)=3$.}  In this case there is  a strictly positive Hausdorff distance between $\mathcal{A}_0$ and  $\mathcal{A}^*$.

\noindent \textit{Simple example.} Let $\mathbf X$ be discrete with two values, $\widetilde{\mathbf x}$ and $\mathbf x^{\diamond}$, such that  $\PP(Y_2=1|\mathbf X=\widetilde{\mathbf x})=\PP(Y_1=1| \mathbf X=\widetilde{\mathbf x})$ and $\PP(Y_2=1|\mathbf X=\mathbf x^{\diamond})>\PP(Y_1=1|\mathbf X=\mathbf x^{\diamond})$. In other words, 
$\phi(F(Y_{1},Y_{2}|\mathbf X=\widetilde{\mathbf x})) \in C_3$, $\phi(F(Y_{1},Y_{2}|\mathbf X=\mathbf x^{\diamond})) \in C_1.$ The exhaustive characterization  implies that the identified set $\mathcal{A}_0$ consists of all $a \in \mathcal{A}$  satisfying  
$\psi(\widetilde{ \mathbf x},a) \in D_3$,  $\psi(\mathbf x^{{\diamond}},a) \in D_1,$
or equivalently, 
$(\widetilde{x}_{2}-\widetilde{x}_{1})'a=0, \quad (x^{\diamond}_{2}-x^{\diamond}_{1})'a>0$ 
(affine dimension of $\mathcal{A}_0$ is $k-2$\footnote{Again, this is provided that $\mathcal A_0$ contains a point interior
to $\mathcal A$ in the normalized parameter space,
and the displayed equality is nontrivial on that space.}).  
In contrast, $\mathcal{A}^*$ consists of all $a$ (subject to normalization) satisfying only 
$( x^{\diamond}_{2}- x^{\diamond}_{1})'a \geq 0$ 
(affine dimension of $\mathcal{A}^*$ is $k-1$).

\subsection{Theoretical Example 3: Multinomial rank-order model}
\label{sec:theoreticalmultinomial} 

Let \(Y\in\{1,\ldots,J\}\),and suppose \(\PP(Y=j| x)>0\) for every \(j\) and, for \(P_X\)-almost every \(x\),
\begin{equation}\label{rank-order}
\PP(Y=j| x)\gtreqless \PP(Y=k| x)
\quad\Longleftrightarrow\quad
 x_j'\alpha_j\gtreqless x_k'\alpha_k,
 \qquad j,k=1,\ldots,J,
 \end{equation}
 where the equivalence applies separately to \(>\), \(=\), and \(<\). These restrictions follow from \citet{manski1975}'s additive random utility model with continuous, conditionally i.i.d. errors of full support. We do not  rely on random utility rationalizability, however, and instead impose the rank-order property directly, taking \eqref{rank-order} as the definition of our maintained multinomial model.\footnote{\citet{komarova2026rankings} analyzes various behavioral and distributional extensions of random utility models and of \eqref{rank-order}. Some of the results there show that \eqref{rank-order} can hold outside the original \citet{manski1975} setting.} 
 
Since \eqref{rank-order} is invariant to a common location shift of all indices and to a common rescaling, we impose one common location normalization and one common scale normalization. The parameter space $\mathcal A$ (and, hence. sets $\mathcal{A}_0$ and $\mathcal{A}^*$) is understood to incorporate these normalizations.

To express \eqref{rank-order}  in the framework of Section  \ref{sec:generalsetting}, define the functionals 
$\phi(F(Y| \mathbf X=\mathbf x))
=
\bigl(\PP(Y=1|\mathbf x),\ldots,\PP(Y=J|\mathbf x)\bigr)'$, 
$
\psi(\mathbf x,a)
=
\bigl(x_1'a_1,\ldots,x_J'a_J\bigr)'$, 
and for each pair $(j,k)$ with $j<k$, let 
$$\widetilde C_{jk;>}=\{c\in\mathbb R^J:\ c_j>c_k\},\quad
\widetilde C_{jk;<}=\{c\in\mathbb R^J:\ c_j<c_k\},\quad
\widetilde C_{jk;=}=\{c\in\mathbb R^J:\ c_j=c_k\},$$
and analogously $\widetilde D_{jk;>}$, $\widetilde D_{jk;<}$, and
$\widetilde D_{jk;=}$ in $\mathbb R^J$.
Let $\kappa_{jk}\in\{>,<,=\}$ denote one of the three relations. Collections $\mathcal C$ and $\mathcal D$ consist of all nonempty intersections
$C_m=\bigcap_{j<k}\widetilde C_{jk;\kappa_{jk}}$, $
D_m=\bigcap_{j<k}\widetilde D_{jk;\kappa_{jk}},$ 
taken over all assignments $\{\kappa_{jk}\}_{j<k}$.
These sets correspond to complete (possibly weak) orderings of the choice probabilities/utility indices and form irreducible collections. The resulting characterization is exhaustive for the maintained rank-order model: a candidate satisfies its restrictions exactly when its index ordering agrees with the observed probability ordering at almost every support point. These are all the restrictions imposed by this model. 

\vskip 0.05in 

\noindent \textit{Original multinomial maximum score.}
The multinomial maximum score estimator proposed by \citeasnoun{manski1975} maximizes,
over $a$, the function 
$\frac{1}{n}\sum_{i=1}^n\sum_{j=1}^J
\mathbf 1(Y_i=j)\,
W (\sum_{k\neq j}\mathbf 1(x_{ij}'a_j>x_{ik}'a_k)),$ 
where $W(\cdot)$ is a strictly increasing sequence of real numbers. That is, $W(J-1) > W(J-2) > \cdots > W(1) > W(0)$. Equivalently  
the sample objective function can be written as $ \sum_{\mathbf x \in \mathcal{X}_n} \sum_{j=1}^{J}  \widehat{\PP}(y={j}|\mathbf{x}) \, W\left( \sum_{k \neq j} \mathbf{1}\left( x_{j}' a_j > x_{k}' a_k \right) \right) \widehat{\PP}(\mathbf{X}=\mathbf{x})$ and its  population version is 
\begin{equation}
\label{eqLoriginalMSmult} E_{\mathbf{X}}\left( \sum_{j=1}^{J}  {\PP}(y={j}|\mathbf{x}) \, W( \sum_{k \neq j} \mathbf{1}\left( x_{j}' a_j > x_{k}' a_k \right) )\right).
\end{equation} 
When conditional choice probabilities are strictly ordered,
(\ref{eqLoriginalMSmult}) is maximized by aligning deterministic utilities in the same way, so the population maximizers coincide with $\mathcal A_0$.

When $\PP(Y=j|\mathbf x)=\PP(Y=k|\mathbf x)$ for some $j\neq k$, the objective breaks ties via $x_{ij}'a_j-x_{ik}'a_k$, favoring strict inequalities over equality. Thus the maximizers need not include $\mathcal A_0$, unlike the binary case where ties impose no restriction.

\vskip 0.05in

\noindent \textit{A tie-indifferent (block) maximum score objective.}
We modify the objective to coincide with the original without ties, but to impose no ordering restrictions when ties occur, yielding a superset of $\mathcal{A}_0$. This mirrors the resolution that we saw in the binary choice cases and ensures that the coarse $\mathcal{A}^*$ contains the identified set $\mathcal{A}_0$. 

Let $p$ denote a generic probability vector in the
$(J-1)$-dimensional probability simplex, and let
$\mathfrak B(p)
=
(B_1(p),\ldots,B_{R(p)}(p))$
denote the ordered partition of $\{1,\ldots,J\}$ into probability
tie blocks, ordered from highest to lowest probability. 
Since $B_r(p)$ is a probability tie block, $\frac{1}{|B_r(p)|}
\sum_{j\in B_r(p)}p_j$ is simply
the common probability of the alternatives in that block.

For $u\in\mathbb R^J$, define
\[
N_r(u;p)
=
\sum_{s\neq r}
|B_s(p)| \,
\mathbf 1
\left\{
\min_{j\in B_r(p)}u_j
>
\max_{k\in B_s(p)}u_k
\right\}.
\]
Thus $N_r(u;p)$ is the number of alternatives belonging to blocks
that lie strictly below block $r$ in the ordering induced by $u$.
Define
\begin{equation}
\label{eq:block_score}
s^{\mathrm{blk}}(p,u)
=
\sum_{r=1}^{R(p)}
\left(\frac{1}{|B_r(p)|}
\sum_{j\in B_r(p)}p_j \right)
\sum_{q=0}^{|B_r(p)|-1}
W\!\left(N_r(u;p)+q\right),
\qquad
u\in\mathbb R^J,
\end{equation}
where $W(0),\ldots,W(J-1)$ is strictly increasing. When all conditional choice probabilities are distinct, every
block is a singleton. In that case $N_r(u;p)
=
\sum_{k\neq j}
\mathbf 1\{u_j>u_k\}$ for the unique $j\in B_r(p)$, and
\eqref{eq:block_score} coincides with the original multinomial
maximum score criterion.

When a probability tie block has size greater than one,
\eqref{eq:block_score} assigns that block the collection of rank
weights corresponding to the positions occupied by the block,
without assigning those weights to particular alternatives within
the block. Consequently, the score is invariant to the ordering
of indices within a probability tie block.

The corresponding population objective is
\begin{equation}
\label{eq:block_population_objective}
Q(a)
=
E_{\mathbf x}\!\left[
s^{\mathrm{blk}}\!
\bigl(
\phi(F(Y|\mathbf x)),
\psi(\mathbf x,a)
\bigr)
\right],
\end{equation}
and its sample analogue is
\begin{equation}
\label{eq:block_sample_objective}
\widehat Q(a)
=
\sum_{\mathbf x\in\mathcal X_n}
s^{\mathrm{blk}}\!
\bigl(
\phi(\widehat F(Y|\mathbf x)),
\psi(\mathbf x,a)
\bigr)
\widehat \PP(\mathbf X=\mathbf x).
\end{equation}

For a fixed population probability vector $p$, the pointwise
maximizers of \eqref{eq:block_score} are $u$ that
respect the strict ordering between distinct probability blocks,
while the score imposes no restriction on the ordering within a
tie block. The Appendix provides the corresponding verification.

\vskip 0.05in

\noindent \textit{Coarse characterization induced by the tie-indifferent  maximum score objective.} As in the previous two theoretical examples, we now take a closer look at which coarse structure is induced by our tie-indifferent objective. For a given $\mathbf{x}$, we have its  exhaustive regime $m(\mathbf{x})$ (thus, with $\phi( F(Y|\mathbf{x}))\in C_{m(\mathbf{x})}$) 
and the block partition $\mathfrak B(\phi( F(Y|\mathbf{x})))$ of the alternatives. The score $s^{\mathrm{blk}}\!\bigl(\phi(F(Y|\mathbf{x})),\psi(\mathbf{x},a)\bigr)$ 
 is maximized (pointwise in $\mathbf{x}$) whenever indices respect the strict
\emph{between-block} ordering, while imposing \emph{no restriction within blocks}.
This induces the following coarse set for the regime $m(x)$:
\begin{equation}
\label{eq:Dm_star_block}
D^*_{m(\mathbf{x})}
:=
\Bigl\{u\in\mathbb R^J:\ 
\min_{j\in B_r(\phi( (F(Y|\mathbf{x}))))} u_j>\max_{k\in B_{r'}(\phi( (F(Y|\mathbf{x}))))} u_k\ \ \forall r<r'\Bigr\},
\end{equation}
where $\mathfrak B(\phi( F(Y|\mathbf{x})))=(B_1(\phi( F(Y|\mathbf{x}))),\ldots,B_{R(\mathbf{x})}(\phi( F(Y|\mathbf{x}))))$.
Equivalently, \eqref{eq:Dm_star_block} keeps all strict inequalities implied by strict probability
orderings across blocks, and \emph{drops all equality restrictions} implied by ties within blocks. In other words,  for a regime $m$ with relations $\kappa_{jk}$,
$$
D_m^*
=
\bigcap_{j<k}
\begin{cases}
\widetilde D_{jk;>}, & \kappa_{jk}  \text{ is } >,\\
\widetilde D_{jk;<}, & \kappa_{jk} \text{ is } <,\\
\mathbb R^{J}, & \kappa_{jk}\text{ is } =.
\end{cases} \quad \supseteq  \quad D_m=
\bigcap_{j<k}
\begin{cases}
\widetilde D_{jk;>}, & \kappa_{jk}  \text{ is } >,\\
\widetilde D_{jk;<}, & \kappa_{jk} \text{ is } <,\\
\widetilde D_{jk;=}, & \kappa_{jk}\text{ is } =.
\end{cases} .
$$
The corresponding coarse set is 
$\mathcal A^*
=
\bigl\{a\in\mathcal A:\ \psi(\mathbf x,a)\in D^*_{m(\mathbf x)}\ \text{for a.e.\ }\mathbf x\bigr\},$ 
which always satisfies $\mathcal A_0\subseteq \mathcal A^*$ and is strict superset of $\mathcal A_0$ when a tie  occurs with a positive
probability (e.g. under discrete regressors), and at least one of such ties gives an  equality restriction  which is nontrivial on that space.  $\mathcal A^*$ maximizes  \eqref{eq:block_population_objective}. 

\vskip 0.05in 

\noindent \textit{Simple example} To illustrate the difference between the identified set and the coarse set,  take $J=3$ and suppose $\mathbf{x}$ is discrete and its  support consists of two points $\widetilde{\mathbf{x}}$ and $\mathbf{x}^{\diamond}$ for which we have 
${\PP}(Y=1|\widetilde{\mathbf{x}}) = {\PP}(Y=2|\widetilde{\mathbf{x}})< {\PP}(Y=3|\widetilde{\mathbf{x}})$, ${\PP}(Y=1|\mathbf{x}^{\diamond}) < {\PP}(Y=2|\mathbf{x}^{\diamond})= {\PP}(Y=3|\mathbf{x}^{\diamond}), $ 
In other words, 
$\phi(F(y|\widetilde{\mathbf{x}})) \in C_{m_1}=\widetilde{C}_{12; =} \cap \widetilde{C}_{13; < }\cap \widetilde{C}_{23; < }$, $\phi(F(y|\mathbf{x}^{\diamond})) \in C_{m_2}=\widetilde{C}_{12; <} \cap \widetilde{C}_{13; < }\cap \widetilde{C}_{23; = }.$ The exhaustive characterization implies 
$\widetilde{x}_1'a_1=\widetilde{x}_2'a_2<\widetilde{x}_3'a_3$, $ x^{\diamond '}_1 a_1<x^{\diamond '}_2 a_2=x^{\diamond '}_3 a_3.$ 

Under the block maximum score objective, alternatives $1$ and $2$ for $\widetilde{\mathbf{x}}$ and  alternatives $2$ and $3$ for $\mathbf{x}^{\diamond}$ form single blocks,
and the induced coarse restriction is
$\max\{\widetilde{x}_1' a_1,\widetilde{x}_2'a_2\}<\widetilde{x}_3'a_3,$ $ x^{\diamond '}_1 a_1<\min\{x^{\diamond '}_2 a_2,x^{\diamond '}_3 a_3\}$
with no restrictions on the relative ordering of $\widetilde{x}_1'a_1$, $\widetilde{x}_2'a_2$ and of $x^{\diamond '}_2 a_2$, $x^{\diamond '}_3 a_3$. Thus, if $\mathcal{A}$ has a nonempty relative interior in the normalized space and at least one of the displayed tie equalities is nontrivial on that space, the coarse set  $\mathcal{A}^*$  has strictly higher affine dimension than
the identified set $\mathcal{A}_0$ of the rank-order model.

\subsection{Illustrative designs for theoretical examples}\label{sec:illustrativedesigns}

We give illustrative designs for the first two examples here and for the multinomial example in the Appendix.

\textbf{Illustrative Design 1 (Cross-sectional binary choice under median condition).} 
Let the vector of covariates be $(1,X_1,X_2)$ and, thus,  have a fixed first component. Let $X=(1,X_1,X_2)' \in \{(1,0,1),(1,1,0)\}$.  Suppose $\PP(X=(1,0,1)')=q>1/2$ (without loss of generality). We take ${\alpha} \equiv (\alpha_{1},1,\alpha_{2})$.  
 Let $\mbox{Med}(u|X)=0.$ 
 
The indices at the two support points are $a_1+a_2$ and $a_1+1$.
Their signs determine the corresponding choice probabilities
relative to one half. Using  the $\sign$ function such that 
$\sign(0)=0$,\footnote{Note that the function $\text{sgn}$ used by us above has the property $\text{sgn}(0)=1$ which aligns with its properties in  \citeasnoun{manski1985,manski-jasa}}  the identified set is 
$$
\mcA_0=\left\{
(a_1,1,a_2)\,:\,
\sign(a_1+a_2)=
\sign(\alpha_{1}+\alpha_{2}),\,
\sign(a_1+1)=\sign(\alpha_{1}+1)
\right\} \cap \mcA.
$$
Suppose  $\mcA$ contains \((-1,1,1)\) in its interior in the normalized space.  Then we have the following four cases for the identified set
$\mcA_0:$
\begin{enumerate}
\item[(1A)] If $\alpha=(-1,1,1)$, then $\mcA_0=\{(-1,1,1)\}$ (\textit{the only point identification case}). 
\item[(1B)] If $\alpha_{1}=-1$,  $\alpha_{2}\neq 1$, then 
$\mcA_0= \left\{(-1, 1,a_{2}): \sign(-1+a_{2})=\sign(-1+\alpha_{2})\right\} \cap \mcA.$

$\mcA_0$ is  fully contained in one-dimensional hyperplane $\alpha_{1}=-1$ (\textit{partial  identification case; identified set has affine dimension 1}).

\item[(1C)] If $\alpha_{1}\neq -1$,  $\alpha_{1}+\alpha_{2}=0$, then 
$\mcA_0= \left\{(a_{1}, 1,-a_{1}): \sign(a_{1}+1)=\sign(\alpha_{1}+1)\}\right\} \cap \mcA.$

$\mcA_0$ is  fully contained in the one-dimensional hyperplane $a_{1}+a_{2}=0$ (\textit{partial  identification case; identified set has affine dimension 1}).

\item[(1D)] If $\alpha_{1}\neq -1$ and $\alpha_{1}+\alpha_{2}\neq 0$, then only the general expression provided above applies
(\textit{partial  identification case; identified set has a non-empty 2-dimensional interior in the normalized parameter space  and, hence, has affine dimension 2}).
\end{enumerate}
 Partial identification can therefore coexist with
conditional choice probabilities equal to $1/2$ and
an identified set of lower affine dimension. 

\textbf{Illustrative design 2 (Two-period static binary choice panel data model under identical distribution of errors).} Suppose the support of each ${X}_{it}$ is 
$\widetilde{\mathcal{X}}=\{(0,0)', (0,1)', (1,0)', (1,1)'\}$, $t=1,2$, and the joint support of $(X_{i1},X_{i2})'$ is the  Cartesian product  $\widetilde{\mathcal{X}} \times \widetilde{\mathcal{X}}$.  
Suppose the second parameter is normalized to 1, thus  leaving the unknown parameter in the normalized space to take the form $(\alpha_1,1)$. As $\alpha_1$ varies, the sign of an informative first difference
index $(x_2-x_1)'(\alpha_1,1)$ can change only when that index is
zero.  In this design, the resulting threshold values are
$\alpha_1=-1,0,1$.

At $\alpha_1=1$, the configurations $(x_{i1},x_{i2})=((0,1)',(1,0)')$ and its reversal have equal conditional choice probabilities across periods. Their index restriction gives $\mathcal A_0=\{(1,1)\}$.
\footnote{The cases $\alpha_1\in\{0,-1\}$ are analogous.}
If $\alpha_1 \notin \{-1,0,1\}$ in the DGP, then the identified set will be an intersection of an open interval in the normalized space  with $\mathcal{A}$.

\section{Properties of estimators based on coarse structures} 
\label{sec:main_estimators_coarse} 
\setcounter{equation}{0}

The preceding examples show how coarse characterizations
discard equality restrictions that can hold at support points
of positive probability. At these points, empirical regime
classifications can remain random even asymptotically.
This section characterizes the resulting limits of estimators
based on coarse structures, focusing the contributions on cases in which
$d_H(\overline{\mathcal A^*},\overline{\mathcal A_0})>0$.

\subsection{Discrete setup: finite support and regime classification}
\label{sec:discrete_setup} 

Before we proceed to our formal results, we describe a discrete setting underlying them.

Throughout this section, $\{(Y_i,X_i)\}_{i=1}^n$ is an i.i.d.\ sample, $X$ has finite support $\mathcal X=\{x_1,\ldots,x_S\}$. Throughout the analysis below, we omit support points at which \(a\mapsto\psi(x,a)\) is constant on the normalized parameter space \(\mathcal A\), and continue to denote the remaining support by \(\mathcal X\). Such points impose no parameter restrictions under the maintained model and contribute only  terms to the sample objective that are independent of parameters.

We assume that \(\phi(\widehat F(\cdot| x))\) consistently estimates \(\phi(F(\cdot| x))\) for every \(x\in\mathcal X\), and that \(a\mapsto\psi(x,a)\) is continuous on \(\mathcal A\) for every \(x\in\mathcal X\).\footnote{For an unobserved support cell, we define
$\widehat F(\cdot| x)$ as a fixed admissible
distribution. This convention does not affect the
asymptotic results, since the support is finite and
every support point has positive probability.}

For each $x\in\mathcal X$, define the finite sample feasible regime set
$$
\mathcal M_{fs}(x)
:=\Bigl\{m\in\mathcal M:
\PP\bigl(\phi(\widehat F(\cdot\mid x))\in C_m\bigr)>0
\text{ for some sample size }n\geq1
\Bigr\}.
$$ 
the asymptotically relevant regime set 
$$
\mathcal M_{as}(x)
:=\Bigl\{m\in\mathcal M:\ \limsup_{n\to\infty}
\PP\bigl(\phi(\widehat F(Y| X=x))\in C_m\bigr)>0\Bigr\}, 
$$ 
and denote the set of persistently ambiguous covariate points as 
$$\mathcal X_0:=\{x\in\mathcal X:\ |\mathcal M_{as}(x)|\geq 2\}=\{x_1,\ldots,x_T\}.$$
For $x\notin\mathcal X_0$ we will impose asymptotic uniqueness of regime
classification: 
\begin{align} 
\label{eq:nonambig1} \mathcal M_{as}(x) & =\{m(x)\}, \quad x \in \mathcal{X}\setminus \mathcal{X}_0,
\end{align}

Under the conditions established below, persistent regime ambiguity for covariates in $\mathcal{X}_0$ is what  leads to distinct limiting parameter regions and inconsistency. 


\paragraph{Feasible systems and their solution sets.}
For each regime of indices $m=(m_1,\ldots,m_T)$, $m_t\in\mathcal M_{as}(x_t)$, $t=1,\ldots,T$, we 
define the associated  constraint system:
\begin{align}
\psi(x,a) &\in D^*_{m(x)} \; \text{ for } \; x\in\mathcal X\setminus\mathcal X_0,\label{eq:sys_outside}\\
\psi(x_t,a) &\in D^*_{m_t}, \quad\; \text{ for } \;  t=1,\ldots,T.\label{eq:sys_inside}
\end{align}
Let $A(m)$ be the solution set to \eqref{eq:sys_outside}--\eqref{eq:sys_inside} in $\mathcal A$. Some $A(m)$ may be empty. Let $A_1,\ldots,A_L$ denote the distinct nonempty sets among
$\{A(m)\}$. Thus, each region $A_{\ell}$  corresponds to some resolution
of the boundary regime induced by sampling noise at a finite number of ambiguous 
covariate points in $\mathcal{X}_0$.

For each $\ell= 1,\ldots,L$, and each $t\in\{1,\ldots,T\}$, fix  an index
$m_{\ell,t}\in\mathcal M_{as}(x_t)$ such that $A_\ell=A(m_{\ell,1},\ldots,m_{\ell,T}).\footnote{Hypothetically, at this stage we can agree that if multiple index vectors potentially generate the same $A_\ell$, then we pick any one as our results do not depend 
on which representative is chosen. As we will see later, such cases will be eliminated by our conditions in the general theorems.}$ For each \(\ell\), choose a subset \(\mathcal X_{0,\ell}\subseteq\mathcal X_0\) such that the constraints in \eqref{eq:sys_outside}, together with those in \eqref{eq:sys_inside} indexed by \(\mathcal X_{0,\ell}\), define exactly \(A_\ell\), and removing any retained constraint strictly enlarges the solution set. Such a subset exists because \(\mathcal X_0\) is finite. If several subsets fit this requirement, fix any one of them.

First, note that for $\alpha$ in the DGP that results in $\mathcal{X}_0=\varnothing$, we have  then $L=1$ and $A_1=\mathcal{A}^*$ as only (\ref{eq:sys_outside}) apply. Second, if $\alpha$ in the DGP results in $\mathcal{X}_0 \neq \varnothing$, then little can be said about the nature of $A_1$, \ldots, $A_L$ without further describing the properties of  $D^*_{m_t}$ across different $m_t \in \mathcal{M}_{as}(x_t)$. For our theoretical examples  we will show that in such cases $L \geq 2$ given the parameter space $\mathcal{A}$ is large enough and has non-empty relative interior in the normalized space. 

These considerations  motivate us to  impose the following Assumption \ref{assn:assnexitense}, which is a weak requirement of at least one feasible profile of asymptotically relevant regimes. 
 This holds under the conditions used below for binary choice. In the multinomial model it excludes structural ties that cannot be resolved by parameter variation.
The theorems below will impose additional conditions.

\begin{assumption}\label{assn:assnexitense} We will assume that $L\geq 1$ when $\mathcal{X}_0 \neq \varnothing$.
\end{assumption}

\subsection{Estimators and general theorems}
\label{sec:general_theorems}

Given our discrete setup we consider a population objective in the form of \eqref{ex_obective_population} and its sample version \eqref{ex_obective_sample_discrete} with a score function $s$ that satisfies \eqref{eq:pointwise_separation1}-\eqref{eq:pointwise_separation2}. 

As discussed earlier, the Argmax set for \eqref{ex_obective_population} coincides with $\mathcal{A}^*$. We will denote the set of maximizers   of the sample objective function  \eqref{ex_obective_sample_discrete} as $\widehat{A}$: 
$$\widehat{A} := \text{Argmax}_{a \in \mathcal{A}}\sum_{x\in\mathcal X_n} w(x)
s\!\big(\phi(\widehat{F}(Y\mid X=x)),\psi(x,a)\big) \, \widehat \PP(X=x).$$
We assume that the sample maximum is attained and that
$\overline{\widehat A}$ is a measurable random closed set. We present our results related to  the properties of $\widehat{A}$ in five different theorems. 

Our first general Theorem~\ref{th:general1} shows,  among other things,  that in our discrete setup the estimator $\widehat{A}$ under a coarse structure does not converge to a single deterministic set if $L>1$. Instead, it converges in distribution to a {random choice} among 
deterministic regions $A_1,\ldots,A_L$.

\begin{theorem}
\label{th:general1}
Assume the discrete setup in Section~\ref{sec:discrete_setup} and let
Assumption~\ref{assn:assnexitense} hold. Consider the following
conditions.

\begin{enumerate}

\item[(C1)]
For each $t=1,\ldots,T$, we have 
$D^*_{m(x_t)}
\supseteq
D^*_{m_t}$ $\forall\,m_t\in\mathcal M_{as}(x_t)$.

\item[(C2)]
For each $t=1,\ldots,T$, we have 
$D^*_{m_t}\cap D^*_{m'_t}
=
\varnothing$
$\forall\,m_t\neq m'_t$, 
$m_t,m'_t\in\mathcal M_{as}(x_t)$.

\item[(C3)]
For every $x\in\mathcal X$, the scores associated with the
asymptotically relevant coarse regions are asymptotically
indistinguishable. 
\begin{equation}
\label{C3part1a}
\sup_{\substack{
u,v\in
\cup_{m\in\mathcal M_{as}(x)}D_m^*
}}
\left|
s\!\left(\phi(\widehat F(Y| X=x)),u\right)
-
s\!\left(\phi(\widehat F(Y|X=x)),v\right)
\right|
=o_p(1).
\end{equation}

Moreover, whenever
$(\bigcup_{m\in\mathcal M_{as}(x)}D_m^*)^c 
\neq  \varnothing$,
there exists $\delta_x>0$ such that
{\small\begin{equation}
\label{C3part1b}
\PP\left(
\inf_{\substack{
u\in\cup_{m\in\mathcal M_{as}(x)}D_m^*\\
u'\notin\cup_{m\in\mathcal M_{as}(x)}D_m^*
}}
\left[
s\!\left(\phi(\widehat F(Y| X=x)),u\right)
-
s\!\left(\phi(\widehat F(Y| X=x)),u'\right)
\right]
>\delta_x
\right)
\to1.
\end{equation}}

For every $\ell=1,\ldots,L$,
{\small\begin{equation}
\label{C3part2}
\liminf_{n\to\infty}
\PP\left(
\left[
\bigcap_{x\in\mathcal X\setminus\mathcal X_0}
\left\{
\phi(\widehat F(Y| X=x))
\in C_{m(x)}
\right\}
\right]
\cap
\left[
\bigcap_{t=1}^T
\left\{
\phi(\widehat F(Y| X=x_t))
\in C_{m_{\ell,t}}
\right\}
\right]
\right)
>0.
\end{equation}}

Additionally, impose the following regionwise constancy condition:\footnote{Condition~\eqref{eq:region_constancy} is distinct from
\eqref{eq:pointwise_separation1}. The latter concerns constancy on
the pointwise maximizing coarse set $D_m^*$ when the observable value
$p$ belongs to the matching regime $C_m$. In contrast,
\eqref{eq:region_constancy} requires the score contribution to be
constant over each candidate limiting region $A_\ell$ for every
admissible value of $p$, including values belonging to a different
regime.}
\begin{equation}
\label{eq:region_constancy}
s\!\left(p,\psi(x,a)\right)
=
s\!\left(p,\psi(x,b)\right)
\quad
\forall\,x\in\mathcal X,\; 
\forall\,\ell=1,\ldots,L,\; 
\forall\,a,b\in A_\ell,\; 
\forall\,p\in\cup_{m\in\mathcal M}C_m,
\end{equation}
and for each $\ell=1,\ldots,L$, choose any
$a_\ell\in A_\ell$ and define
$\widehat Q_\ell
:=
\widehat Q(a_\ell)$, which, by 
 \eqref{eq:region_constancy}, does not depend on
the particular choice of $a_\ell\in A_\ell$.

If $L\geq2$, suppose that there exists a sequence $r_n\to\infty$
and a random vector $(V_2,\ldots,V_L)'$ such that
\begin{equation}
\label{C3part3}
r_n
(
\widehat Q_2-\widehat Q_1, \ldots, 
\widehat Q_L-\widehat Q_1
)^{\top}
\stackrel{d}{\longrightarrow}
(
V_2, \ldots, 
V_L)^{\top}.
\end{equation}
For notational convenience, set $V_1=0$. This simply takes $A_1$
as the reference region and has no effect on the identity of the
maximizing region. Assume further that
\begin{equation}
\label{C3part3unique}
\PP\left(
\left|
\arg\max_{\ell\in\{1,\ldots,L\}}V_\ell
\right|=1
\right)=1.
\end{equation}
When $L=1$, conditions
\eqref{C3part3}--\eqref{C3part3unique} hold automatically.

\end{enumerate}

Then:

\begin{enumerate}

\item[(a)] Under (C1),
$A_\ell\subseteq\mathcal A^*$, $\ell=1,\ldots,L$.

\item[(b)]
Under (C2), $A_\ell\cap A_{\ell'}=\varnothing$ 
for all $\ell\neq\ell'$.

\item[(c)]
Under \eqref{C3part1a}--\eqref{C3part1b} and
\eqref{eq:region_constancy}, $\PP\left(
\widehat A\in
\left\{
\bigcup_{\ell\in\mathcal I}A_\ell:
\varnothing\neq\mathcal I
\subseteq\{1,\ldots,L\}
\right\}
\right)\to1$.

Together with \eqref{C3part2}-\eqref{eq:region_constancy}, these conditions additionally imply
that, for every $\ell=1,\ldots,L$,
\[
\liminf_{n\to\infty}
\PP(A_\ell\subseteq\widehat A)>0.
\]

Finally, under
\eqref{C3part3}--\eqref{C3part3unique}, 
$\PP(\widehat A=A_\ell)\to \PP(\arg\max_{j\in\{1,\ldots,L\}}V_j = \ell)$, $\ell=1,\ldots,L$. 

Overall, the implications of all \eqref{C3part1a}-\eqref{C3part3unique} can be summarized as 
\begin{equation} 
\label{convdistrgeneral} \widehat{A} \stackrel{d}{\to} \sum_{\ell=1}^{L} d_{\ell} A_{\ell},
\end{equation}
where $d_{\ell}$ are binary $0/1$ variables such that $\sum_{\ell=1}^{L} d_{\ell} =1$, $d_{\ell} d_{\ell'}=0$ for $\ell \neq \ell'$. 

\end{enumerate}
\end{theorem}

Conditions (C1) and (C2) are straightforward. The different parts of (C3) characterize the asymptotic behavior of
the sample maximizer. Conditions \eqref{C3part1a} and
\eqref{C3part1b} play complementary roles. The first says that, at
each support point $x$, the score differences among the asymptotically
relevant resolutions $D_m^*$ for $m \in \mathcal M_{as}(x)$ vanish asymptotically. The second guarantees
that moving outside the union of those resolutions entails a
nonvanishing score loss. Condition \eqref{C3part2} then ensures that none of these regions
disappears asymptotically. It imposes a joint requirement on the
sample regimes across the support. Condition \eqref{eq:region_constancy}
 ensures that the sample objective can be assigned a single value to each candidate region, irrespective of which empirical regimes are realized. Finally, \eqref{C3part3} describes the asymptotic competition among
the finitely many regions. After an appropriate normalization, the
vector of their sample objective differences has a joint limiting
distribution. Condition \eqref{C3part3unique} requires this limiting
competition to have an almost surely unique winner. Together these
conditions imply that the probability with which each $A_\ell$ is
selected converges, and therefore give the distributional limit in
\eqref{convdistrgeneral}.

The limit in \eqref{convdistrgeneral} is described by selection probabilities over finitely many deterministic regions. We later introduce random set convergence after taking closures, as required by that notion (sets $A_1, \ldots, A_L$ are not necessarily closed).

The next theorem provides conditions under which the identified set $\mathcal A_0$ lies on the boundary of each $A_\ell$, reflecting that the information lost under the coarse characterization consists of equality-type restrictions (e.g.\ $x'\alpha=0$), which generically define lower-dimensional manifolds.

\begin{theorem} 
\label{th:general2}
Assume a discrete  setup in Section \ref{sec:discrete_setup}. Let  Assumption  \ref{assn:assnexitense} hold. Consider the following three conditions: 
\begin{enumerate}
\item[(C4)] 
For every $\ell$ and every $a\in\mathcal A$ such that
\[
\psi(x,a)\in \overline D^*_{m(x)} \quad \forall x\notin\mathcal X_0,
\qquad
\psi(x_t,a)\in \overline D^*_{m_{\ell,t}} \quad \forall x_t\in\mathcal X_{0,\ell},
\]
and for every $\varepsilon>0$, there exists $a^+_\varepsilon\in\mathcal A$
with $\|a^+_\varepsilon-a\|<\varepsilon$ such that
\[
\psi(x,a^+_\varepsilon)\in D^*_{m(x)} \quad \forall x\notin\mathcal X_0,
\qquad
\psi(x_t,a^+_\varepsilon)\in D^*_{m_{\ell,t}} \quad \forall x_t\in\mathcal X_{0,\ell}.
\]

\item[(C5)]  

For every $\alpha\in\mathcal A_0$, every $\ell$, for every $\varepsilon>0$ there exist
$a^+_\varepsilon,a^-_\varepsilon\in\mathcal A$ with
$\|a^\pm_\varepsilon-\alpha\|<\varepsilon$ such that
\[
\psi(x,a^\pm_\varepsilon)\in D^*_{m(x)} \quad \forall x\notin\mathcal X_0; 
\]
\[
\psi(x_s,a^+_\varepsilon) \in D^*_{m_{\ell,s}}, \quad  \forall x_s\in\mathcal X_{0,\ell},  
\qquad
\psi(x_s,a^-_\varepsilon)\notin D^*_{m_{\ell,s}},   \text{ for some  } x_s\in\mathcal X_{0,\ell}.
\]

\item[(C6)]  
For every $\ell=1,\ldots,L$ and every
$x_t\in\mathcal X_{0,\ell}$, we have
$D_{m(x_t)} \subseteq \partial D^*_{m_{\ell,t}} $.

\end{enumerate}

Then: 
\begin{enumerate}
\item[(a)]   
Under (C4), for each $\ell=1,\ldots,L$.
\[
\overline A_\ell
=
\Bigl\{a\in\mathcal A:\ 
\psi(x,a)\in \overline D^*_{m(x)} \ \forall x\notin\mathcal X_0,\ 
\psi(x_t,a)\in \overline D^*_{m_{\ell,t}} \ \forall x_t\in\mathcal X_{0,\ell}
\Bigr\}.
\]

\item[(b)] 
When  $\mathcal{X}_0 \neq \varnothing$, 
 (C5) and (C6) imply that $\mathcal A_0 \subseteq \partial A_\ell$
for all $\ell=1,\ldots,L.$
\end{enumerate}
\end{theorem}
Condition (C4) allows every point satisfying the closed
constraints for $A_\ell$ to be approximated by points in
$A_\ell$. Condition (C5) gives points inside and outside
$A_\ell$ arbitrarily close to each identified parameter.
Condition (C6) places the exhaustive restrictions at ambiguous support points on the boundaries of the competing coarse sets.

Theorem \ref{th:general3} below establishes under which conditions the closure of the coarse identified set
$\mathcal A^*$ is the union of the closures of  regions $A_{\ell}$, which means that coarsening
expands the identified set by filling in neighborhoods around the true boundary
constraints.

\begin{theorem} 
\label{th:general3} 
Assume a discrete  setup in Section \ref{sec:discrete_setup} and let  Assumption  \ref{assn:assnexitense} hold.  Define the analogue of the coarse set but with closures  
$\mathcal A^{*,\bullet}
:=
\{a\in\mathcal A:\ \psi(x,a)\in \overline D^*_{m(x)}\ \forall x\in\mathcal X\}$, and  consider the following conditions: 
\begin{enumerate}
\item[(C7)]  For every $a\in\mathcal A^{*,\bullet}$ and every $\varepsilon>0$, there exists $a_\varepsilon\in\mathcal A$ such that $\|a_\varepsilon-a\|<\varepsilon$ and  $\psi(x,a_\varepsilon)\in D^*_{m(x)}$ $\forall x\in\mathcal X$. 

\item[(C8)]
For every $a\in\mathcal A^{*,\bullet}$, there exists a profile $\mathbf m=(m_1,\ldots,m_T)$, $m_t\in\mathcal M_{as}(x_t)$, such that the corresponding system
\eqref{eq:sys_outside}--\eqref{eq:sys_inside} has a nonempty solution
set and
$$
\psi(x,a)\in\overline D^*_{m(x)}
\quad
\forall\,x\in\mathcal X\setminus\mathcal X_0,
\quad \text{ and } \quad 
\psi(x_t,a)\in\overline D^*_{m_t}
\quad
t=1,\ldots,T.
$$

\end{enumerate}
Then: 
\begin{enumerate}
\item[(a)] Under (C7), $\overline{\mathcal A^*}=\mathcal A^{*,\bullet}$.

\item[(b)] (C1), (C2), (C4)  and (C8) imply that $\overline{\mathcal A^*}= \bigcup_{\ell=1}^L \overline A_\ell$.
\end{enumerate}
\end{theorem}

Condition (C7) is an analogue of (C4) but for the coarse set $\mathcal{A}^*$ and  rules out isolated 
boundary points of $\mathcal{A}^*$ that cannot be 
approximated from the interior. Condition (C8) is a joint feasibility requirement. It says that every
point satisfying the closed coarse restrictions is also compatible
with the closed restrictions associated with at least one profile of
asymptotically relevant regimes whose corresponding original constraint system is
nonempty. Thus, no point in the closed coarse feasible set is lost
when we pass to the union of the closures of the nonempty candidate
regions.

The  next theorem gives us a result on how $\overline{\mathcal{A}}_0$ can be recovered from a sub-collection of $\{A_{\ell}\}_{\ell=1}^L$. It will also imply a sample mechanism through which this can be attained. To formulate this theorem, let us bring back the fact that each $A_{\ell}$ is obtained through a choice of indices $\mathbf{m}=(m_1,\ldots,m_T),\qquad m_t\in\mathcal M_{as}(x_t)$, \quad $t=1, \ldots, T$. For convenience we may represent such profiles as $\mathbf{m} =(m_t,m_{-t})$ when dealing with a particular $t=1, \ldots, T$.  

For a given $A_{\ell}$ let $\mathcal{F}_{\ell}$ collect all profiles $\mathbf{m}$ that result in $A_{\ell}$. Since $A_{\ell}$ is nonempty, then $\mathcal{F}_{\ell} \neq \emptyset$.

\begin{theorem}
\label{th:general4} Assume a discrete  setup in Section \ref{sec:discrete_setup} with $\mathcal{X}_0 \neq \varnothing$ and let  Assumption  \ref{assn:assnexitense} hold. 
Suppose all conditions of Theorems ~\ref{th:general1}-~\ref{th:general3} hold. Consider the following conditions: 
\begin{itemize}
\item[(C9)] $\overline{D^*_{m(x)}}=\overline{D_{m(x)}}$ for each  $x\in\mathcal X\setminus\mathcal X_0$, and 
 for every $a\in\mathcal A$ satisfying
$\psi(x,a)\in\overline D_{m(x)}$ $\forall x\in\mathcal X$, and every $\varepsilon>0$, there exists
$a_\varepsilon\in\mathcal A_0$ such that $\|a_\varepsilon-a\|<\varepsilon$.

Moreover, each $\mathcal{F}_{\ell}$ is a singleton. 
For each $x_t$ and each subset $\mathcal{J}_t \subseteq \mathcal{M}_{as}(x_t)$ such that $|\mathcal{J}_t| \geq \lfloor \frac{|\mathcal{M}_{as}(x_t)|}{2}\rfloor +1$, it holds that 
$\overline D_{m(x_t)}
=
\bigcap_{m\in\mathcal J_t} \overline D_m^*.$

\item[(C10)] For each $\ell$ and each $t$ in the decomposition  $\mathbf{m}=(m_t,m_{-t}) \in \mathcal{F}_{\ell}$, all $(|\mathcal{M}_{as}(x_t)|-1)$ profiles $\mathbf{m}=(\widetilde{m}_t,m_{-t})$ with $\widetilde{m}_t \in \mathcal{M}_{as}(x_t)$, $\widetilde{m}_t \neq m_t$,  belong to other $\mathcal{F}_{\ell'}$. 

\item[(C11)] The collection of regions $A_1,\ldots,A_L$ can be arranged into pairs, possibly leaving some regions unpaired, such that for every pair $(A_\ell,A_{\ell'})$,  $\overline A_\ell\cap\overline A_{\ell'} = \overline{\mathcal A}_0$.  Any region that remains unpaired satisfies $ \overline A_\ell = \overline{\mathcal A}_0$. 

\end{itemize}

If (C9) and (C10) hold, or if (C11) holds, then the intersection of
any $\lfloor L/2\rfloor+1$ sets $\overline A_\ell$ equals the closure
of the identified set:
\[
\bigcap_{\ell\in S}\overline A_\ell
=
\overline{\mathcal A}_0
\qquad
\text{for all }
S\subseteq\{1,\ldots,L\},
\quad
|S|\geq\left\lfloor\frac L2\right\rfloor+1.
\]

\end{theorem}

Theorem \ref{th:general4} tells us that  the identified set can be recovered (up to the closure) from the intersection of any $[L/2]+1$ sets $\overline{A}_{\ell}$. This will  be a major motivation for us to consider in the discrete setup of Section \ref{sec:quantile}  a Random Set Quantile estimator in a sample.  

The first part of condition (C9) is analogous to (C4) and guarantees that $\overline{\mathcal{A}}_0$ can be represented through closures $\overline{D}_{m(x)}$. The second part of (C9) requires each family $\mathcal{F}_\ell$ to be a 
singleton, meaning that each nonempty region $A_\ell$ 
is generated by exactly one profile 
$\mathbf{m}=(m_1,\ldots,m_T)$. 
The third part requires that the majority 
intersection of the closed coarse sets 
$\overline{D^*_m}$ recovers the closed exhaustive 
constraint set $\overline{D_{m(x_t)}}$ at each 
$x_t\in\mathcal{X}_0$.

Condition (C10) requires feasibility after changing any one component of a feasible regime profile. Repeating such changes
one coordinate at a time gives feasibility of every complete
profile. Conversely, full profile feasibility implies (C10). Condition (C11) requires enough symmetry in the collection of limiting regions that no point outside the identified set can be supported by a majority of them.

When we apply Theorem \ref{th:general4} to our theoretical examples, the two conditions play different roles. In the cross-sectional binary choice model, for instance, (C10) holds under linear independence of the ambiguous regressors, whereas (C11) continues to hold when they are linearly dependent.  More generally, however, (C10) and (C11) should be viewed as alternative sufficient conditions for the majority property.

We can summarize the results of general theorems in the following result. 

\begin{theorem}\label{th:general5}
    Assume a discrete  setup in Section \ref{sec:discrete_setup}.  
\begin{enumerate}
\item Suppose parameter value $\alpha$ in the data generating process yields  $\mathcal{X}_0 \neq \varnothing$. Let  Assumption  \ref{assn:assnexitense} hold, and let   $A_1$, \ldots, $A_L$, $L \geq 1$, be the deterministic nonempty distinct sets obtained as all possible distinct solutions to (\ref{eq:sys_outside})-(\ref{eq:sys_inside}). 

If (C1)-(C8) hold and, in addition, either (C9) and (C10) hold or (C11) holds, then 
$$\widehat{A}
\stackrel{d}{\rightarrow} d_1 A_1 + \ldots +d_L A_L,$$
where $\cup_{\ell=1}^L \overline{A}_{\ell} = \overline{\mcA^*}$. In addition, 
\begin{itemize}
\item[(a)] $d_1$, \ldots, $d_L$ are dummy variables such that 
$d_{\ell} d_m =0$ for  $\ell \neq m$ (mutually exclusive), and $d_1+\ldots+d_L=1$ (collectively exhaustive),  
and $\PP(d_{\ell}=1)>0$ for each $\ell$, $\sum_{\ell=1}^L \PP(d_{\ell}=1)=1$.  

\item[(b)] The boundary of each $A_{\ell}$ contains the identified set $\mcA_0$.  

\item[(c)] The intersection of any $[L/2]+1$ closed sets $\overline{A}_{\ell}$     coincides with the closure $\overline{\mcA}_{0}$. 

\end{itemize} 

Also, $\overline{\widehat{A}} \stackrel{W}{\rightarrow } \mathbf{A}(\alpha)$, where $\stackrel{W}{\rightarrow }$ is the weak convergence from the random set theory, and random set $\mathbf{A}(\alpha)$ has the following distribution:  $\PP\left(\mathbf A(\alpha)=K\right)=\sum_{\ell:\,\bar A_\ell=K}\PP(d_\ell=1)$, $K\in\{\bar A_1,\ldots,\bar A_L\}$. 

If, in addition, every $A_\ell$ is convex, their pairwise disjointness under (C2) ensures that their closures are distinct. In this case,\footnote{This special case applies to the linear inequality regions in our three theoretical examples.}
$$
\PP\left(\mathbf A(\alpha)=\bar A_\ell\right)
=\mathbb P(d_\ell=1),
\quad
\sum_{\ell=1}^{L} \PP\!\left(\mathbf A(\alpha)=\bar A_\ell\right)=1.
$$

\item Suppose parameter value $\alpha$ in the data generating process yields $\mathcal{X}_0 = \varnothing$. 
Then  
$\PP(\widehat{A}\neq \mathcal{A}^*) \to 0 \text{ as } n \to \infty,$ which implies   
$\widehat{A}
\stackrel{d}{\rightarrow} \mathcal{A}^*$, $\overline{\widehat{A}}
\stackrel{W}{\rightarrow} \overline{\mathcal{A}^*}$. 
\end{enumerate}
\end{theorem} 
In Theorem~\ref{th:general5}, any strict majority of the closed limiting regions recovers $\overline{\mathcal A}_0$ by intersection. At this stage this is purely a geometric property with probabilistic majority implications, which are established later separately in each theoretical example.  Case 2 of Theorem \ref{th:general5} follows directly from the discrete setup of Section \ref{sec:discrete_setup}, independently of Theorems \ref{th:general1}-\ref{th:general4}.

\subsection{Application of general theorems to our examples} 
\label{sec:generaltheorems_applications}

We now illustrate how our general theorems apply to the theoretical examples in Sections \ref{sec:theoreticalMSbinary}–\ref{sec:theoreticalmultinomial}. In each case, the same mechanism as in our general discrete setup arises, with discrete covariates generating boundary regimes that classical estimators ignore or weaken, resulting in multiple admissible regions.

\subsubsection{Application to Theoretical Example 1 (cross-sectional binary choice)}

We continue to consider the model discussed in Section \ref{sec:theoreticalMSbinary}, where we gave exhaustive and coarse characterizations. Since the coarse characterization of interest is the one induced by the maximum score estimator, it means that application of Theorems \ref{th:general1}-\ref{th:general4} will produce  the result for the asymptotic behavior of the maximum score estimator $\widehat{A}$ obtained as a maximizer of the sample objective function \eqref{ex_obective_sample_discrete}  with the score function \eqref{scorefunction_bc}. Recall that the coarse $\mathcal{A}^*$ in this case is the maximizer of the population objective \eqref{ex_obective_population} with the same score function. 

Suppose  $X$ is discrete. Let $\mathcal X=\{x_1,\ldots,x_S\}$  denote the finite effective support, following the convention in Section \ref{sec:discrete_setup} (so \(x=0\), if it has positive probability, is excluded as uninformative). Before we engage with the terminology  of Section \ref{sec:discrete_setup} for this model, let us formulate the following result about the joint and individual asymptotic behaviors of plug-in estimators $\widehat \PP(Y=1|{x})$.  

\begin{proposition}
\label{prop:bccross_probest1}  Suppose \(X\) has finite support, and let \(\mathcal X\) denote its effective support as defined in Section \ref{sec:discrete_setup}. Suppose we have a random sample $\{(x_i,y_i)\}_{i=1}^n$. Then, for a given $x \in \mathcal{X}$:  
\begin{itemize}
\item[(a)] If $\PP(Y=1|{x})-\gamma \in C_j$ for  $j \in \{1,2\}$, then $\PP(\widehat \PP(Y=1|{x}) -\gamma \in C_j) \to 1$ as $n \to \infty$.
\item[(b)] If $\PP(Y=1|{x})-\gamma \in C_3$, then  $\PP(\widehat \PP(Y=1|{x}) -\gamma \in C_m) \to \frac{1}{2}$ for $m \in \{1,2\}$. 

\end{itemize}
\end{proposition}

Parts (a) and (b) in Proposition \ref{prop:bccross_probest1}  imply that in the terminology and description of Section \ref{sec:discrete_setup} the subset $\mathcal{X}_0 \subseteq \mathcal{X}$ of ambiguous points with multiple possible asymptotic regimes will be composed of all $x \in \mathcal{X}$ for which $m(x)=3$: 
\begin{equation}\label{bc1_X0def}\mathcal X_0 := \{x\in\mathcal X:\ \PP(Y=1|X=x)=\gamma\},
\end{equation}
For each $x_t \in \mathcal{X}_0$ we have $\mathcal M_{as}(x_t)=\{1,2\}$ whereas  $\mathcal M_{as}(x)=\{m(x)\}$ for each $x \in \mathcal{X}\setminus \mathcal{X}_0$,

If $\mathcal X_0=\varnothing$, then things are pretty straightforward as all asymptotic regimes are aligned with the true regimes 1 or 2. $\mathcal A^*$ in this case differs from $\mathcal A_0$ at most at the boundary (due to the weak inequality in $D_1^*$), we have $L=1$ and  $A_1=\mathcal A^*$.

The challenging case is that of $\mathcal X_0\neq\varnothing$, in which identification uses boundary equalities $x'a=0$ 
that are weakened or dropped by the coarse structure. For this case, we suppose  that the true $\alpha$ lies
in the interior of $\mathcal A$ in the normalized parameter space. This guarantees that the set $\mathcal{A}_{out} = \{a \in \mathcal{A}: x'a \in D_{m(x)} \text{ for all } x \in \mathcal{X}\setminus \mathcal{X}_0\}$ described by only strict inequalities in the exhaustive characterization  has a nonempty relative interior in the $(k-1)$-dimensional normalized space. This feature of $\mathcal{A}_{out}$ is important in verifying conditions of Theorem \ref{th:general4} for this model. The coarse set $\mathcal{A}^*$ only differs from $\mathcal{A}_{out}$ at the boundary as $D_1^*\setminus D_1 =\{0\}$.  Since points in $\mathcal X_0$ correspond to nontrivial equality restrictions on the normalized parameter space, this condition on $\mathcal A$ guarantees that, when $\mathcal X_0\neq\varnothing$, the affine dimension of $\mathcal A_0$ is strictly smaller than that of $\mathcal A^*$. Those points in $\mathcal{A}_0$ that are not at the boundary at $\mathcal{A}$ are in the relative interior of $\mathcal{A}^*$ in the $(k-1)$-dimensional normalized space.

The mechanics of constructing sets $A_1,..,A_L$ when $\mathcal X_0\neq\varnothing$ is straightforward. E.g., if $|\mathcal{X}_0|=1$, then our condition on the nonempty relative interior of the parameters space guarantees that we can construct two distinct regions $A_1$ and $A_2$, defined by replacing the
equality constraint $x_{t}'a=0$ with one of the two inequalities. When $|\mathcal{X}_0|>1$, then analogously  replace each equality $x_t'a=0$ with either $x_t'a\geq0$ or $x_t'a<0$, enumerate all resulting sign profiles, and retain their nonempty solution sets, including lower-dimensional regions.

To sum up, a nonempty $A_\ell$ in the construction of Section ~\ref{sec:discrete_setup} applied to this model, satisfies  
$x'a\geq 0$  if $m(x)=1$,  $x'a<0$  if $ m(x)=2$, and  for ambiguous points $x_t\in\mathcal X_{0,\ell}\subseteq\mathcal X_0$, imposes one of the two coarse constraints:
$x_t'a\ge 0$  if $m_{\ell,t}=1$, or 
$x_t'a<0$  if $m_{\ell,t}=2$,
where, recall, $m_{\ell,t}$ denotes  the relevant regime from $\mathcal{M}_{as}(x_t)$ that participates in the construction of $A_\ell$. 

  Proposition \ref{prop:bccross_probest2} below helps us to verify (C3). 

\begin{proposition}
\label{prop:bccross_probest2}
Suppose \(X\) has finite support, and let \(\mathcal X\) denote its effective support as defined in Section \ref{sec:discrete_setup}. Suppose $\{(X_i,Y_i)\}_{i=1}^n$ is a random sample. Let
$\mathcal X_0=\{x_1,\ldots,x_T\}$, where $\mathcal X_0$ is defined
in \eqref{bc1_X0def}. Then

\begin{itemize}
\item[(a)] $\sqrt n
\left(
\widehat \PP(Y=1|x_t)-\gamma
\right)_{t=1}^T
\stackrel{d}{\longrightarrow}
\mathcal N(0,\Sigma)$,
where
$\Sigma
=
\operatorname{diag}\left(
\frac{\gamma(1-\gamma)}{\PP(X=x_1)},\ldots,
\frac{\gamma(1-\gamma)}{\PP(X=x_T)}
\right)$.  In particular, $\Sigma$ is positive definite.

\item[(b)]
For every $(m_1,\ldots,m_T)\in\{1,2\}^T$,
\[
\PP\left(
\left[
\cap_{x\in\mathcal X\setminus\mathcal X_0}
\{\widehat \PP(Y=1| x)-\gamma\in C_{m(x)}\}
\right]
\cap
\left[
\cap_{t=1}^T
\{\widehat \PP(Y=1| x_t)-\gamma\in C_{m_t}\}
\right]
\right)
\to {1}/{2^T}.
\]

\end{itemize}
\end{proposition}

Verification of conditions (C1)--(C9), followed by either (C10) or
(C11), is given in the Appendix. (C10) holds when the  index functionals for points in $\mathcal{X}_0$  are linearly independent.  The geometric
pairing condition (C11), by contrast, continues to hold under
arbitrary linear dependence. 

The following result specializes the general results above to the
cross-sectional binary choice model and further characterizes the
limiting membership probabilities of the regions selected by the
maximum score estimator.

\begin{theorem}
\label{th:MSgeneral2}
Consider the cross-sectional binary choice model under the assumptions of Section \ref{sec:theoreticalMSbinary} and the discrete setup of Section~\ref{sec:discrete_setup}. Suppose that the true parameter $\alpha$ lies in the relative interior of the normalized parameter space. When \(\mathcal X_0\neq\varnothing\), conditions (C1)–(C9) and (C11) hold for this model. When \(\mathcal X_0=\varnothing\), the conclusion in part 2 follows directly. 

\begin{enumerate}

\item Suppose that the parameter value $\alpha$ in the DGP yields $\mathcal X_0\neq\varnothing$. Let $A_1,\ldots,A_L$, $L \geq 1$, denote the distinct nonempty deterministic regions constructed as described above. Then the maximum score estimator $\widehat A$ satisfies part~1 of Theorem~\ref{th:general5}.

Let
$\pi_\ell:=\PP(d_\ell=1),$ $\ell=1,\ldots,L$,
denote the probabilities of these regions in the limiting distribution, and define
\[
p_\infty(u)
:=
\sum_{\ell=1}^L
\pi_\ell\mathbf 1\{u\in\overline A_\ell\}.
\]

\begin{itemize}

\item[(a)] The limiting membership probability satisfies
\[
p_\infty(u)=1
\qquad
\text{for }u\in\overline{\mathcal A}_0, \quad 
p_\infty(u)\leq\frac12
\qquad
\text{for }u\notin\overline{\mathcal A}_0.
\]

\item[(b)] If, in addition, condition (C10) holds, then $L=2^T$ and the limiting regions are equiprobable: $\pi_\ell=\frac1L$, $\ell=1,\ldots,L$.

\end{itemize}

\item Suppose that the parameter value $\alpha$ in the DGP yields
$\mathcal X_0=\varnothing$. Then
$\overline{\mathcal A}_0= \overline{\mathcal A}^*$, $L=1$ with $A_1=\mathcal A^*$, and $\widehat A
\stackrel{d}{\longrightarrow}
\mathcal A^*$, $\overline{\widehat A}
\stackrel{W}{\longrightarrow}
\overline{\mathcal A}_0$.

\end{enumerate}
\end{theorem}

The second statement follows directly from
Proposition~\ref{prop:bccross_probest1}(a). When
$\mathcal X_0=\varnothing$, every support point has a unique population
regime, so all sample regimes coincide with their population counterparts
with probability approaching one. Hence
$\PP(\widehat A=\mathcal A^*)\rightarrow1$. Moreover, in this case the exhaustive and coarse characterizations can
differ only at boundary points, and, therefore, 
$\overline{\mathcal A}_0=\overline{\mathcal A^*}$.

\noindent \textbf{Illustrative Design 1 (continued)} To give a concrete application of  our theorems, consider Illustrative Design 1. In line with our discussion previously, we continue to assume that the parameter set $\mathcal{A}$ is a compact convex subset of the normalized plane, containing the specified true parameter in its relative interior.   Denote the points in the covariate support as $x^a=(1,0,1)'$, $x^b=(1,1,0)'$. 

(1A): $\alpha_{1}+\alpha_{2} = 0$, $\alpha_{1}+1 =0$ in the DGP. Then $\mathcal{A}_0=\{(-1,1,1)\}$,  $\mathcal X_0=\{x^a,x^b\}$, $\mathcal M_{as}(x^a)=\mathcal M_{as}(x^b)=\{1,2\}$. There are four possible inequality selections $m=(m_1,m_2)\in\{1,2\}^2$.
The corresponding regions that partition the coarse $\mathcal{A}^*$ are
\begin{align*}
A_{(1,1)}&=\{a\in\mathcal A:\ a_1+a_2\ge 0,\ a_1+1\ge 0\}, \quad 
A_{(1,2)}=\{a\in\mathcal A:\ a_1+a_2\ge 0,\ a_1+1< 0\},\\
A_{(2,1)}&=\{a\in\mathcal A:\ a_1+a_2< 0,\ a_1+1\ge 0\}, \quad 
A_{(2,2)}=\{a\in\mathcal A:\ a_1+a_2< 0,\ a_1+1< 0\}.
\end{align*}
Since $\mathcal A$ can be assumed to be large enough to intersect all four wedges, then $L=4$ and $\{A_\ell\}_{\ell=1}^4$ are exactly these sets. With $A_{(k,j)}$, $k,j=1,2$, defined in this way, the maximum score estimator $\widehat{A}$ asymptotically behaves as  
$\widehat{A} \stackrel{d}{\rightarrow} \sum_{k,j=1}^2 d_{k,j} A_{(k,j)},$
where $d_{(k,j)}$, $k,j=1,2$, are binary variables such that $\sum_{k,j=1}^2 d_{k,j}=1$,  $d_{k,j} \cdot d_{k_1,j_1} =0$ for any two different dummies. This is a direct implication of Theorem  \ref{th:general5}. In addition  
$P(d_{(k,j)}=1)=1/4$ for each $(k,j)$, since $\Pr\!\left(
\widehat \PP(Y=1| x^a)-1/2\in C_k,\,
\widehat \PP(Y=1| x^b)-1/2\in C_j
\right)\rightarrow\frac14$
by Proposition~\ref{prop:bccross_probest2}.  Finally, note that the intersection of any three $\overline{A}_{k,j}$ out of these four sets gives us  $\overline{A}_0.$

(1B): $\alpha_{1}+\alpha_{2} \neq 0$, $\alpha_{1}+1 =0$ in the DGP. In this case,  
$\mathcal{X}_0=\{x^b\}$, $\mathcal M_{as}(x^b)=\{1,2\}$. As for the other support point, $\mathcal M_{as}(x^a)=\{m(x^a)\}$ with $m(x^a)=1$ if $ \alpha_{2}>1$ and $m(x^a)=2$ if $ \alpha_{2}<1$. 
We can write $\mathcal A_0
=\{(-1,1,a_2)\in\mathcal A: \, a_2-1\in D_{m(x^a)}\}$. We obtain that $L=2$ and the two sets that partition the coarse $\mathcal{A}^*=\{(a_1,1,a_2)\in\mathcal A: \, a_1+a_2 \in D^*_{m(x^a)}\}$ by  
\[
A_1^{(B)}=
\{a\in\mathcal A:\ a_1+a_2\in D^*_{m(x^a)},\ a_1+1\ge 0\},
\quad
A_2^{(B)}=
\{a\in\mathcal A:\ a_1+a_2\in D^*_{m(x^a)},\ a_1+1< 0\} 
\]
(again, both are nonempty as long as $\mathcal A$ is not too small which we have assumed). The MS estimator $\widehat{A}$ asymptotically behaves as 
$\widehat{A} \stackrel{d}{\rightarrow} d_1A_1^{(B)}+d_2A_2^{(B)},$ 
where $d_i$, $i=1,2$, are binary variables such that $d_1+d_2=1$, $d_1 \cdot d_2=0$. In addition, 
$\PP(d_1=1)=\PP(d_2=1)=0.5$, as is implied by
$\PP(\widehat{\PP}(Y=1| x^b)-0.5\in C_j)
\to0.5$, $j=1,2$.  Finally, note that $\overline{A}_1^{(B)} \cap \overline{A}_2^{(B)}= \overline{\mathcal A_0}$.

(1C):  $\alpha_{1}+\alpha_{2} = 0$, $\alpha_{1}+1 \neq 0$ in the DGP. 
Here $\alpha_{2}=-\alpha_{1}$ and $\alpha_{1}\neq -1$. Then
$\mathcal X_0=\{x^a\}$, $\mathcal M_{as}(x^a)=\{1,2\}$.
As for  $x^b$, $\mathcal M_{as}(x^b)=\{m(x^b)\}$, where $m(x^b)=1$ if $\alpha_{1}>-1$ and $m(x^b)=2$ if $\alpha_{1}<-1$. The identified set can be written as $\mathcal A_0=\{(a_1,1,-a_1)\in\mathcal A:\ a_1+1\in D_{m(x^b)}\}$ (it is one-dimensional). We have $L=2$ as different selection regimes result in partitioning the coarse $\mathcal{A}^*=\{(a_1,1,a_2)\in\mathcal A:\ a_1+1\in D^*_{m(x^b)}\}$ by the following two regions in our collection: 
$$A_1^{(C)}=\{a\in\mathcal A:\ a_1+1\in D^*_{m(x^b)},\ a_1+a_2\ge 0\},
\quad
A_2^{(C)}=
\{a\in\mathcal A:\ a_1+1\in D^*_{m(x^b)},\ a_1+a_2< 0\},$$
which are non empty if $\mathcal{A}$ is not too small. The MS estimator $\widehat{A}$ asymptotically behaves as 
$\widehat{A} \stackrel{d}{\rightarrow} d_1A_1^{(C)}+d_2A_2^{(C)},$ 
where $d_i$, $i=1,2$, are binary variables such that $d_1+d_2=1$, $d_1d_2=0$. Additionally, 
$\PP(d_1=1)=\PP(d_2=1)=0.5$ as in this case $\PP(\widehat{\PP}(Y=1|x^a) -0.5 \in D^*_1) =1- \PP(\widehat{\PP}(Y=1|x^a) -0.5 \in D^*_2)\to 0.5$. Finally, note that $\overline{A}_1^{(C)}\cap\overline{A}_2^{(C)}=
\{(a_1,1,-a_1)\in\mathcal A: a_1+1 \in \overline{D}_{m(x^b)}\}=\overline{\mathcal{A}}_0.$

(1D): $\alpha_{1}+\alpha_{2} \neq 0$, $\alpha_{1}+1 \neq 0$ in the DGP. In this case, $\mathcal{X}_0=\varnothing$, $\mathcal M_{as}(x^j)=\{m(x^j)\}$, $j \in \{a,b\}$. Hence, $L=1$ and 
$A_1=\mathcal{A}^*=\{a\in\mathcal A:\ a_1+a_2 \in D^*_{m(x^a)},\ a_1+1 \in D^*_{m(x^b)}\}.$ 
It also holds in this case $\overline{A}_1=\overline{\mathcal{A}}_0 =\{a\in\mathcal A:\ a_1+a_2 \in \overline{D}_{m(x^a)},\ a_1+1 \in \overline{D}_{m(x^b)}\}.$

Geometric interpretations of this design for all cases (1A)-(1D) are given in Figure \ref{fig:ID1_regions}, where we plot the sets described here in their projection on the 2-dimensional space of components that were not normalized. 

\begin{figure}[htbp]
\centering
\begin{tikzpicture}[font=\small, line join=round, line cap=round, >=stealth]

\def\xmin{-3}\def\xmax{3}
\def\ymin{-3}\def\ymax{3}

\coordinate (DL) at (\xmin,\ymax); 
\coordinate (DR) at (\xmax,\ymin); 
\coordinate (VL1) at (-1,\ymin);   
\coordinate (VL2) at (-1,\ymax);   

\coordinate (Vcap) at (-1,1);

\tikzset{
  Aset/.style={draw, thick},
  Astar/.style={draw, very thick},
  sep/.style={draw, thick, dashed},
  A0/.style={draw, very thick},
  regionfill/.style={fill=black, fill opacity=0.08, draw=none},
  label/.style={inner sep=1pt, fill=white, fill opacity=0.9, text opacity=1}
}

\newcommand{\drawAxesAndBox}[1]{%
  \draw[->] (\xmin-0.3,0) -- (\xmax+0.5,0) node[below right] {$a_1$};
  \draw[->] (0,\ymin-0.3) -- (0,\ymax+0.5) node[above left] {$a_2$};
  \draw[Aset] (\xmin,\ymin) rectangle (\xmax,\ymax);
  \node[anchor=west] at (\xmin,\ymax+0.45) {\textbf{#1}};
}


\begin{scope}[shift={(0,0)}]
  \drawAxesAndBox{Case (1A)}

  \begin{scope}

    \path[regionfill] (\xmin,\ymin) rectangle (\xmax,\ymax);

    \draw[sep] (DL) -- (DR) node[pos=0.15, above left, label] {}; 
    \draw[sep] (VL1) -- (VL2) node[pos=0.85, above, label] {}; %

    \filldraw[black] (Vcap) circle (2.2pt);
    \node[label, anchor=west] at (-0.85,1.08) {$\mathcal A_0$};
  \end{scope}

  \node[label] at (1.2,1.8) {$A_1$};
  \node[label] at (-2.0,2.5) {$A_2$};
  \node[label] at (1,-1.6) {$A_3$};
  \node[label] at (-2.0,-1.6) {$A_4$};

  \node[label] at (2.0,2.55) {}; 
\end{scope}


\begin{scope}[shift={(9.4,0)}]
  \drawAxesAndBox{Case (1B)}

  \begin{scope}
    \clip (\xmin,\ymin) rectangle (\xmax,\ymax);

    \path[regionfill] (-3,3) -- (3,-3) -- (3,3) -- cycle;

    \draw[A0] (-1,1) -- (-1,\ymax) node[pos=0.55, right, label] {$\mathcal A_0$};

    \draw[sep, line width=1.1pt] (-1.07,1) -- (-1.07,\ymax)
      node[pos=0.9, above right, label] {}; 
  \end{scope}

  \node[label] at (2,1.7) {$A_1$};
  \node[label] at (-2.1,2.5) {$A_2$};
\end{scope}


\begin{scope}[shift={(0,-7.2)}]
  \drawAxesAndBox{Case (1C)}

  \begin{scope}
    \clip (\xmin,\ymin) rectangle (\xmax,\ymax);

    \path[regionfill] (-1,\ymin) -- (\xmax,\ymin) -- (\xmax,\ymax) -- (-1,\ymax) -- cycle;

    \draw[sep] (VL1) -- (VL2) node[pos=0.85, above, label] {}; 

    \draw[A0] (-1,1) -- (3,-3) node[pos=0.45, below left, label] {$\mathcal A_0$};

    \draw[sep, line width=1.1pt] (-1,1.09) -- (3,-2.91)
      node[pos=0.25, above left, label] {}; 
  \end{scope}

  \node[label] at (1.6,1.4) {$A_1$};
  \node[label] at (-0.5,-2.5) {$A_2$};

\end{scope}


\begin{scope}[shift={(9.4,-7.2)}]
  \drawAxesAndBox{Case (1D)}

  \begin{scope}
    \clip (\xmin,\ymin) rectangle (\xmax,\ymax);

    \path[regionfill]
      (-1,1) -- (-1,\ymax) -- (\xmax,\ymax) -- (\xmax,\ymin) -- cycle;


    \draw[sep] (DL) -- (DR) node[pos=0.15, above left, label] {$a_1+a_2=0$};
    \draw[sep] (VL1) -- (VL2) node[pos=0.85, above, label] {$a_1+1=0$};
  \end{scope}

  \node[label] at (2.0,1.95) {$A_1$};
  \node[label] at (2.0,1.35) {$\mathcal A_0$};
\end{scope}

\end{tikzpicture}

\caption{Illustrative Design 1: all the sets here are projections in the 2-dimensional space of non-normalized components.  The shaded region in each plot is the coarse population maximizer set $\mathcal A^*$.
The partition $\{A_\ell\}$ of $\mathcal A^*$ is indicated by the dashed boundary line(s).
The identified set $\mathcal A_0$ is indicated by a thick segment (Cases (1B)–(1C)),
a point (Case (1A)), and coincides with $\mathcal A^*$ (up to the boundary) in Case (1D). Case (1D) takes $\alpha_{1}+\alpha_{2}>0$ (also in case (1B)) and $\alpha_{1}+1>0$ (also in case (1C)).}
\label{fig:ID1_regions}
\end{figure}

\subsubsection{Application to Theoretical Example 2 (static panel binary choice)} 

Denote $\Delta \mathbf x_i = x_{i2}-x_{i1}$, $\Delta \PP(Y=1|\mathbf x_i) = \PP(Y_{i2}=1|\mathbf x_i)-\PP(Y_{i1}=1|\mathbf x_i)$, $\Delta \widehat{\PP}(Y=1|\mathbf x_i) = \widehat{\PP}(Y_{i2}=1|\mathbf x_i)-\widehat{\PP}(Y_{i1}=1|\mathbf x_i)$. Let  $X$ be discrete.  We assume that  the true parameter $\alpha$ lies in the interior of $\mathcal A$ in the normalized parameter space. Let $\mathcal X=\{\mathbf x_1,\ldots,\mathbf x_S\}$ denote its finite effective support, following the convention in Section \ref{sec:discrete_setup} (so \(x\) with $\Delta \mathbf x'a$ invariant over $a \in \mathcal{A}$, if it has positive probability, is excluded as uninformative). We start with the individual asymptotic behaviors of plug-in estimators $\widehat \PP(Y=1|\mathbf{x})$.  

\begin{proposition}
\label{prop:bcpanel_probest1}  Suppose \(\mathbf X\) is discrete with the  effective support \(\mathcal X\), as defined in Section \ref{sec:discrete_setup}. Let  $\{(\mathbf x_i,y_i)\}_{i=1}^n$ be a random sample. Then, for a given $\mathbf x \in \mathcal{X}$:  
\begin{itemize}
\item[(a)] If $\Delta \PP(Y=1|\mathbf {x}) \in C_j$ for  $j \in \{1,2\}$, then $\PP(\Delta \widehat \PP(Y=1|\mathbf{x}) \in C_j) \to 1$ as $n \to \infty$.
\item[(b)] If $\Delta \PP(Y=1|\mathbf{x}) \in C_3$, then  $\PP(\Delta  \widehat \PP(Y=1|\mathbf{x}) \in C_m) \to 1/2$ for $m \in \{1,2\}$. 
\end{itemize}
\end{proposition}
Analogously to Theoretical Example~1, on the basis of Proposition~\ref{prop:bcpanel_probest1} we define \begin{equation} \label{X0panel} \mathcal X_0 = \{\mathbf x\in\mathcal X:\Delta \PP(Y=1|\mathbf x)\in C_3\}, \end{equation} and construct the sets $A_1,\ldots,A_L$ as all distinct nonempty solutions to \eqref{eq:sys_outside}--\eqref{eq:sys_inside}, as in Section~\ref{sec:discrete_setup}. Proposition~\ref{prop:bcpanel_probest2} below is used to verify (C3).

\begin{proposition} \label{prop:bcpanel_probest2}
 Suppose \(\mathbf X\) is discrete with the  effective support \(\mathcal X\), as defined in Section \ref{sec:discrete_setup}. Let $\{(\mathbf x_i,y_i)\}_{i=1}^n$ be a random sample. Let $\mathcal{X}_0=\{\mathbf x_1,\ldots,\mathbf x_T\}$ for $\mathcal{X}_0$ defined in \eqref{X0panel}. Then 
 \begin{itemize}
 \item[(a)] $\sqrt n\left(
\Delta\widehat \PP(Y=1| \mathbf x_t)
-
\Delta \PP(Y=1| \mathbf x_t)
\right)_{t=1}^T
\stackrel{d}{\longrightarrow}
\mathcal N(0,\Sigma)$, where
$\Sigma
=
\operatorname{diag}(\sigma_1^2,\ldots,\sigma_T^2)$, 
$
\sigma_t^2
=
\frac{
\operatorname{Var}(Y_{i2}-Y_{i1}|\mathbf X_i=\mathbf x_t)
}{
P(\mathbf X_i=\mathbf x_t)
}$.  By condition (iii) of Theoretical example 2, $\sigma_t^2>0$ for every $t$, so $\Sigma$ is
positive definite.

 \item[(b)]
For any $(m_1,\ldots, m_T) \in \{1,2\}^T$.
$$\PP\left(\left(\cap_{x \in \mathcal{X}\setminus \mathcal{X}_0} (\Delta \widehat \PP(Y=1|\mathbf{x})  \in C_{m(\mathbf x)}) \right) \cap \left(\cap_{t=1}^T (\Delta \widehat \PP(Y=1|\mathbf{x}_t) \in C_{m_t}) \right) \right) \to {1}/{2^T}.$$   
\end{itemize}
\end{proposition}

The verification of the geometric conditions is analogous to the cross-sectional binary choice case. (C10) holds when the ambiguous first difference index functionals are linearly independent. Under arbitrary linear dependence, the geometric pairing condition (C11) continues to hold by the same argument as in the cross-sectional binary choice, with $x_t$ replaced by $\Delta \mathbf x_t$. The panel binary choice model also satisfies the same limiting membership probability separation as the cross-sectional binary choice model (see Theorem \ref{th:MSgeneral2}). If (C10) additionally holds, all $2^T$ sign profiles are feasible and the limiting regions are equiprobable.

\vskip 0.05in 

\noindent \textbf{Illustrative Design 2 (continued)}
Take $\alpha_0 =1$. Then $\mathcal{A}_0=\{1\}$ and $\mathcal{A}^*=(0,+\infty)\cap \mathcal{A}$. Suppose the parameter space $\mathcal{A}$ is a large interval containing both 0 and 1. The  ambiguous support points in the effective support are  $\mathbf x^a=\bigl((0,1)',(1,0)'\bigr)'$, $\mathbf x^b=\bigl((1,0)',(0,1)'\bigr)'$ (recall that in the effective support we eliminate $\mathbf x$ whose first difference index is constant over \(\mathcal A\)). Their first difference index functions are, respectively, $(\Delta \mathbf x^a)'(a,1)=a-1$, $(\Delta \mathbf x^b)'(a,1)=1-a$, 
and 
$\mathcal M_{as}(\mathbf x^a)=\mathcal M_{as}(\mathbf x^b)=\{1,2\}$.

There are therefore four possible sign profiles.  They imply
\[
\begin{array}{c|c|c}
(m_a,m_b) & \text{restrictions} & \text{solution set}\\
\hline
(1,1) & a-1\geq0,\quad 1-a\geq0 & \{1\},\\
(1,2) & a-1\geq0,\quad 1-a<0 & (1,\infty)\cap\mathcal A,\\
(2,1) & a-1<0,\quad 1-a\geq0 & (0,1),\\
(2,2) & a-1<0,\quad 1-a<0 & \varnothing.
\end{array}
\]
Thus, there are three distinct nonempty candidate regions, $A_1=(0,1)$, 
$A_2=\{1\}$, 
$A_3=(1,\infty)\cap\mathcal A$,
whose union is $\mathcal A^*$. This provides a simple illustration of condition (C11).  Pair
$A_1$ and $A_3$, for which
$\overline A_1\cap\overline A_3
=
\{1\}
=
\overline{\mathcal A}_0$,
and leave $A_2$ unpaired.  The unpaired region itself satisfies $\overline A_2
=
\{1\}
=
\overline{\mathcal A}_0$.

By Proposition~\ref{prop:bcpanel_probest2}, each of the four sign profiles occurs with limiting probability $1/4$. The profile $(1,1)$ selects $A_2=\{1\}$, while $(2,1)$ and $(1,2)$ select $A_1$ and $A_3$, respectively. Under $(2,2)$, the two restrictions are incompatible, and the sample objective selects one of the two outer regions according to the relative magnitudes of the two ambiguous score contributions. Let
\[s_v^2=\PP(\mathbf X=\mathbf x^v) \operatorname{Var}(Y_{i2}-Y_{i1}|\mathbf X=\mathbf x^v),\quad v=a,b.\]
Since the two contributions are asymptotically independent centered normals, it follows that
\[\begin{aligned}
\PP(\widehat A=A_1)&\rightarrow \frac14+\frac{1}{2\pi}\arctan\left(\frac{s_a}{s_b}\right), \quad 
\PP(\widehat A=A_2)&\rightarrow\frac14,\quad  \PP(\widehat A=A_3)&\rightarrow \frac14+\frac{1}{2\pi}\arctan\left(\frac{s_b}{s_a}\right). 
\end{aligned}\]
In particular, when $s_a=s_b$, the limiting probabilities are $3/8$, $1/4$, and $3/8$, respectively.

If  e.g. $\alpha_0 \in (0,1)$ in the DGP, then $\mathcal{A}^*=\mathcal{A}_0=(0,1)$, $L=1$, and $A_1$ coincides with $\mathcal{A}^*$.

\subsubsection{Multinomial choice and block maximum score estimation} 

We consider the multinomial rank-order model of Section \ref{sec:theoreticalmultinomial}, with strictly positive conditional choice probabilities.  The parameter space $\mathcal A$
below is understood to incorporate one common location normalization together with 
one common scale normalization. 
Suppose the true parameter $\alpha$ lies in the interior of $\mathcal A$ in the normalized parameter space. We maintain Assumption~\ref{assn:assnexitense}. The effective support $\mathcal{X}$ is taken to be $\{\mathbf{x}_1,\ldots,\mathbf{x}_S\}$. 

We start by formulating the individual asymptotic behaviors of plug-in estimators $\widehat{p}(\mathbf{x}) := \phi(\widehat{F}(Y|\mathbf{x}))= (\widehat \PP(Y=1|\mathbf{x}), \ldots, \widehat{\PP}(Y=J|\mathbf{x}))^{\top}$ in Proposition \ref{prop:multinom_probest1}.  $p(\mathbf x)=\phi(F(Y\mid\mathbf x))$ denotes its
population counterpart.  

\begin{proposition}
\label{prop:multinom_probest1}
  Suppose \(X\) is discrete with the  effective support \(\mathcal X\), as defined in Section \ref{sec:discrete_setup}. Let 
$\{(\mathbf x_i,y_i)\}_{i=1}^n$ be a random sample. Fix
$\mathbf x\in\mathcal X$.

\begin{itemize}

\item[(a)]
If $\PP(Y=k_1|\mathbf x)>
\cdots>
\PP(Y=k_J|\mathbf x)$ for a permutation $(k_1,\ldots,k_J)$ of $(1,\ldots,J)$, then
\[
\mathbb{P}\left(
\widehat \PP(Y=k_1|\mathbf x)>
\cdots>
\widehat \PP(Y=k_J|\mathbf x)
\right)
\rightarrow 1.
\]

\item[(b)]
Suppose the population probability ordering at $\mathbf x$ has
ordered tie blocks $\mathfrak B(p(\mathbf x))
=(B_1,\ldots,B_R)$. Let $M(\mathbf x)
=
\prod_{r=1}^R |B_r| !$. Then, for every strict refinement
$m\in\mathcal M_{as}(\mathbf x)$ of the population ordering, $\mathbb{P}\left(
\widehat p(\mathbf x)\in C_m
\right)
\rightarrow
{1}/{M(\mathbf x)}$.

\end{itemize}
\end{proposition}

On the basis of results of Proposition \ref{prop:multinom_probest1} we define 
\begin{equation} 
\label{X0_multinomial} \mathcal{X}_0 = \{\mathbf{x} \in \mathcal{X}: \PP(Y=k|\mathbf{x}) = \PP(Y=j|\mathbf{x}) \text{ for some } k \neq j\}.
\end{equation}
Using the ordered partitioned block terminology of Section \ref{sec:theoreticalmultinomial}, we can equivalently define it as 
$\mathcal{X}_0 = \{\mathbf{x} \in \mathcal{X}: R(\mathbf{x})<J\},$
where $R(\mathbf{x})$ was the number of ordered tie blocks for $\mathbf{x}$.  

For $\mathbf{x}_t\in\mathcal X_0$ ,  $\mathfrak B(\phi({F}(y|\mathbf{x}_t)))$ has at least one block of size $\ge 2$ and $m(\mathbf{x}_t)$ is such that 
$$
D_{m(\mathbf{x}_t)}=
\bigcap_{j<k}
\begin{cases}
\widetilde D_{jk;>}, &   \text{ if } \PP(Y=j|\mathbf{x}_t)>\PP(Y=k|\mathbf{x}_t),\\
\widetilde D_{jk;<}, &  \text{ if } \PP(Y=j|\mathbf{x}_t)<\PP(Y=k|\mathbf{x}_t),\\
\widetilde D_{jk;=}, & \text{ if } \PP(Y=j|\mathbf{x}_t)=\PP(Y=k|\mathbf{x}_t)
\end{cases}.
$$
The collection $ \mathcal M_{as}(\mathbf{x}_t)$ consists of all $m_t$ such that $D_{m_t}$ replaces each $\widetilde D_{jk;=}$ in this definition with either $\widetilde D_{jk;>}$ or $\widetilde D_{jk;<}$
in such a way that the resulting intersection remains non-empty (i.e., the induced pairwise orderings are mutually consistent).
Equivalently, each regime $m_t \in\mathcal M_{as}(\mathbf{x}_t)$ corresponds to a refinement of $\mathfrak B(\phi({F}(y|\mathbf{x}_t)))$
that breaks all the within-block ties into a strict complete ordering within that block, while preserving the
between-block ordering.

We can construct sets $A_1,\ldots,A_L$ by solving
\eqref{eq:sys_outside}--\eqref{eq:sys_inside} with various profiles
$(m_1,\ldots,m_T)$, where
$m_t\in\mathcal M_{as}(\mathbf x_t)$, and retain the distinct
nonempty solution sets. We next have Proposition \ref{prop:multinom_probest2} which helps us to verify (C3).

\begin{proposition}
\label{prop:multinom_probest2}
 Suppose \(X\) is discrete with the  effective support \(\mathcal X\), as defined in Section \ref{sec:discrete_setup}.   
Let $\{(\mathbf x_i,y_i)\}_{i=1}^n$ be a random sample. 

\begin{itemize}
\item[(a)] We have $\sqrt{n}(\widehat{p}(\mathbf{x}_1)-p(\mathbf{x}_1), \ldots, \widehat{p}(\mathbf{x}_T)-p(\mathbf{x}_T))' \stackrel{d}{\to} \mathcal{N}(0, \Sigma)$ 
where $\Sigma$ is block diagonal with $\Sigma_t
=
\frac{1}{ \PP(\mathbf X=\mathbf x_t)}
\left[
\operatorname{diag}(p(\mathbf{x}_t))-p(\mathbf{x}_t)p(\mathbf{x}_t)'
\right]$. Each $\Sigma_t$ has null space
$\operatorname{span}\{\mathbf 1_J\}$ and is positive definite on
the  subspace
$\mathcal V
=
\{v\in\mathbb R^J:\mathbf 1_J'v=0\}$.

\item[(b)]
For $\mathbf x_t\in\mathcal X_0$, let
$M_t
=
\prod_{r=1}^{R_t}|B_{tr}|!$, where $|B_{tr}|$ are the sizes of the population tie blocks at
$\mathbf x_t$. Then, for every 
$(m_1,\ldots,m_T)
\in
\prod_{t=1}^T
\mathcal M_{as}(\mathbf x_t)$,
\[
\PP\!\left[
\left(
\cap_{\mathbf x\in\mathcal X\setminus\mathcal X_0}
\{
\widehat{p}(\mathbf x)
\in C_{m(\mathbf x)}
\}
\right)
\cap
\left(
\cap_{t=1}^T
\{
\widehat{p}(\mathbf x_t)
\in C_{m_t}
\}
\right)
\right]
\rightarrow
\prod_{t=1}^T\frac{1}{M_t}.
\]

\end{itemize}
\end{proposition}

Verification of the conditions of
Theorems~\ref{th:general1}--\ref{th:general4} for the multinomial
model is given in the Appendix. When
$\mathcal X_0=\varnothing$, the population probability ordering is
strict at every support point, $L=1$, and the block score
criterion is asymptotically equivalent to the original
multinomial maximum score criterion.

When $\mathcal X_0\neq\varnothing$, the relevant local regimes are
the strict refinements of the population tie blocks. A sufficient
condition for (C10) is a full rank condition on the collection of
nonredundant within-block utility contrasts. This is the direct
multinomial analogue of linear independence of the ambiguous
covariate vectors in the binary choice model.

When the full rank condition fails, some complete strict refinement profiles may be infeasible. Nevertheless, a sufficiently small reflection around an interior identified parameter maps every feasible profile into its global reversal. Within each population tie block,
$j_1>\cdots>j_d$
is replaced by
$j_d>\cdots>j_1$.
The closures of the two corresponding regions impose opposite weak
orderings within every population tie block. Their intersection
therefore requires equality of all indices within every such block,
while preserving the same between-block restrictions. Hence,
$\overline A_\ell\cap\overline A_{fl(\ell)}
=
\overline{\mathcal A}_0$.
Thus global reversal verifies the geometric pairing condition (C11),
irrespective of the sizes of the population tie blocks.

The geometric pairing in (C11) does not by itself imply the membership probability  separation analogous to the one for binary choice in Theorem \ref{th:MSgeneral2}. Under (C10), all complete strict refinements are feasible and Proposition~\ref{prop:multinom_probest2} gives equal limiting probabilities across the corresponding regions. When (C10) fails, global equiprobability does not necessarily hold. For the important special case $W(q)=q$, however, score equivariance under global reversal gives equal limiting selection probabilities within each reversal pair. Combined with (C11), this yields the required membership probability separation. The details are given in Appendix~\ref{app:consistency}. Illustrative Design~3 admits a stronger conclusion. 
Although (C10) fails, all four limiting regions are
equiprobable for every strictly increasing $W$.
The argument is given in the Appendix, in the
continuation of Illustrative Design~3.

\section{Criteria used to compare estimators}\label{estimatorandcriteria}

\setcounter{equation}{0}

We now turn to what we consider desirable properties for estimators of the identified set. We consider two properties. The first is a familiar property of 
consistency which looks into whether the estimator recovers the identified set as
the sample size increases (at a fixed DGP). 
The second is local distributional stability which looks into whether the
estimator's limiting distribution varies continuously under
small local changes in the DGP. These
properties address different questions. An estimator e.g. may
have a stable limiting distribution without consistently
estimating the identified set.

\begin{definition}[Consistency]
\label{def:sharpness}
We say that $\widehat{{A}}$ is \emph{consistent} at a given parameter value $\alpha$ in the DGP if $d_H\!\left(\overline{\widehat{{A}}},\, \overline{\mathcal{A}}_0\right) \stackrel{p}{\rightarrow} 0$,
where $d_H(\cdot,\cdot)$ denotes the Hausdorff distance on the normalized parameter space $\mathcal{A}$.\footnote{To allow for empty finite sample estimates, we set
$d_H(\varnothing,K)=d_H(K,\varnothing)
=1+\operatorname{diam}(\mathcal A)$ for nonempty compact $K$,
and $d_H(\varnothing,\varnothing)=0$.
For the RSQ estimator introduced later our results establish nonemptiness with probability
tending to one, so this convention does not affect their
limiting conclusions.}
\end{definition}

As seen in the theoretical examples above, neither $\mathcal{A}_0$ nor $\widehat{A}$ need be closed, and hence not necessarily compact. Since the Hausdorff metric is a metric on the space of nonempty compact sets in finite dimensional Euclidean spaces, Definition~\ref{def:sharpness} takes closures of both sets. This does not lead to a loss of content as two sets with identical closures are indistinguishable under $d_H$, so consistency is inherently defined only up to boundary points. This formulation also aligns with the notion of convergence in probability for random sets in the random set literature (Definition 6.19 in \citeasnoun{molchanov2006book}).\footnote{Note that \cite{beresteanumolinari2008} used an analogous definition of consistency, though in that setting all sets were compact by assumption, rendering the closure operation unnecessary.}

Following the random set framework in \citet{molchanov2006book}, 
we view the closures of the set estimators as random
elements of $\mathcal K_0(\mathcal A)$, which is union of the empty set and the space of nonempty compact subsets of $\mathcal A$, equipped with the Hausdorff
metric.\footnote{Our convention on the application of the Hausdorff metric to empty sets is given in the previous footnote.} Weak convergence
$\overline{\widehat A_n}\stackrel W\longrightarrow\mathscr R$
means convergence in distribution in this metric space.
Equivalently,
$E[f(\overline{\widehat A_n})]
\longrightarrow E[f(\mathscr R)]$ 
for every bounded continuous function
$f:\mathcal K_0(\mathcal A)\to\mathbb R$.
This allows a nondegenerate random set limit.

Definition~\ref{def:sharpness} is pointwise with the DGP held fixed as the sample size increases. Under the conditions of Theorem~\ref{th:MSgeneral2}, maximum score in Theoretical Example~1 is consistent when $\mathcal X_0=\varnothing$ and inconsistent at the configurations with $\mathcal X_0\neq\varnothing$ considered there. Analogous distinctions arise in Theoretical Examples~2 and~3.

Our second property concerns behavior near a process with
binding equality restrictions. In the illustrative designs,
a support point  at which a choice probability contrast is
zero (in Theoretical Example 1, e.g., this contrast is $\PP(Y=1|x)-\gamma$, whereas in Theoretical Example 2 it is $\Delta \PP(Y=1|x)$) imposes an equality restriction on the parameter.
A perturbation that makes the contrast nonzero replaces
the equality by an inequality. Thus, a small change
in observable probabilities can produce a discontinuous
change in the identified set, including a change in its
dimension. We want to describe
how an estimator behaves near such a process.

To study this question, we must specify how the
DGP changes. For the binary choice models, we construct local alternatives
by varying the structural parameter while holding fixed the
distribution of covariates and the conditional distribution
of unobservables given covariates. Let $\PP_{a}$
denote the resulting observable distribution when $a$ is near $\alpha$, with 
$\alpha_n(h)=\alpha+h/\sqrt n$ ($h$ is a  fixed vector of the same dimension as $\alpha$) being a special case for that $a$, and let $\PP_\alpha$
denote the baseline distribution. For the multinomial rank-order model, we hold the covariate distribution fixed and consider families of strictly positive conditional choice probabilities that satisfy the maintained rank-order and are differentiable at $\alpha$. Together with the covariate distribution, the conditional choice probabilities evaluated at $\alpha_n(h)$ determine the observable law $\PP_{\alpha_n(h)}$. A direction $h$ in $\alpha_n(h)=\alpha+h/\sqrt n$ is admissible if it respects the
normalization and $\alpha_n(h)\in\mathcal A$ for all
sufficiently large $n$. We choose root-$n$ local alternatives because under the smoothness conditions given later, these alternatives can change the relevant conditional probabilities on the same scale as their sampling fluctuations.

 Suppose that, for every fixed admissible $h$ in a
neighborhood of zero, we have 
$\overline{\widehat A}
\stackrel W\longrightarrow \mathscr R_h$
under $\PP_{\alpha_n(h)}$.
Our second property essentially asks whether these limiting
distributions approach the baseline limiting
distribution as the local direction approaches zero.

\begin{definition}[Local distributional stability]
\label{def:localrobustness}
Relative to the specified DGP family, $\widehat A$ is
\emph{locally distributionally stable at $\alpha$} if $\mathscr R_h\stackrel{W}{\longrightarrow}\mathscr R_0$ as $h\to0$ through admissible directions.
\end{definition}

The limit defining $\mathscr R_h$ is taken first, as $n\to\infty$ with $h$ fixed. Stability then concerns the behavior of these
limiting laws as $h\to0$. This definition does not require uniform convergence
over local directions and does not require $\mathscr R_0$ to equal
$\overline{\mathcal A_0(\alpha)}$. If the estimator is consistent
at the baseline DGP, then
$\mathscr R_0=\overline{\mathcal A_0(\alpha)}$ almost surely.
For an inconsistent maximum score estimator, $\mathscr R_0$
can instead be nondegenerate.

There is a useful connection between the consistency and local distributional stability properties
when the local sample laws
$\PP_{\alpha_n(h)}^{\,n}$ are contiguous to the baseline
sample laws $\PP_\alpha^{\,n}$. In that case, consistency
at the baseline implies convergence in probability
to the same baseline set under each fixed local
direction for which contiguity holds:  $d_H\!\left(
\overline{\widehat A_n},\,
\overline{\mathcal A_0(\alpha)}
\right)
\stackrel p\longrightarrow 0$
under $\PP_{\alpha_n(h)}^{\,n}$. 
Thus, if contiguity holds for all admissible directions
in a neighborhood of zero, baseline consistency implies
local distributional stability.\footnote{Contiguity
means that any sequence of sample events with
probability tending to zero under $\PP_\alpha^{\,n}$
also has probability tending to zero under
$\PP_{\alpha_n(h)}^{\,n}$. Apply this implication, for
each $\varepsilon>0$, to the event
$\{d_H(\overline{\widehat A_n},
\overline{\mathcal A_0(\alpha)})>\varepsilon\}$.
For a related argument in model selection, see
Appendix~C of the supplement to
\citet{chernozhukov2022locally}.}

We establish local behavior directly under the
conditions stated for each model, without imposing
contiguity as a general assumption. Our use of local
alternatives connects the analysis to the literature
on nonregular inference, including
\citeasnoun{andrews-boundary} and
\citeasnoun{andrews-guggenberger-ema}. Our criterion,
however, concerns the limiting distribution of the
unscaled estimated set. It does not determine the
limits of scaled estimation errors or test statistics,
nor does it establish uniform coverage of confidence
sets.

Section~\ref{sec:quantile} constructs an oracle RSQ functional and a feasible implementation that are both consistent under the conditions given there. We then compare their local behavior with that of Maximum Score.

\section{Random Set Quantile (RSQ) estimator}\label{sec:quantile} 

\setcounter{equation}{0}

Theorem~\ref{th:MSgeneral2} gives membership probability one on the closure of the identified set and at most one half outside it. This motivates estimating the set by thresholding membership probabilities at $\tau\in(1/2,1)$. We establish the corresponding separation for the other two examples below.

To illustrate our idea, let's look at Illustrative Design 2 with compact $\mathcal{A}$ large enough to contain $[0,1]$ in its interior. When   $\alpha=1$ is the parameter value in the DGP, then as shown in Section \ref{sec:generaltheorems_applications}, the weak limit of the maximum score estimator is
supported on the three regions $(0,1)$, $\{1\}$, $(1,\infty)\cap\mathcal A$.
The singleton region is selected with limiting probability $1/4$.
Each outer region has limiting probability $3/8$ when the two ambiguous cells contribute equally variable terms
to the sample objective asymptotically. For the closures, the point $1$ belongs to every realization of the
limiting random set, whereas every point strictly below $1$ belongs
only to the closure of the left region and every point strictly
above $1$ belongs only to the closure of the right region. In other words, 
$p_\infty(1)=1$,
but 
$p_\infty(a)\leq\frac12$
for every $a\neq1$, and, hence, $p_\infty(a) \geq \tau$ with any $\tau>1/2$ is equivalent to characterizing 
$\overline{\mathcal A}_0=\{1\}$.

Formally, we use the probabilistic membership separation  through  the notion of the random set quantile (RSQ). For estimator $\widehat{A}$ whose asymptotic behavior is described by Theorem \ref{th:general5} (either case 1 or 2)  we define the oracle RSQ functional as the $\tau$-quantile
of the random closed set $\overline{\widehat A}$:
\begin{equation}\label{quantile:estimator}
\widehat{A}_{RSQ,\tau}=q_{\tau} \left(
\overline{\widehat{A}}\right), 
\end{equation}
where, following the definition of the $\tau$-th quantile of a random set in \citeasnoun{molchanov2006book} (p.176), $q_{\tau} (
\overline{\widehat{A}}) = \left\{p \left(\cdot;\,\overline{\widehat{A}}\right) \geq \tau \right\}$ 
for {\it coverage function } $p(u\,;\,\overline{\widehat{A}}) \equiv \PP( u \in
\overline{\widehat{A}}).$ Note that this notion from \citeasnoun{molchanov2006book} also applies to vector settings.\footnote{\citet{Molchanov2017}, p.~36, defines random set
quantiles using a family $\mathfrak M$ of compact subsets.
Here $\mathfrak M$ consists of singletons, giving the
pointwise membership definition above.} 

For a fixed sample size and DGP, this quantile is a deterministic functional of the sampling distribution of $\overline{\widehat A}$ and it is not computed from a single realization of $\widehat A$. We refer to it as the oracle RSQ functional. The feasible estimator in Section~\ref{sec:feasible} replaces the required distributional information by bootstrap membership frequencies.

The restriction $\tau>1/2$ cannot in general be
relaxed when $\mathcal X_0\neq\varnothing$. To see this, take Illustrative Design~1C, 
$\overline{\widehat A}$ converges to either
$\overline A_1^{(C)}$ or $\overline A_2^{(C)}$, each occurring with probability $1/2$. Hence the limiting membership probability is
equal to one on $\overline{\mathcal A_0}$ and is $1/2$  at every other point in
the union of the two limiting regions, and zero outside that union. Then for every $\tau\in(1/2,1)$,
$q_\tau(\overline{\widehat A})
\rightarrow 
\overline{\mathcal A_0}$, and this convergence does not hold for $\tau \leq 1/2$. 

When $\mathcal X_0=\varnothing$, no such separation problem arises. In 
our theoretical examples the underlying estimator selects a set
whose closure converges to $\overline{\mathcal A_0}$ with
probability approaching one, so any $\tau\in(0,1)$ gives the same
limiting quantile. We still use $\tau\in(1/2,1)$ throughout,
since this range applies to both cases.

For readers unfamiliar with the notion of the random set quantile, it may be useful to give it a voting interpretation. Imagine
repeatedly drawing independent samples of size $n$ from the
same DGP and computing the maximum score set for each sample.
Each sample casts a vote for every parameter value in the
closure of its maximizing set. The oracle RSQ functional
retains the parameter values receiving votes from at least
a fraction $\tau$ of samples. Thus, $\tau=1/2$ corresponds
to a weak majority voting rule, whereas $\tau>1/2$ requires
a strict majority at the level $\tau$. The membership probability separation
established above explains why a  threshold
for strict majority recovers the closure of the identified set asymptotically.

\subsection{Consistency of the random set quantile estimator in Theoretical examples 1-3}
We now formulate our result regarding the asymptotic behavior of the RSQ estimator in the settings of interest. 

Consider an original estimator $\widehat A$ such that 
\begin{equation}
\label{eq:finite_region_approx}
\PP\!\left(
\widehat A\notin\{A_1,\ldots,A_L\}
\right)\rightarrow0,
\end{equation}
\begin{equation}
\label{eq:region_prob_convergence}
\forall \ell=1,\ldots,L \qquad \PP(\widehat A=A_\ell)
\rightarrow\pi_\ell,
\qquad
\pi_\ell>0,
\qquad
\sum_{\ell=1}^L\pi_\ell=1.
\end{equation}
for a finite
collection of regions $A_1,\ldots,A_L$. We saw such behavior formulated in Theorem \ref{th:general1} (and later in Theorem \ref{th:general5}) with $\mathcal{X}_0 \neq \varnothing$ being particularly interesting case. Denote the corresponding limiting membership probability by
$
p_\infty(u)
=
\sum_{\ell=1}^L
\pi_\ell\mathbf 1\{u\in\overline A_\ell\}.$

 Theorem~\ref{th:estimatorafterQRSE}
below establishes consistency of $\widehat A_{RSQ,\tau}$ under  the following  membership separation condition for the
model under consideration: 
\begin{equation}
\label{eq:membership_separation}
p_\infty(u)=1
\quad\text{for }u\in\overline{\mathcal A_0},
\qquad
p_\infty(u)\leq\frac12
\quad\text{for }u\notin\overline{\mathcal A_0}.
\end{equation}

\begin{theorem}
\label{th:estimatorafterQRSE}
Suppose $\mathcal X_0\neq\varnothing$ and that
${\widehat A}$ satisfies the finite region approximation
and region probability convergence in
\eqref{eq:finite_region_approx}--\eqref{eq:region_prob_convergence}.
If the resulting limiting membership probability $p_\infty$
satisfies \eqref{eq:membership_separation}, then, for every
$\tau\in(1/2,1)$,
$d_H\!\left(
q_\tau(\overline{\widehat A}),
\overline{\mathcal A_0}
\right)
\longrightarrow0$.

If $\mathcal X_0=\varnothing$ and
$
\PP\!\left(
\overline{\widehat A}
=
\overline{\mathcal A_0}
\right)
\rightarrow1,
$
then the same conclusion holds for every $\tau\in(0,1)$.
\end{theorem}
For the Maximum Score estimator in cross sectional binary choice,
property \eqref{eq:membership_separation}  is established in
Theorem~\ref{th:MSgeneral2}. The corresponding arguments for
Theoretical Examples~2 and 3 are given in
Appendix~\ref{app:consistency}. For the multinomial model in particular,
\eqref{eq:membership_separation} holds for a general strictly increasing rank function $W$ under
(C10), and for the linear rank score $W(q)=q$ under (C11).

Although every fixed $\tau\in(1/2,1)$ gives consistency under these conditions, finite sample estimates depend on $\tau$. Section~\ref{sec:monte_carlo} will later discuss  sensitivity to this choice.

\subsection{Local behavior of Maximum Score and oracle RSQ  in Theoretical examples 1-3} \label{sec:local_behavior}

We now move to exploring the properties of Maximum Score and RSQ estimators (both oracle and later feasible) from the perspective of our local distributional stability  criterion. 

We start with the cross-sectional binary choice model with $\mathcal X_0$ and the candidate regions
$A_1,\ldots,A_L$ defined at the baseline DGP. Fix the
covariate distribution and $F_{u|X}$  and vary $\alpha_n(h)=\alpha+h/\sqrt n$. 

Our first result is auxiliary and it is on the behavior of estimated choice probabilities under the drifting sequences. At $x_t\in\mathcal X_0$, the estimated conditional choice
probability under the baseline DGP with $\alpha$ fluctuates around $\gamma$ at rate $n^{-1/2}$.
Under the differentiability condition below, the perturbation
$\alpha_n(h)=\alpha+h/\sqrt n$ changes the corresponding
population probability by an amount of order $n^{-1/2}$.
The limiting comparison among candidate regions therefore
contains both a Gaussian fluctuation and a deterministic
drift. 

\begin{lemma}
\label{lemmaLdriftasymp}
Consider the cross-sectional binary choice model under the assumptions of Section \ref{sec:theoreticalMSbinary} and the discrete i.i.d. setup of Section \ref{sec:discrete_setup}, with fixed positive support-cell probabilities.   Suppose that
$F_{u|X}(\cdot|x)$ is differentiable at zero for every
$x\in\mathcal X$, with $f_{u|X}(0|x)\in(0,\infty)$.

For $x_t\in\mathcal X_0$, define $\mu_t(h)
=
 f_{u|X}(0|x_t)x_t'h$, $\mu(h)=(\mu_1(h),\ldots,\mu_T(h))'$.
Then, under
$\alpha_n(h)=\alpha+h/\sqrt n$,
\[
\PP_{\alpha_n(h)}(Y=1|x_t)-\gamma
=
\frac{\mu_t(h)}{\sqrt n}
+o(n^{-1/2}),\qquad 
\sqrt n
\left(
\widehat \PP(Y=1|x_t)-\gamma
\right)_{t=1}^T
\stackrel{d}{\longrightarrow}
Z+\mu(h),
\]
where $Z\sim N(0,\Sigma)$ and $\Sigma$ is the covariance
matrix in Proposition \ref{prop:bccross_probest2}(a). It is positive definite.
\end{lemma}

Theorem \ref{th:MSdrift} establishes local distributional stability of  Maximum Score by showing that  for its random set limit probabilities of selecting different deterministic regions vary continuously with the local direction.

\begin{theorem}
\label{th:MSdrift}
Suppose the conditions of Theorem \ref{th:MSgeneral2} and
Lemma \ref{lemmaLdriftasymp} hold.

If $\mathcal X_0\neq\varnothing$, define, for
$\ell=1,\ldots,L$, $V_\ell(h,Z)
=
\frac12\sum_{t=1}^T
\PP(X=x_t)\{Z_t+\mu_t(h)\}s_{\ell t}$, 
where $s_{\ell t}=1$ if $A_\ell$ imposes
$x_t'a\geq0$ and $s_{\ell t}=-1$ if it imposes
$x_t'a<0$. Let $M(h,Z)
=
\arg\max_{\ell=1,\ldots,L}V_\ell(h,Z)$. The maximizer is unique with probability one. Define
$\pi_\ell(h)=\PP\{M(h,Z)=\ell\}$, $\ell=1,\ldots,L$. For every fixed admissible $h$,
\[
\PP_{\alpha_n(h)}
\{\widehat A\notin\{A_1,\ldots,A_L\}\}\rightarrow0,
\quad
\PP_{\alpha_n(h)}\{\widehat A=A_\ell\}
\rightarrow\pi_\ell(h).
\]
Hence, 
\(
\overline{\widehat A}
\stackrel{W}{\longrightarrow}
\mathscr R_h:=\overline A_{M(h,Z)}\) under $\PP_{\alpha_n(h)}$. Each function $h\mapsto\pi_\ell(h)$ is continuous. In particular,
$\mathscr R_h\stackrel{W}{\longrightarrow}\mathscr R_0$
as $h\to0$, so maximum score is locally distributionally
stable at $\alpha$.

If $\mathcal X_0=\varnothing$, then, for every fixed
admissible $h$,
\(
\PP_{\alpha_n(h)}
\{\overline{\widehat A}
=\overline{\mathcal A_0(\alpha)}\}\longrightarrow1,
\)
and the same stability conclusion holds.
\end{theorem}

We next apply the RSQ operation to the sampling distribution
under the same local alternatives and define $p_{n,h}(u)
=
\PP_{\alpha_n(h)}
\{u\in\overline{\widehat A_n}\}$, $Q_{n,\tau}(h)
=
\{u\in\mathcal A:p_{n,h}(u)\geq\tau\}$. Thus $Q_{n,\tau}(h)$ is the oracle RSQ functional evaluated
under $\PP_{\alpha_n(h)}$. Theorem \ref{th:QRSEdrift} shows that the oracle RSQ functional is locally stable. The key mechanism is that, for all sufficiently small local directions, the functional returns the closure of the baseline identified set. This mechanism is different from the one underlying local stability of the Maximum Score in Theorem \ref{th:MSdrift}, but it does rely on the continuity of the probabilities for selecting deterministic regions  established there. 

\begin{theorem}
\label{th:QRSEdrift}
Suppose the conditions of Theorem~\ref{th:MSdrift} hold,
and fix $\tau\in(1/2,1)$.

If $\mathcal X_0\neq\varnothing$, define $p_h(u)
=
\sum_{\ell=1}^L
\pi_\ell(h)\mathbf 1\{u\in\overline A_\ell\}$, $Q_\tau(h)=\{u\in\mathcal A:p_h(u)\geq\tau\}$. There exists a neighborhood of zero such that,
for every fixed admissible $h$ in the neighborhood,
\[
Q_\tau(h)=\overline{\mathcal A_0(\alpha)},
\quad
d_H\!\left(
Q_{n,\tau}(h),\overline{\mathcal A_0(\alpha)}
\right)\rightarrow0.
\]

If $\mathcal X_0=\varnothing$, the latter convergence holds
for every fixed admissible $h$.

Hence, the limiting oracle functional is locally constant at zero drift. Equivalently, its associated degenerate laws satisfy the continuity condition in Definition~\ref{def:localrobustness}.
\end{theorem} 
The limit in Theorem~\ref{th:QRSEdrift} is the baseline
identified set closure $\overline{\mathcal A_0(\alpha)}$.
The identified set closures under the local alternatives,
$\overline{\mathcal A_0(\alpha_n(h))}$ do not necessarily converge
to this set.

The proofs of Theorems~\ref{th:MSdrift} and \ref{th:QRSEdrift} extend to the other models once the corresponding shifted Gaussian limit and baseline membership separation are established.

For the static panel model, hold the covariate distribution
and the joint conditional distribution of individual effects
and disturbances fixed. Write $g_t(a)
=
\PP_a(Y_2=1|\mathbf x_t)
-
\PP_a(Y_1=1|\mathbf x_t).$ At $\mathbf x_t\in\mathcal X_0$, $g_t(\alpha)=0$.
Suppose the conditional distribution of the observed outcome
pair is continuous at $\alpha$, each $g_t$ is differentiable
there, and the baseline covariance matrix of the estimated
contrasts is positive definite. Then the bounded
triangular array CLT gives, for every fixed admissible $h$, $\sqrt n\,\widehat g
\stackrel d\longrightarrow
Z+Gh$, where $Z\sim \mathcal N(0,\Omega_{\mathrm{pan}})$ and the $t$th row of $G$ is $\nabla g_t(\alpha)'$.
The  arguments about finite regions and on score comparison then give
continuous limiting  probabilities of selecting different deterministic regions.
Together with the baseline membership separation established
in Appendix~\ref{app:consistency}, this also gives local
constancy of the oracle RSQ functional near zero.

For the multinomial ordering model, fix the covariate distribution and specify a family of conditional choice probabilities $p_j(a,\mathbf x)=\PP_a(Y=j|\mathbf X=\mathbf x)$ that satisfies the maintained ordering restrictions for
$a$ near $\alpha$. Suppose these probabilities are strictly
positive and differentiable at $\alpha$.  Consider a support point
$\mathbf x_t\in\mathcal X_0$ at which a block of alternatives
$B=\{j_1,\ldots,j_d\}$ has equal conditional choice probabilities
under $\alpha$.  Rather than working with all
$\binom{d}{2}$ pairwise probability differences, which are
redundant, take one alternative, say $j_d$, as a reference and
define $g_{t,r}(a)
=
\PP_a(Y=j_r|\mathbf x_t)
-
\PP_a(Y=j_d|\mathbf x_t)$, $ r=1,\ldots,d-1$. 
At the boundary value $\alpha$, $g_{t,r}(\alpha)=0$, $r=1,\ldots,d-1$. For nearby parameter values, the $d-1$ contrasts are sufficient
to determine all pairwise probability differences within the block.

Let
\[
G_t
=
\left.
\frac{\partial g_t(a)}{\partial a'}
\right|_{a=\alpha},
\qquad
g_t(a)
=
(g_{t,1}(a),\ldots,g_{t,d-1}(a))'.
\]
Then, under the local sequence
$\alpha_n(h)=\alpha+h/\sqrt n$, we have $\sqrt n\,g_t(\alpha_n(h))
=
G_t h+o(1)$. Thus the local perturbation produces a shift in the
probability contrasts that determine how the population tie is
resolved. When a support point contains several tie blocks, apply
this construction separately to each block (the block
index is suppressed in the notation).

Stack these contrasts over all baseline tie blocks and
support points, and denote the resulting vector by $g(a)$. The contrast rows are linearly independent and lie in the zero sum subspace within each support cell. Strict positivity of the conditional choice probabilities implies that the multinomial covariance matrix is positive definite on that subspace. Together with positive probabilities $\PP(X=x)$ for support points,
this gives a positive definite covariance matrix $\Omega$
for the stacked contrasts. The triangular array CLT then delivers $\sqrt n\,\widehat g
\stackrel d\longrightarrow
Z+Gh$, $Z\sim \mathcal N(0,\Omega)$, where $G$ stacks the derivative matrices.

For each feasible candidate region, the limiting score
difference relative to a reference region is a linear
function of $Z+Gh$. The argument verifying
\eqref{C3part3unique} for the multinomial model shows that
the difference between the limiting scores of two distinct
regions has positive variance. Adding the deterministic
local drift does not change that variance. Score ties
therefore have probability zero, and the probabilities of the maximizing region
 vary continuously with $h$.

This argument compares scores across feasible regions and
does not require every formal refinement profile to be
feasible. The oracle RSQ conclusion additionally uses
baseline membership separation which  holds for a general
strictly increasing $W$ under (C10), and for $W(q)=q$
under (C11), together with the local conditions
given above.

Thus, under the conditions given above, RSQ combines pointwise consistency
with local distributional stability, whereas Maximum Score
can satisfy the latter without the former.

\subsection{Feasible version of the random set quantile estimator}
\label{sec:feasible}

The oracle RSQ functional depends on the unknown sampling
distribution of the underlying estimator. We construct a
feasible estimator using membership frequencies from
$m$-out-of-$n$ bootstrap samples, with $m\to\infty$ and
$m/n\to0$ (this resampling device builds on the $m$-out-of-$n$ bootstrap literature
\citep{bickelGotzeVanZwet1997}).\footnote{Maximum score resampling has also been
studied in the cube-root setting
\citep{delgadoRodriguezPooWolf2001,abrevayaHuang2005}.} At each fixed DGP, these frequencies recover
the limiting membership probabilities needed for
consistency. Under root-$n$ local alternatives, we 
establish the local behavior of the feasible procedure
separately.

We describe this procedure for an RSQ estimator based on a  general classical estimator in Section  \ref{sec:general_theorems}. Given the original data sample $\{(y_i,x_i)\}_{i=1}^n$, draw 
with replacement $m$ observations to obtain bootstrap samples 
$(y^*_1,x^*_1),\ldots,(y^*_m,x^*_m)$ with $m\to\infty$ and 
$m=o(n)$. For each bootstrap replication $b=1,\ldots,B$, compute 
the estimate on the $b$-th bootstrap sample as  
$$\widehat{A}^{(b)} := \mathrm{Argmax}_{a\in\mathcal{A}}
\sum_{x\in\mathcal{X}_{m,b}} w(x)\,
s\!\big(\phi(\widehat{F}^{(b)}(Y| X=x)),\psi(x,a)\big)
\,\widehat{\PP}^{(b)}(X=x),$$
where $\mathcal{X}_{m,b}$ is the sample support of the $b$-th 
bootstrap sample, $\widehat{F}^{(b)}(\cdot| X=x)$ and 
$\widehat{\PP}^{(b)}(X=x)$ are the corresponding plug-in estimators 
of $F(Y| X=x)$ and $\PP(X=x)$.

Using the closures $\overline{\widehat{A}^{(1)}},\ldots,
\overline{\widehat{A}^{(B)}}$ of the bootstrap maximum score 
estimates, construct the bootstrap approximation to the coverage 
function:
\begin{equation}\label{eq:bootstrap_coverage}
\widehat{p}(u) := \frac{1}{B}\sum_{b=1}^B 
\mathbf{1}\!\left\{u\in\overline{\widehat{A}^{(b)}}\right\}.
\end{equation}
The feasible random set quantile estimator is then defined as 
\begin{equation}\label{eq:feasible_RSQ}
\widehat{A}^{boot}_{RSQ,\tau} := 
\left\{u\in\mathcal{A}: \widehat{p}(u) \geq \tau\right\}, \quad \tau \in \left(\frac{1}{2},1\right).
\end{equation}
We now establish the consistency and local distributional stability of the feasible RSQ estimator in Theoretical Example 1. 

\begin{theorem}[Consistency of the feasible RSQ estimator]
\label{th:bootstrap_consistency}
Consider the cross-sectional binary choice model under the
conditions of Theorem~\ref{th:MSgeneral2}. Let
$\widehat A^{(1)},\ldots,\widehat A^{(B)}$ be obtained from
independent $m$-out-of-$n$ bootstrap samples conditional on the
original sample.

Suppose $m\to\infty$,$\frac mn\to0$, $B\to\infty$. Then, for every $\tau\in(1/2,1)$, $d_H\!\left(
\widehat A^{boot}_{RSQ,\tau},
\overline{\mathcal A_0}
\right)
\stackrel p\longrightarrow0$.

\end{theorem}

\begin{theorem}[Local behavior of the feasible RSQ estimator]
\label{th:bootstrap_robustness}
Suppose the conditions of
Theorem~\ref{th:bootstrap_consistency} and
Lemma~\ref{lemmaLdriftasymp} hold, and consider $\alpha_n(h)=\alpha+\frac{h}{\sqrt n}$ for fixed admissible $h$. Then, under $\PP_{\alpha_n(h)}$,
\[
d_H\!\left(
\widehat A^{boot}_{RSQ,\tau},
\overline{\mathcal A_0(\alpha)}
\right)
\stackrel p\longrightarrow0,
\qquad
\tau\in(1/2,1).
\]
Hence the local limit of the feasible RSQ estimator is constant over all fixed admissible $h$, and the estimator is locally distributionally stable in the sense of Definition~\ref{def:localrobustness}.
\end{theorem}
Theorem~\ref{th:bootstrap_robustness} holds because root-$n$
local drift is negligible relative to bootstrap fluctuations
when $m/n\to0$. To see the mechanism directly, consider an ambiguous support point
$x_t$ in the binary choice model.  Conditional on the original
sample,
\[
\sqrt m\bigl(\widehat \PP^*(Y=1| x_t)-\gamma\bigr)
=
\sqrt m\bigl(\widehat \PP^*(Y=1| x_t)
             -\widehat \PP(Y=1| x_t)\bigr)
+
\sqrt{\frac mn}\,
\sqrt n\bigl(\widehat \PP(Y=1| x_t)-\gamma\bigr).
\]
Under $\PP_{\alpha_n(h)}$, the second term is $O_p(\sqrt{m/n})$
and therefore converges to zero.  The bootstrap consequently
reproduces the same asymptotic tie resolution probabilities as at
the boundary DGP $\alpha$, rather than the $h$-dependent
probabilities governing the original maximum score estimator under
the local sequence.

Note that the bootstrap and the oracle coverage functions can generally have
different local limits as for every fixed admissible $h$,
$\sup_{u\in\mathcal A}
|\widehat p(u)-p_0(u)|
\stackrel p\longrightarrow0$ for the bootstrap 
under $\PP_{\alpha_n(h)}$, whereas $\sup_{u\in\mathcal A}
|p_{n,h}(u)-p_h(u)|\rightarrow0$ for the oracle. Nevertheless, for sufficiently
small $h$, their thresholded limits agree as both
$p_0$ and $p_h$ separate the closure of the baseline identified set from its complement at the threshold $\tau$.

For completeness, we also note that the proof establishes the stronger conclusion that the feasible
estimator equals the closure of the baseline identified set  with
probability approaching one.

The extension of Theorems~\ref{th:bootstrap_consistency} and \ref{th:bootstrap_robustness} to Theoretical
Examples 2 and 3  uses the
same approximation for finite regions and the corresponding
conditional bootstrap CLT for the estimated contrasts.

For the static panel model, resample complete individual
records and use the conditional probability differences
defined in Section~\ref{sec:local_behavior}. Under the
local and covariance conditions above, the
bootstrap contrast vector has the same centered Gaussian
limit as at the baseline DGP as  its root-$n$ local
drift is multiplied by $\sqrt{m/n}$.

For the multinomial model, use the nonredundant within-block
contrasts in Section~\ref{sec:local_behavior}. Under the
specified smooth probability family and positive cell and
choice probabilities, the same conditional bootstrap
argument applies. Combined with baseline membership
separation, this gives pointwise consistency and local
distributional stability under (C10) for a general strictly
increasing $W$, and under (C11) for $W(q)=q$. For Illustrative Design~3, the argument in the Appendix
establishes baseline membership separation for every
strictly increasing $W$, despite the failure of (C10).
Under the local smoothness conditions stated above,
the same bootstrap argument therefore gives pointwise
consistency and local distributional stability for
every such $W$. 

\noindent \textbf{Other possible estimation approaches} We want to note that in binary choice models with discrete regressors, equality restrictions can also be estimated by  thresholding empirical cell probabilities around one half either directly or by incorporating such thresholding into an objective function. These possible approaches are described in our original working paper \cite{KKN2024}. RSQ instead operates through membership frequencies in bootstrap score argmax closures. This does  not necessarily give the same finite sample set as separate classifications of the probability contrasts. 

In our view, RSQ offers two potential practical advantages. First, its membership function shows how often each parameter combination belongs to a bootstrap argmax closure. Increasing the threshold produces progressively smaller sets, revealing which parts receive stronger bootstrap support and persist across thresholds (this is  illustrated in the next two sections). Second, restrictions from different covariate cells involve the same parameters and may conflict in a finite sample. When they cannot all be satisfied simultaneously, score optimization compares feasible parameter regions using the cells’ weighted contributions. RSQ then aggregates these bootstrap results rather than simply imposing a conflicting collection of cellwise restrictions.

\section{Monte Carlo simulations and practical guide}
\label{sec:monte_carlo}

\setcounter{equation}{0}

In this section, we first conduct a Monte Carlo simulation study to evaluate the finite sample performance of the feasible Random Set Quantile (RSQ) estimator and compare it to the standard maximum score estimator. The experimental setup is built on Illustrative Design 1(A) and 1(B) introduced in Section \ref{sec:illustrativedesigns}. 

We then present a guide to practitioners on the choice of quantile indices in the finite sample and what to expect in practice. 

\subsection{Monte Carlo}

 We consider a binary choice model in the Illustrative Design 1: $Y_i = \mathbf{1}\{\alpha_0 + X_{1i} + \alpha_2 X_{2i} - u_i \geq 0\}, $ where $Med(u_i|X_{1i},X_{2i})=0$, and the covariates $(X_{1i}, X_{2i})$ are drawn from a discrete support $\{(0,1), (1,0)\}$ with probabilities: $P((X_1,X_2) = (0,1)) = q > 0.5$, $P((X_1,X_2) = (1,0)) = 1-q$.

In our simulations, we set $q = 0.7$ and take $u_i \sim \mathcal{N}(0,1)$. (the researcher only maintains the conditional median restriction and  \eqref{maxscoredgpc} in Theoretical Example 1, but does not assume normality or independence). 

For the simulations, we denote the intercept by
$\alpha_0$ and its candidate value by $a_0$,
corresponding to $\alpha_1$ and $a_1$, respectively,
in Illustrative Design~1.

We analyze the following two settings: (i) Illustrative Design 1(A): $\alpha_0 = -1$, $\alpha_2 = 1$ (case of point identification); (ii) special case within Illustrative Design 1(B): $\alpha_0 = -1$, $\alpha_2 = 0.5$ (identified set has dimension 1).

\vskip 0.05in 

 The utility indices  $a_0+1$ (obtained when $(X_1,X_2)=(1,0)$) and $a_0+a_2$ (obtained when $(X_1,X_2)=(0,1)$) naturally split the parameter space into 9 cells, which are depicted as $S_1$, \ldots, $S_9$ in Figure \ref{fig:coverage_design1}: $S_1$ is obtained when $<;<$ (both indices are negative), $S_2$ is when $=;<$ (first index is zero, the second one is negative), $S_3$ is when $>;<$, $S_4$ is when $>;=$, $S_5$ is when $>;>$, $S_6$ is when $=;>$, $S_7$ is when $<;>$, $S_8$ is when $<;=$, $S_9$ is when $=;=$. 

In Illustrative Design 1A we have $\mathcal{A}_0$ associated with $S_9$ and  $\mathcal{A}^*$ associated with the union of all these 9 cells. The maximum score estimator asymptotically fluctuates among four sets $A_{(k,j)}$, $k,j=1,2,$ introduced in an earlier discussion of this design in Section \ref{sec:generaltheorems_applications}. Set $A_{(1,1)}$ is associated with $S_4 \cup S_5 \cup S_6 \cup S_9$ and is closed. Set $A_{(2,1)}$ is associated with $S_2 \cup S_3$ with its closure being associated with $S_2 \cup S_3 \cup S_4 \cup S_9$. Set $A_{(1,2)}$ is associated with $S_7 \cup S_8$ with its closure being associated with $S_6 \cup S_7 \cup S_8 \cup S_9$. Set $A_{(2,2)}$ is associated with $S_1$ with its closure being associated with $S_1 \cup S_2 \cup S_8 \cup S_9$.  

In Illustrative Design 1B with our choice of simulation parameters we have $\mathcal{A}_0$ associated with $S_2$ and $\mathcal{A}^*$ associated with $S_1\cup S_2 \cup S_3$. The maximum score estimator asymptotically fluctuates among two sets $A_1^{(B)}$ and $A_2^{(B)}$ (notation used in our earlier discussion of this design in Section \ref{sec:generaltheorems_applications}). Set  $A_1^{(B)}$ is associated with $S_2\cup S_3$ and its closure is associated with  $S_2\cup S_3 \cup S_4 \cup S_9$. Set  $A_2^{(B)}$ is associated with $S_1$ and its closure is associated with  $S_1 \cup S_2\cup S_8 \cup S_9$. 

\noindent \textbf{Estimation procedure.} 
In both 1A and 1B we analyze how our estimators compare to the closure of the identified set. Exact coincidence requires the entire estimated set to be equal to this identified set closure. 

We compare the following estimators. 
\begin{enumerate}
    \item \textit{Standard Maximum Score estimator.} We apply Manski's Maximum Score estimator to the full Monte Carlo sample. 
    
    In both our designs in each sample the maximum score objective function is optimized at one of $A_{(k,j)}$, $k,j=1,2$, or their various unions but these four sets are the smallest sets which can be maximizers in the sample. 

When evaluating how Maximum Score estimate compares to the closure of the identified set, we take the closure of its complete argmax set, including all score ties.

    \item \textit{Random Set Quantile (RSQ) estimator.} We construct the empirical coverage function using an $m$-out-of-$n$ bootstrap procedure. The RSQ estimator evaluates the bootstrap probability that the closure of the bootstrap maximum score set contains each cell. The estimated set is the union of cells whose bootstrap frequency is at least $\tau$, where $1/2<\tau<1$. A covariate value absent from a bootstrap sample, or with equal numbers of the two outcomes, contributes zero to the score. To assess sensitivity to the bootstrap sample size, the local alternative experiment below considers $m=\lfloor n^a\rfloor$ with $a\in\{0.1,0.25,0.5,0.75\}$.

    According to our theory,  the RSQ estimator asymptotically recovers $S_9$  in Illustrative Design 1A and asymptotically recovers $S_2 \cup S_9$ (associated with the closure of the identified set) in our case of Illustrative Design 1B. 
\end{enumerate}

\noindent \textbf{Simulation results.}
\begin{figure}[h]
    \centering
     \includegraphics[width=0.7\textwidth]{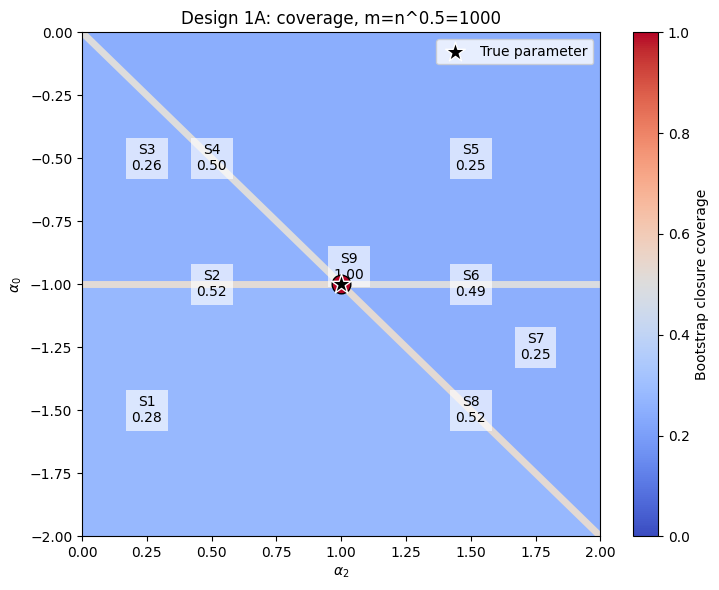}
    \caption{Empirical coverage function for Illustrative Design 1A (point identification case).
For each  cell in the parameter space, the figure reports the bootstrap probability
that the closure of the bootstrap maximum score estimator contains that cell.
The calculation uses $n=1{,}000{,}000$, $m=1{,}000$, and $B=4{,}999$. The displayed frequencies are rounded.}
    \label{fig:coverage_design1}
\end{figure}

\begin{figure}[h]
    \centering
     \includegraphics[width=0.7\textwidth]{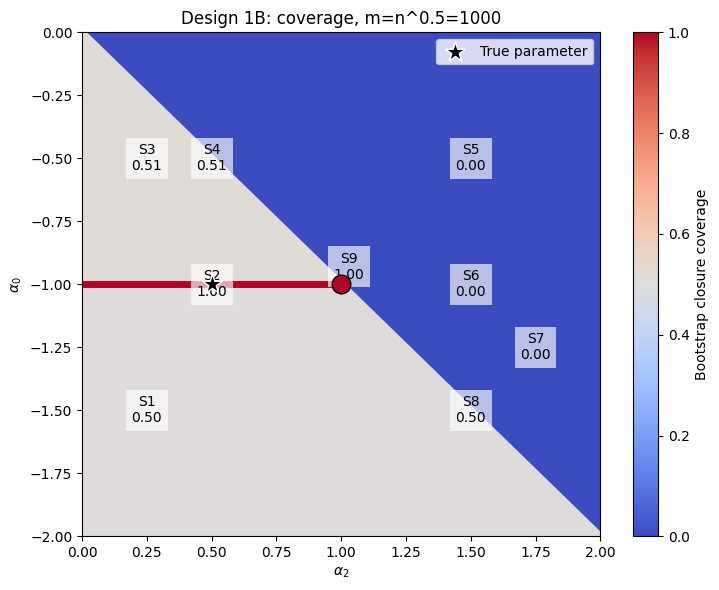}
    \caption{Empirical coverage function for Design 1B ($\alpha_0 = -1, \alpha_2 = 0.5$: partial identification with 1-dimensional identified set). The plot visualizes the bootstrap probability that the closure of the maximum score selected set contains each parameter space cell for a large simulated sample ($n=1{,}000{,}000$), with $m=1{,}000$ and $B=4{,}999$. The displayed frequencies are rounded.}
    \label{fig:coverage_design2}
\end{figure}

Figures \ref{fig:coverage_design1} and \ref{fig:coverage_design2}
display the empirical coverage function over the 9-cell partition for
Illustrative Designs 1A and 1B, respectively. The figures are computed using a very large
simulated sample ($n=1{,}000{,}000$), with $m=1{,}000$ and $B=4{,}999$. They show membership frequencies across bootstrap draws for the displayed samples, rather than recovery frequencies across repeated samples. The threshold ranges below describe exact recovery in these calculations.

In Illustrative Design 1A, the true identified set corresponds to the singleton
cell $S_9$, which belongs to every bootstrap maximum score closure and therefore has empirical coverage one.
Several neighboring cells receive coverage slightly above $0.5$, with the
largest nuisance coverage approximately equal to $0.52$. Consequently, for any
quantile level $\tau\in(1/2,1)$ above the largest nuisance frequency, approximately $0.52$, the feasible RSQ estimator selects only $S_9$. Thus, the estimated set equals the identified set over this range of thresholds.

The smaller bootstrap sample size leaves variability in region selection even when the original sample is large. In this experiment, nuisance membership frequencies above one half explain the larger RSQ sets for $\tau\in(0.5,0.52]$. 

In Illustrative Design 1B, the target set for the RSQ estimator is $S_2\cup S_9$, corresponding to
the closure of the identified set. 
The displayed coverage of both target cells is one to the reported precision. All nuisance cells receive substantially smaller coverage frequencies, approximately up to $0.51$.  The figure therefore shows recovery over approximately $(0.51,1)$ (exact endpoints require the unrounded counts).

\noindent\textbf{Local alternatives.}
To illustrate the behavior of the estimators under local alternatives, we perturb the baseline parameters of Design 1B by setting
\[
\alpha_{0,n}=-1+h/\sqrt n,\qquad \alpha_{2,n}=0.5,
\qquad h\in\{-2,-1,0,1,2\}.
\]
We retain $q=0.7$ and standard normal errors, using
$n\in\{10,100,1000,10000\}$, $m=\lfloor n^a\rfloor$ with
$a\in\{0.1,0.25,0.5,0.75\}$, $\tau=0.75$, and the fixed parameter
space $[-3,2]\times[-2,4]$. 

Each configuration uses $R=2{,}000$ replications and $B=1{,}999$ bootstrap draws per replication.

The local robustness target is the baseline closure $S_2\cup S_9$ and we count the  exact coincidence frequency 
\[
\frac{1}{R}\sum_{r=1}^{R}
\mathbf 1\{\widehat{\mathcal A}_{\mathrm{RSQ},\tau}^{(r)}=S_2\cup S_9\}
\]
requiring both target cells and no additional cells.
For maximum score (MS), we also report coincidence of its closure with
the closure of the current identified set which is $\overline{S_1}$ for $h<0$,
$\overline{S_2}$ for $h=0$, and $\overline{S_3}$ for $h>0$ (except
for $(n,h)=(10,2)$ that has target $\overline{S_5}$). MS coincidence with the closure of the baseline identified set ($S_2\cup S_9$) is zero throughout (every MS closure contains a two-dimensional cell). 

Table~\ref{tab:local_coincidence} reports the exact-coincidence frequencies.
At $n=10{,}000$, RSQ baseline coincidence is $0.000$--$0.002$ for
$m=2$ (the rule $m=\lfloor n^{0.1}\rfloor$ gives resample sizes $1,1,1,2$ at the four sample sizes), $1.000$ for $m=10,100$, and $0.874$--$0.958$ for $m=1000$.
Thus, the smallest subsample rule performs poorly by exact coincidence,
whereas the two intermediate rules recover the baseline closure in every
replication at this sample size. MS coincidence with the current closure
is $0.640$--$0.804$ for $h\ne0$ and zero for $h=0$. (Recall that MS and RSQ rows evaluate different targets when $h\ne0$.)

\begin{table}[h]
\centering\small
\begin{tabular}{clrrrr}
\hline
 & & \multicolumn{4}{c}{$n$}\\
$h$ & Estimator / target / subsample rule & 10 & 100 & 1,000 & 10,000\\
\hline
-2 & MS closure, current target & 0.709 & 0.776 & 0.805 & 0.800\\
 & RSQ, baseline, $m\sim n^{0.1}$ & 0.004 & 0.000 & 0.000 & 0.001\\
 & RSQ, baseline, $m\sim n^{0.25}$ & 0.004 & 0.264 & 0.987 & 1.000\\
 & RSQ, baseline, $m\sim n^{0.5}$ & 0.060 & 0.590 & 0.974 & 1.000\\
 & RSQ, baseline, $m\sim n^{0.75}$ & 0.157 & 0.478 & 0.711 & 0.874\\
\hline
-1 & MS closure, current target & 0.528 & 0.637 & 0.670 & 0.640\\
 & RSQ, baseline, $m\sim n^{0.1}$ & 0.008 & 0.000 & 0.000 & 0.001\\
 & RSQ, baseline, $m\sim n^{0.25}$ & 0.008 & 0.318 & 0.997 & 1.000\\
 & RSQ, baseline, $m\sim n^{0.5}$ & 0.104 & 0.705 & 0.995 & 1.000\\
 & RSQ, baseline, $m\sim n^{0.75}$ & 0.210 & 0.587 & 0.803 & 0.946\\
\hline
0 & MS closure, current target & 0.000 & 0.000 & 0.000 & 0.000\\
 & RSQ, baseline, $m\sim n^{0.1}$ & 0.009 & 0.000 & 0.000 & 0.000\\
 & RSQ, baseline, $m\sim n^{0.25}$ & 0.009 & 0.323 & 0.997 & 1.000\\
 & RSQ, baseline, $m\sim n^{0.5}$ & 0.090 & 0.724 & 0.995 & 1.000\\
 & RSQ, baseline, $m\sim n^{0.75}$ & 0.187 & 0.630 & 0.841 & 0.958\\
\hline
1 & MS closure, current target & 0.327 & 0.631 & 0.668 & 0.664\\
 & RSQ, baseline, $m\sim n^{0.1}$ & 0.004 & 0.000 & 0.000 & 0.000\\
 & RSQ, baseline, $m\sim n^{0.25}$ & 0.004 & 0.250 & 0.996 & 1.000\\
 & RSQ, baseline, $m\sim n^{0.5}$ & 0.044 & 0.594 & 0.990 & 1.000\\
 & RSQ, baseline, $m\sim n^{0.75}$ & 0.107 & 0.555 & 0.801 & 0.941\\
\hline
2 & MS closure, current target & 0.401 & 0.779 & 0.808 & 0.804\\
 & RSQ, baseline, $m\sim n^{0.1}$ & 0.002 & 0.000 & 0.000 & 0.002\\
 & RSQ, baseline, $m\sim n^{0.25}$ & 0.002 & 0.131 & 0.985 & 1.000\\
 & RSQ, baseline, $m\sim n^{0.5}$ & 0.022 & 0.397 & 0.972 & 1.000\\
 & RSQ, baseline, $m\sim n^{0.75}$ & 0.042 & 0.387 & 0.681 & 0.879\\
\hline
\end{tabular}
\caption{Frequencies for exact coincidence under local alternatives.
RSQ targets the baseline closure, MS targets the current DGP's identified set closure.}
\label{tab:local_coincidence}
\begin{minipage}{0.97\textwidth}\footnotesize
\emph{Notes.} $R=2{,}000$, $B=1{,}999$, and $m=\lfloor n^a\rfloor$. RSQ selects cells with coverage at least $0.75$; equality to this threshold is unattainable with $B=1{,}999$. At the four sample sizes, the resample sizes are $(1,1,1,2)$ for $a=0.1$, $(1,3,5,10)$ for $a=0.25$, $(3,10,31,100)$ for $a=0.5$, and $(5,31,177,1000)$ for $a=0.75$. The Monte Carlo standard error of a reported frequency is at most $0.0112$. 
\end{minipage}
\end{table}

Recovery of the baseline closure does not necessarily imply recovery of the closure of the \textit{current}  identified set. For example, at $n=10{,}000$ and $m=100$, RSQ recovers the baseline closure in every replication for each value of $h$. Its exact coincidence with the current closure had frequency zero for $h\ne0$. This distinction is part of the local behavior studied in Section~\ref{sec:feasible}.

\subsection{Practical guide}

We recommend reporting sensitivity to several choices of $m$ as the requirements $m\to\infty$ and $m/n\to0$ concern sequences. Bootstrap samples should also adequately represent the covariate support.

The local calculation in Section~\ref{sec:feasible} explains one consequence of this choice. A population contrast of size $\delta_n$ enters the bootstrap comparison on the scale $\sqrt m\,\delta_n$. Smaller bootstrap samples therefore make small departures from a binding restriction less influential relative to bootstrap variation. The results proved here cover root-$n$ alternatives. A full analysis of alternatives on the $m^{-1/2}$ scale would require a separate argument.

In finite samples, the estimated sets are nested:
$\widehat A^{boot}_{RSQ,\tau_2}
\subseteq
\widehat A^{boot}_{RSQ,\tau_1}$ for $\tau_2>\tau_1$.
They change in discrete jumps, with intervals of thresholds
over which the estimated set is unchanged. We call these
intervals plateaus. Reporting the plateau sequence describes
sensitivity to the threshold choice. Plateau width does not,
by itself, measure confidence coverage or accuracy for the
identified set. Choices approaching $1/2$ or $1$ require
additional control of the membership probability approximation error. 

In practice, the feasible RSQ estimate may be empty for thresholds near one. This occurs when no parameter value belongs to a sufficiently large fraction of the bootstrap argmax closures. This is a finite sample feature. In the two designs considered above, the common point $S_9$ belongs to every bootstrap closure, so the RSQ estimate cannot be empty for $\tau<1$.

\section{Empirical application}\label{application}

\setcounter{equation}{0}

Labour won 202 of the 650 seats in the 2019 UK General
Election, a historic low both in terms of numerical count and proportion since the year 1935.  Debate over its losses focused partly on Brexit. An example is a Guardian opinion article that argued ``It was Brexit,
not left-wing policies, that lost Labour this election.'' \footnote{\url{https://www.theguardian.com/commentisfree/2019/dec/21/labour-leftwing-brexit-policies-election}} 

We examine the association between constituency Leave
voting and the index of seat retention, allowing this
association to differ between Labour-held and
non-Labour-held constituencies. The semiparametric binary
choice model leaves the conditional error distribution
unspecified subject to the maintained restrictions.
The discrete covariates give a finite support setting
in which point identification may fail.

We construct a  dataset at a constituency level covering all 650 UK parliamentary
constituencies. Constituency identifiers, winning parties, vote totals and
electoral majorities for the 2015 and 2019 General Elections are obtained
from the House of Commons Library's constituency results files
\citep{HoC2015data,HoC2019data}. The estimated constituency share voting Leave
in the 2016 EU referendum is taken from the Hanretty estimate column of the
Commons Library's referendum workbook \citep{HoCBrexitdata,Hanretty2017}\footnote{We
use this estimate series consistently rather than the workbook's combined
actual and estimated series.}. Mean total income  among income tax payers is
obtained from HMRC's Survey of Personal Incomes, Table~3.15, for tax years
2015--16 and 2018--19 \citep{HMRC1516,HMRC1819}. Sources are joined using ONS
constituency identifiers, with verified constituency name matching for the
2015--16 HMRC table, which does not contain these identifiers. The outcome
$Y_i$ equals one if the same party won constituency $i$ in both general
elections. We define $X_{i1}$ as an indicator that the estimated Leave share
exceeds $50\%$, and $X_{i2}$ as an indicator that the 2015 winning margin,
measured as the majority divided by valid plus invalid ballots, is below
$5\%$. Nominal mean income growth is the proportional change from 2015--16
to 2018--19. Its categorical measure $X_{i3}$ equals zero for nonpositive
growth, one for positive growth below the sample median ($6.4137\%$), and
two for growth at or above the median. Finally, $X_{i4}$ indicates a Labour
victory in 2015 and $X_{i5}$ is the interaction between $X_{i1}$ and $X_{i4}$, 
allowing the Leave index contrast to differ between Labour-held and non-Labour-held constituencies. The sample contains 561
constituencies ($86.3\%$) that retained the same winning party.

Table~\ref{table:App1} summarizes the variables, and
Table~\ref{Tab2} tabulates Labour victories in 2015
and Leave votes. Of the 232 constituencies Labour held
in 2015, 149 (64.2\%) voted Leave. 

\begin{table}[!ht]
\centering
\caption{Constituency level summary statistics.}
\label{table:App1}
\begin{tabular}{lrrrr}
\toprule
Variable & Mean & Std. dev. & Min & Max \\
\midrule
Same party wins in 2015 and 2019 ($Y$) & 0.8631 & 0.3440 & 0 & 1 \\
Estimated Leave majority ($X_1$) & 0.6246 & 0.4846 & 0 & 1 \\
2015 winning margin below $5\%$ ($X_2$) & 0.0877 & 0.2831 & 0 & 1 \\
2015--16 to 2018--19 income-growth category ($X_3$) & 1.4185 & 0.6380 & 0 & 2 \\
Labour won in 2015 ($X_4$) & 0.3569 & 0.4795 & 0 & 1 \\
Leave $\times$ Labour ($X_5$) & 0.2292 & 0.4207 & 0 & 1 \\
\bottomrule
\end{tabular}
\par\smallskip
\begin{minipage}{0.96\linewidth}\footnotesize
\textit{Notes:} $n=650$. Standard deviations use the $n-1$ denominator.
 The income growth categories contain 53, 272 and 325 constituencies,
respectively. Category zero includes zero growth.
\end{minipage}
\end{table}

\begin{table}[!ht]
\centering
\caption{Referendum classification, Labour incumbency and party retention.}
\label{Tab2}
\begin{tabular}{llrrr}
\toprule
Referendum & Labour won & Constituencies & Same party & Retention \\
classification & in 2015 & & retained & rate \\
\midrule
Leave & Yes & 149 & 102 & 68.5\% \\
Leave & No & 257 & 249 & 96.9\% \\
Remain & Yes & 83 & 83 & 100.0\% \\
Remain & No & 161 & 127 & 78.9\% \\
\midrule
Total & & 650 & 561 & 86.3\% \\
\bottomrule
\end{tabular}
\par\smallskip
\begin{minipage}{0.96\linewidth}\footnotesize
\textit{Notes:} Referendum classification uses estimated constituency Leave
shares. Retention refers to the same winning party in 2015 and 2019.
\end{minipage}
\end{table}

We estimate the semiparametric binary choice model as in Theoretical Example 1, with the 
median restriction $Q_{1/2}(u_i|X_i)=0$ and covariate vector 
${X}_i=(1,X_{i1},\ldots,X_{i5})'$ and normalize 
the coefficient on $X_{i2}$ to $-1$, reflecting the 
natural prior that a close 2015 election is 
negatively associated with retaining a seat in 2019. 
The normalized parameter vector is 
${\alpha}=(\alpha_0,\alpha_1,-1,\alpha_3,
\alpha_4,\alpha_5)'$. We maintain the continuity and strict increase conditions near zero imposed in Theoretical Example 1.

The sample support contains 23 distinct covariate values: 4 have $\widehat \PP(Y=1|x)<0.5$, 17 have
$\widehat \PP(Y=1| x)>0.5$, and 2 have
$\widehat \PP(Y=1| x)=0.5$. Maximum score imposes no index
restriction at the two sample ties.\footnote{Sample ties do not identify
the population set $\mathcal X_0$.  Sampling variation can also
move population ties away from one half.}  Figure
\ref{fig:Empiriccs_preliminary} in the computational supplement shows the estimated probabilities
and cell frequencies. Bootstrap draws can reverse contrasts
both at exact sample ties and at nearby probabilities, affecting
the resulting RSQ restrictions.

\noindent \textbf{Maximum score estimate.} The maximum score estimate $\widehat{A}$ is the set of all 
$a=(a_0,a_1,-1,a_3,a_4,a_5)'$ 
satisfying 17 weak inequality constraints (from support 
points where $\widehat{\PP}(Y=1|x)>0.5$) and 4 strict 
inequality constraints (from those where 
$\widehat{\PP}(Y=1|x)<0.5$). This system is feasible and 
the resulting set has a nonempty interior in the normalized 
five-dimensional space. This describes the sample estimate and does not by itself establish population partial identification. 

From the structure of $\widehat A$ we can already establish several directional conclusions about the index. The constraints imply $\alpha_1\geq0$, so the Leave index contrast is weakly positive for non-Labour constituencies. They also imply $\alpha_1+\alpha_5<0$, so the Leave index contrast in Labour-held constituencies is strictly negative as the interaction term $\alpha_5$ is sufficiently negative to outweigh $\alpha_1$. We can make the interpretation of $\widehat A$ more explicit using the latent
utility indices: $U^*_{00}=\alpha_0-X_2+\alpha_3 X_3$ (Remain, non-Labour 2015); $U^*_{01}=\alpha_0-X_2+\alpha_3 X_3+\alpha_4$ (Remain, Labour 2015); $U^*_{10}=\alpha_0+\alpha_1-X_2+\alpha_3 X_3$ (Leave, non-Labour 2015); $U^*_{11}=\alpha_0+\alpha_1-X_2+\alpha_3 X_3
+\alpha_4+\alpha_5$ (Leave, Labour 2015).  The fact that $\widehat{{A}}\subseteq
[0,1.5]\times[0,+\infty)\times \{-1\}\times[-0.5,0.5]\times
[0,+\infty)\times(-\infty,0].$ implies  
$U^*_{01} \geq U^*_{00}, \quad U^*_{10} \geq U^*_{00}, \quad U^*_{01} \geq U^*_{11}.$
Since the system of inequalities for $\widehat A$ additionally implies $\alpha_4+\alpha_5<0$, we also obtain $U^*_{11}<U^*_{10}$. However, $\alpha_0+\alpha_1+\alpha_4+\alpha_5\geq0$ does not imply $U^*_{11}\geq U^*_{00}$ as their difference is $\alpha_1+\alpha_4+\alpha_5$, whose sign is undetermined in the maximum score set. The sign of $\alpha_4-\alpha_1$ is also undetermined, so we cannot rank $U^*_{01}$ and $U^*_{10}$.

We next examine whether random set quantile estimation gives sharper empirical conclusions.

\noindent \textbf{Random set quantile estimate.}  
We implement feasible RSQ using $m=550$ observations sampled with replacement from the full sample in each of $B=19{,}999$ bootstrap draws.\footnote{The bootstrap uses an i.i.d. working model for constituency observations and does not account for dependence generated by common or geographic electoral shocks.} Computation is over
\[
\mathcal A_5=\{a=(a_0,a_1,-1,a_3,a_4,a_5)':
|a_j|\leq5,\ j\in\{0,1,3,4,5\}\}.
\]
For each parameter value, we compute the fraction of draws whose maximum score argmax closure contains it and retain values with coverage at least $\tau$. The estimate is represented as a union of closed polyhedra, including lower-dimensional boundaries, rather than by inequalities common to all retained points.
Since RSQ aggregates membership in the closures of bootstrap maximum score argmax sets, it does not have to lie within the  maximum score estimate obtained from the full sample. In this application, the reported core satisfies \(a_0=1\) and \(a_3=0\), while  the raw maximum score set requires \(a_0+a_3<1\). The core thus is outside the raw set but within its closure. The computational supplement describes the construction of these unions. All results are conditional on the parameter bounds.

As $\tau$ increases, the estimate contracts through the plateaus defined in Section~\ref{sec:monte_carlo}. We report selected sets and the core plateau below, together with sensitivity to $m$. The threshold $\tau$ is a bootstrap membership threshold. Table~\ref{tab:plateaus} summarizes the geometry at selected thresholds (including the final empty interval). The maximal polyhedral components can overlap, and their union can have a higher affine dimension than any individual component.

\begin{table}[!ht]
\begin{center}
\caption{Geometry of the RSQ estimate at selected thresholds.}
\label{tab:plateaus}
\vskip 0.05in
\begin{tabular}{rccc}
\hline\hline
$\tau$ & Number of maximal & Largest component & Affine dimension\\
 & components & dimension & of the union\\
\hline
0.51 & 2 & 4 & 5\\
0.55 & 3 & 3 & 5\\
0.60 & 4 & 2 & 4\\
0.65 & 3 & 2 & 4\\
0.70 & 2 & 2 & 3\\
0.75 & 1 & 2 & 2\\
0.80 & 1 & 2 & 2\\
0.85 & 1 & 2 & 2\\
0.87 & 1 & 1 & 1\\
0.90 & 0 & --- & ---\\
\hline
\end{tabular}
\begin{minipage}{0.96\linewidth}
\small $m=550$, $B=19{,}999$. Dimensions refer to the five free coefficients in $\mathcal A_5$. Components are maximal closed polyhedral pieces, not necessarily
disjoint or connected components.
\end{minipage}
\end{center}
\end{table}
At lower thresholds, the estimate generally consists of several
polyhedral components. At $\tau=0.55$, it has three maximal components of dimension at most three, and their union has affine dimension five. At $\tau=0.60$, it has four two-dimensional maximal components, and their union has affine dimension four.
For $12{,}319/19{,}999<\tau\leq13{,}148/19{,}999$, approximately
$(0.615981,0.657433]$, the estimate is the union of three triangles with  the core $\mathcal C$ defined below and
\[
\begin{split}
\mathcal T_1=\{a\in\mathcal A_5:\;&
a_0=0,\ a_3=\tfrac12,\ a_5=-a_1-a_4, \ a_1\geq1,\ a_4\geq\tfrac12,\ a_1+a_4\leq5\},\\
\mathcal T_2=\{a\in\mathcal A_5:\;&
a_0=1,\ a_3=0,\ a_5=-1-a_1-a_4,\ a_1\geq0,\ a_4\geq0,\ a_1+a_4\leq4\}.
\end{split}
\]
All three satisfy $a_0+2a_3=1$ and $a_1+a_4+a_5\leq0$. The  second common restriction compares the Labour--Leave index with the non-Labour Remain baseline.  Each triangle is two-dimensional, while their union has affine dimension four.

For $14366/19999<\tau\leq17352/19999$, approximately
$0.718336<\tau\leq0.867643$, the estimate equals the two-dimensional triangle we consider as the \textit{core plateau}: 
\[
\mathcal C=\left\{a\in\mathcal A_5:
a_0=1,\; a_3=0,\; a_5=-a_1-a_4,\; 
a_1\geq0,\; a_4\geq0,\;  a_1+a_4\leq5 \right\}.
\]
The last bound follows from the imposed lower bound $\alpha_5\geq-5$. On this plateau, the positive Leave and Labour index terms are exactly offset by the interaction for Labour--Leave constituencies. 

Above the core interval, only the edge $\mathcal C\cap\{a:a_1=0\}$ remains, for $17352/19999<\tau\leq17978/19999$. The estimate is empty for $\tau>17978/19999\simeq0.898945$.

We use $\mathcal C$ as a descriptive summary as it is the entire estimated set over a substantial interval of thresholds with lower thresholds retaining  additional components, and  higher thresholds removing parts of the triangle itself. We call it the \emph{core plateau} to describe this persistence. We examine sensitivity over \(m\in\{450,550,580\}\). With \(B=5{,}000\), the entire estimate equals \(\mathcal C\) over \((0.7154,0.8442]\), \((0.7212,0.8708]\), and \((0.7188,0.8752]\), respectively. Thus the core persists over the reported range, although the threshold endpoints vary.

On the core, $a_0=1$ fixes the baseline index for a Remain non-Labour constituency with no close election with nonpositive income growth. The restriction $a_3=0$ removes income growth from the estimated systematic index. To make the progression concrete, recall the latent utility indices. Over $\mathcal{C} \cup \mathcal{T}_1 \cup \mathcal{T}_2$, we have $a_1+a_4+a_5\leq0$ implying  $U^*_{11}\leq U^*_{00}$ at common values of $X_2$ and $X_3$. On the core plateau this becomes
\[
U^*_{11}=U^*_{00}=1-X_2,\qquad
U^*_{10}=1-X_2+a_1,\qquad
U^*_{01}=1-X_2+a_4.
\]

On the core, the Labour index contrast is weakly positive among Remain constituencies and weakly negative among Leave constituencies. The within-Labour Leave contrast, \(-a_4\), exactly offsets the Labour–Remain advantage relative to the non-Labour Remain baseline, giving  \(U^*_{11}=U^*_{00}\) and placing
Labour--Leave and non-Labour Remain constituencies at the
same retention index, conditional on the remaining covariates. 
This comparison concerns 149 of the 232 constituencies Labour
held in 2015. Maximum Score leaves the comparison unresolved whereas 
RSQ imposes equality throughout the core plateau.  The ranking between Remain Labour and Leave
non-Labour constituencies remains unresolved because
$a_1-a_4$ can have either sign. These index conclusions certainly  can be interpreted as suggestive of a pattern related to Brexit. 

The probability implications are more limited. 
Under the maintained error assumptions, each parameter vector in the core implies retention probability one half for Labour–Leave and non-Labour Remain profiles when \(X_2=1\), and greater than one half otherwise. (Equal positive indices do not imply equal probabilities across profiles when conditional error distributions differ). Therefore, the probabilistic comparisons do not allow us to support conclusions about  the  Brexit effect claimed by The Guardian.

\section{Conclusion}
\label{sec:conclusion}
In this paper we study semiparametric discrete choice models when covariates are discrete, violating the assumption of the continuity of distribution of regressors often imposed to establish point identification. While the parameters of the model are generally only partially identified, depending on the true value of the parameters of the data generating process, the identified set can be a singleton or a non-singleton set. 

The RSQ construction recovers the closure of the identified
set from membership probabilities of random sets produced
by estimators such as maximum score. Under the conditions
established for each model, the feasible estimator is
pointwise consistent, including where maximum score is
inconsistent. We also characterize its behavior under
specified root-$n$ local alternatives. The feasible
$m$-out-of-$n$ procedure converges to the baseline identified
set closure for every fixed admissible local direction
when $m/n\to0$. This stability concerns the baseline target
and can coexist with persistent error relative to identified
sets that change discontinuously along the alternatives.
We illustrate the estimator using simulations and an
application to the 2019 UK General Election. The framework
also covers the static panel model and the multinomial
rank-order model under the conditions stated above.


{\normalsize 
\bibliographystyle{econometrica}
\bibliography{overallkkn}
}

\appendix

\section{Appendix}

\setcounter{equation}{0}

\subsection{Illustrative design 3: Multinomial rank-order model}

Consider the model with three choices 
$$X_1'a_1=a_{11}+X_{11}a_{12}, \; \; 
X_2'a_2=a_{21}+X_{21}a_{22}, \; \; 
X_3'a_3=a_{31}+X_{31} a_{32}.$$

Suppose the support of $\mathbf X=(X_{11},X_{21},X_{31})'$ is $\{ (1,0)', (0,1)', (1,1)'\}\times \{0,1\}$ with 6 points in total. Normalize location by setting $a_{21}=0$ and normalize scale by setting $a_{12}=-1$.
The remaining parameter vector is 
$(a_{11},a_{22},a_{31},a_{32})$, and we take the parameter space $\mathcal A$ to be sufficiently large
and to have a nonempty four-dimensional relative interior in the
normalized parameter space. The probabilities are generated by additive random utility with continuous  conditionally i.i.d.  errors with full support.  Identification is evaluated under the rank-order model of Section \ref{sec:theoreticalmultinomial}.

Consider the DGP with $\alpha_{11}=1$, 
$\alpha_{22}=1$, 
$\alpha_{31}=\frac12$, 
$\alpha_{32}=-1$ 
(together with the normalizations
$\alpha_{21}=0$ and $\alpha_{12}=-1$).
The deterministic utility indices at the six support points are then
ordered according to 
\[
\PP(Y=1\mid\mathbf X=(1,0,0)')
=
\PP(Y=2\mid\mathbf X=(1,0,0)')
<
\PP(Y=3\mid\mathbf X=(1,0,0)'),
\]
\[
\PP(Y=1\mid\mathbf X=(1,0,1)')
=
\PP(Y=2\mid\mathbf X=(1,0,1)')
>
\PP(Y=3\mid\mathbf X=(1,0,1)'),
\]
\[
\PP(Y=1\mid\mathbf X=(0,1,0)')
=
\PP(Y=2\mid\mathbf X=(0,1,0)')
>
\PP(Y=3\mid\mathbf X=(0,1,0)'),
\]
\[
\PP(Y=1\mid\mathbf X=(0,1,1)')
=
\PP(Y=2\mid\mathbf X=(0,1,1)')
>
\PP(Y=3\mid\mathbf X=(0,1,1)'),
\]
\[
\PP(Y=2\mid\mathbf X=(1,1,0)')
>
\PP(Y=3\mid\mathbf X=(1,1,0)')
>
\PP(Y=1\mid\mathbf X=(1,1,0)'),
\]
\[
\PP(Y=2\mid\mathbf X=(1,1,1)')
>
\PP(Y=1\mid\mathbf X=(1,1,1)')
>
\PP(Y=3\mid\mathbf X=(1,1,1)').
\]
The identified set under the maintained rank-order model is  
\[
\mathcal A_0
=
\left\{
(a_{11},a_{22},a_{31},a_{32})\in\mathcal A:
a_{11}=1,\;
a_{22}=1,\;
0<a_{31}<1,\;
a_{31}+a_{32}<0
\right\}.
\]
The identified set therefore has affine dimension two in the
four-dimensional normalized parameter space.

In general, we can describe the identified set either in terms of observed probabilities using the exhaustive characterization or, equivalently, in terms of the true parameter values  in the DGP analogously to what we did in Illustrative Design 1. 

If for $\alpha$ in the DGP  all the support points give three distinct utility indices  (that is,  $x_k'\alpha_{k} \neq x_j'\alpha_{j}$, $x_k'\alpha_{k} \neq x_h'\alpha_{h}$,  $x_j'\alpha_{j} \neq x_h'\alpha_{h}$ for $k,j,h \in \{1,2,3\}$ all different), then the identified set will be described by strict inequalities and will have affine dimension 4. If $x_k'\alpha_{k} = x_j'\alpha_{j}$, $k\neq j$, for some  support points, then the affine dimension of the identified set is strictly less than 4. 

\subsection{Proofs in Section \ref{sec:main_estimators_coarse}}

\textbf{Proof of Theorem \ref{th:general1}.} 
(a) Condition (C1) guarantees that the system \eqref{eq:sys_outside}--\eqref{eq:sys_inside} for a non-empty $A_{\ell}$ implies the system for $\mathcal{A}^*$. Hence, $A_{\ell} \subseteq \mathcal{A}^*$.

(b) Finally, if $\ell\neq\ell'$, then for some $t$ we have $m_{\ell,t}\neq m_{\ell',t}$, hence
$D^*_{m_{\ell, t}}\cap D^*_{m_{\ell', t}}=\emptyset$ by Condition (C2). This implies $A_\ell\cap A_{\ell'}=\emptyset$.

(c) For every $x\in\mathcal X$, the definition of
$\mathcal M_{as}(x)$ and finiteness of $\mathcal M$ imply $\PP\left(
\phi(\widehat F(Y| X=x))
\in
\cup_{m\in\mathcal M_{as}(x)}C_m
\right)\to1$. Since $\mathcal X$ is finite, these events hold simultaneously for
all $x\in\mathcal X$ with probability approaching one.

We first show that, with probability approaching one, no maximizer
can violate the restrictions associated with the asymptotically
relevant regimes. By Assumption~\ref{assn:assnexitense}, at least
one of the systems
\eqref{eq:sys_outside}--\eqref{eq:sys_inside}
has a nonempty solution set. Choose any $b$ belonging to such a
solution set. Then, for every $x\in\mathcal X$,
$\psi(x,b)\in
\bigcup_{m\in\mathcal M_{as}(x)}D_m^*$. Now consider any $a\in\mathcal A$ for which, at some
$x^\ast\in\mathcal X$, $\psi(x^\ast,a)\notin
\bigcup_{m\in\mathcal M_{as}(x^\ast)}D_m^*$.

Condition \eqref{C3part1b} implies that, with probability approaching
one,
$$s\!\left(
\phi(\widehat F(Y\mid X=x^\ast)),\psi(x^\ast,b)
\right)
-
s\!\left(
\phi(\widehat F(Y\mid X=x^\ast)),\psi(x^\ast,a)
\right)
>
\delta_{x^\ast}.$$
At every support point $x$ for which $\psi(x,a),\psi(x,b)
\in
\bigcup_{m\in\mathcal M_{as}(x)}D_m^*$, Condition \eqref{C3part1a} gives
\[
\left|
s\!\left(
\phi(\widehat F(Y\mid X=x)),\psi(x,b)
\right)
-
s\!\left(
\phi(\widehat F(Y\mid X=x)),\psi(x,a)
\right)
\right|
=o_p(1),
\]
uniformly over the corresponding index values. At any additional
support point at which $a$ lies outside the union of the
asymptotically relevant coarse regions,
\eqref{C3part1b} gives another strictly positive score loss. Since $\mathcal X$ is finite, $w(x)>0$, and $\widehat \PP(X=x)\stackrel{p}{\longrightarrow}\PP(X=x)>0$ $\forall x\in\mathcal X$, the nonvanishing score loss at $x^\ast$ cannot be offset by the
$o_p(1)$ score differences at the remaining support points.
The uniformity in
\eqref{C3part1a}--\eqref{C3part1b} therefore implies that, with
probability approaching one, no such $a$ can maximize the sample
objective.

Thus, with probability approaching one, every
$a\in\widehat A$ satisfies $\psi(x,a)\in
\bigcup_{m\in\mathcal M_{as}(x)}D_m^*$ $\forall x\in\mathcal X$. For $x\notin\mathcal X_0$,
$\mathcal M_{as}(x)=\{m(x)\}$, so this condition reduces to $\psi(x,a)\in D^*_{m(x)}$. For each $x_t\in\mathcal X_0$, choose
$m_t\in\mathcal M_{as}(x_t)$ such that
$\psi(x_t,a)\in D^*_{m_t}$. Then $a$ satisfies the system
\eqref{eq:sys_outside}--\eqref{eq:sys_inside}
corresponding to $(m_1,\ldots,m_T)$. Since this system contains
$a$, it is nonempty and therefore its solution set coincides with
one of $A_1,\ldots,A_L$. Conversely, every $A_\ell$ consists only of parameter values
satisfying the corresponding asymptotically relevant restrictions.
It follows that the parameter values satisfying all such
restrictions are exactly $\bigcup_{\ell=1}^L A_\ell$.

By condition~\eqref{eq:region_constancy}, the sample objective is
constant on each candidate region $A_\ell$ for every realization of
the plug-in probabilities.  Hence, with probability approaching one,
$\widehat A$ is the union of precisely those regions
$A_\ell$ attaining the largest sample objective value. Therefore,
\[
\PP\left(
\widehat A\in
\left\{
\bigcup_{\ell\in\mathcal I}A_\ell:
\varnothing\neq\mathcal I
\subseteq\{1,\ldots,L\}
\right\}
\right)\to1.
\]

We next establish the positive probability statement. Fix
$\ell\in\{1,\ldots,L\}$. On the event
$$\left[
\bigcap_{x\in\mathcal X\setminus\mathcal X_0}
\left\{
\phi(\widehat F(Y\mid X=x))\in C_{m(x)}
\right\}
\right]
\cap
\left[
\bigcap_{t=1}^T
\left\{
\phi(\widehat F(Y\mid X=x_t))
\in C_{m_{\ell,t}}
\right\}
\right],$$
every $a\in A_\ell$ satisfies the coarse restriction corresponding
to the realized regime at every support point. By
\eqref{eq:pointwise_separation1}--\eqref{eq:pointwise_separation2},
every such $a$ attains the pointwise maximum at every support point.
Hence, $A_\ell\subseteq\widehat A$ on this event. Condition \eqref{C3part2} therefore gives $\liminf_{n\to\infty}
\PP(A_\ell\subseteq\widehat A)>0$.

Finally, suppose
\eqref{C3part3}--\eqref{C3part3unique} hold. Subtracting
$\widehat Q_1$ from all region objective values and multiplying by
the positive number $r_n$ does not affect their ordering. Thus
\eqref{C3part3} and the continuous mapping theorem imply
\[
\arg\max_{\ell\in\{1,\ldots,L\}}\widehat Q_\ell
\stackrel{d}{\longrightarrow}
\arg\max_{\ell\in\{1,\ldots,L\}}V_\ell.
\]
The latter maximizer is unique almost surely by
\eqref{C3part3unique}. Hence the probability that two or more
regions jointly maximize the sample objective converges to zero,
and, for every $\ell$, $\PP(\widehat A=A_\ell)
\longrightarrow
\PP\left(
\arg\max_{j\in\{1,\ldots,L\}}V_j=\ell
\right)$.

Together with \eqref{C3part2}, this also implies that each of these
limiting probabilities is strictly positive. Since the limiting
maximizer is almost surely unique, the probabilities sum to one.
Defining $d_\ell
=
1\left\{
\arg\max_{j\in\{1,\ldots,L\}}V_j=\ell
\right\}$, we obtain $\widehat A
\stackrel{d}{\longrightarrow}
\sum_{\ell=1}^L d_\ell A_\ell$, 
where $\sum_{\ell=1}^Ld_\ell=1$, $d_\ell d_{\ell'}=0$
for $\ell\neq\ell'$.
$\blacksquare$

\vskip 0.05in 

\noindent \textbf{Proof of Theorem \ref{th:general2}}
By continuity of $\psi$, any limit of points in $A_\ell$ satisfies the
corresponding closed constraints. Hence,
\[
\overline A_\ell
\subseteq
\Bigl\{a\in\mathcal A:\ 
\psi(x,a)\in \overline D^*_{m(x)} \ \forall x\notin\mathcal X_0,\ 
\psi(x_t,a)\in \overline D^*_{m_{\ell,t}} \ \forall x_t\in\mathcal X_{0,\ell}
\Bigr\}.
\]

For the reverse inclusion, let $a$ satisfy all closed constraints.
Condition (C4) implies that every neighborhood of $a$ intersects $A_\ell$, and therefore
$a\in\overline A_\ell$.

(b) 
Fix $\alpha\in\mathcal A_0$ and $\ell$.
For all $x\notin\mathcal X_0$, $\psi(x,\alpha)\in D_{m(x)}$ by the definition
of $\mathcal A_0$ and by the discrete setup in Section \ref{sec:discrete_setup}.
For $x_t\in\mathcal X_{0,\ell}$, condition (C6) implies $\psi(x_t,\alpha)\in \partial D^*_{m_{\ell,t}}$.

By (C5), every neighborhood of $\alpha$ contains a point in
$A_\ell$ and a point outside $A_\ell$, while preserving all non-ambiguous
constraints \eqref{eq:sys_outside}. Hence $\alpha\in\partial A_\ell$.
Since $\alpha$ was arbitrary, $\mathcal A_0\subseteq\partial A_\ell$.
$\blacksquare$

\medskip

\textbf{Proof of Theorem \ref{th:general3}.}
(a) By (C7), every $a\in\mathcal A^{*,\bullet}$ is the limit of a sequence in $\mathcal A^*$, so evidently  $\mathcal A^{*,\bullet}\subseteq\overline{\mathcal A^*}$. The reverse inclusion follows from continuity of $\psi$, so $\overline{\mathcal A^*}\subseteq\mathcal A^{*,\bullet} $. Thus,  $\overline{\mathcal A^*}=\mathcal A^{*,\bullet}$. 

(b) First, (C1) implies
$A_\ell\subseteq\mathcal A^*$ for every $\ell$, and therefore, 
$\bigcup_{\ell=1}^L\overline A_\ell
\subseteq
\overline{\mathcal A^*}$.

For the reverse inclusion, take
$a\in\overline{\mathcal A^*}$. By continuity of $\psi$,
$a\in\mathcal A^{*,\bullet}$. By (C8), there exists a profile
$\mathbf m=(m_1,\ldots,m_T)$ whose corresponding system
\eqref{eq:sys_outside}--\eqref{eq:sys_inside} has a nonempty
solution set and such that $a$ satisfies the closure of every
constraint in that system. Its solution set therefore coincides
with some $A_\ell$. By (C2), distinct profiles have disjoint solution sets. Since this solution set is nonempty, its generating profile is unique and therefore equals the fixed representative profile for $A_\ell$. By (C4) and
Theorem~\ref{th:general2}(a), $a\in\overline A_\ell$. Hence, 
$\overline{\mathcal A^*}
\subseteq
\bigcup_{\ell=1}^L\overline A_\ell$.
Combining the two inclusions gives the result. 
$\blacksquare$

\textbf{Proof of Theorem \ref{th:general4}.} By Theorem~\ref{th:general2}, $\overline{\mathcal A}_0\subseteq \overline A_\ell$
for every $\ell$, hence $\overline{\mathcal A}_0$ belongs to any intersection of
the $\overline A_\ell$. We want to show that no $\alpha\notin\overline{\mathcal A}_0$ can belong to the intersection
of $\lfloor L/2\rfloor+1$ such sets. 

Suppose first that (C9) and (C10) hold. We establish
$\overline{\mathcal A_0}
=
\left\{
a\in\overline{\mathcal A^*}:
\psi(x,a)\in\overline{D_{m(x)}}
\quad\text{for every }x\in\mathcal X_0
\right\}.$ 
The inclusion from left to right follows from \(\mathcal A_0\subseteq\mathcal A^*\) and continuity of \(a\mapsto\psi(x,a)\). Conversely, let \(a\) belong to the set on the right. Since \(a\in\overline{\mathcal A^*}\), continuity implies $\psi(x,a)\in\overline{D^*_{m(x)}}
=\overline{D_{m(x)}}$
for every $x\notin\mathcal X_0.$ 
By construction, the corresponding exhaustive closed restrictions also hold at every \(x\in\mathcal X_0\). Hence \(a\) satisfies all closed exhaustive restrictions. The first approximation clause of (C9) then implies \(a\in\overline{\mathcal A_0}\).

Fix $a\notin\overline{\mathcal A}_0$.
If $a\notin\overline{\mathcal A}^*$, then $a\notin\overline A_\ell$
for all $\ell$ and the claim is trivial.
Thus suppose $a\in\overline{\mathcal A}^*$. Since $a\notin\overline{\mathcal A}_0$, then there exists at least one $t^\diamond \in\{1,\ldots,T\}$ such that
$u^\diamond:=\psi(x_{t^\diamond},a)\notin \overline D_{m(x_{t^\diamond})}$. Since $a\in\overline{\mathcal A}^*$, we have $u^\diamond \in \overline D^*_{m(x_{t^\diamond})}\setminus \overline D_{m(x_{t^\diamond})}$. 

By  condition (C9), 
\begin{equation}
\label{impliedbyC9} \bigl|
\{\, \widetilde{m}_{t^\diamond} \in \mathcal{M}_{as}(x_{t^\diamond}):\; u^\diamond \in \overline D^{*}_{\widetilde{m}_{t^\diamond}} \,\}
\bigr|
\;\le\;
\Bigl\lfloor \tfrac{1}{2}
\bigl| \mathcal{M}_{as}(x_{t^\diamond})\bigr|
\Bigr\rfloor .
\end{equation}

Condition (C10), applied successively one coordinate at a time,
implies that every complete profile in
$\prod_{t=1}^T\mathcal M_{as}(x_t)$
is feasible. Since the second part of (C9) implies that each
nonempty region is generated by a unique profile,
$L
=
\prod_{t=1}^T
\left|\mathcal M_{as}(x_t)\right|$. By \eqref{impliedbyC9}, at most
$\left\lfloor
\frac{
\left|\mathcal M_{as}(x_{t^\diamond})\right|
}{2}
\right\rfloor$
regimes at $x_{t^\diamond}$ have a closed constraint containing
$u^\diamond$. Therefore, $a$ can belong to the closures of at
most
\[
\left\lfloor
\frac{
\left|\mathcal M_{as}(x_{t^\diamond})\right|
}{2}
\right\rfloor
\prod_{s\neq t^\diamond}
\left|\mathcal M_{as}(x_s)\right|
\leq
\left\lfloor\frac{L}{2}\right\rfloor
\]
candidate regions. Hence $a$ cannot belong to the intersection
of $\lfloor L/2\rfloor+1$ closed regions. 

Suppose instead that (C11) holds. Fix $
a\notin\overline{\mathcal A}_0.
$ For each pair of regions in (C11), $a$ can belong to the closure of at most one of them, since the intersection of the two closures is
$\overline{\mathcal A}_0$. Moreover, $a$ cannot belong to the closure of any unpaired region, since such a region has closure equal to
$\overline{\mathcal A}_0$. It follows that $a$ can belong to at most
$
\left\lfloor\frac{L}{2}\right\rfloor
$ of the $L$ closed regions. Hence, no point outside
$\overline{\mathcal A}_0$ can belong to the intersection of a strict majority of the $\overline A_\ell$'s.

Since Theorem~\ref{th:general2} gives $
\overline{\mathcal A}_0\subseteq\overline A_\ell$ 
for every $\ell,
$ the intersection of any strict majority of the closed regions is exactly
$\overline{\mathcal A}_0$. $\blacksquare$

\vskip 0.1in 

\textbf{Proof of Theorem \ref{th:general5}.}
Case 1 follows directly from
Theorems~\ref{th:general1}--\ref{th:general4}.  In particular,
Theorem~\ref{th:general4} gives the majority intersection conclusion
under either (C9) and (C10), or (C11). 

The closed set distribution follows by aggregating selection
probabilities over regions with the same closure. If every
$A_\ell$ is convex, nonempty regions with the same closure
share their relative interior and hence cannot be disjoint.
Thus, under (C2), their closures are distinct, yielding the
stated special case probabilities.

Case 2 is established directly from the discrete setup of Section~\ref{sec:discrete_setup} 
and the structure of the sample objective function ~\eqref{ex_obective_sample_discrete}, independently of 
Theorems~\ref{th:general1}--\ref{th:general4} and conditions (C1)--(C11).

When $\mathcal{X}_0 = \emptyset$, the asymptotically relevant regime set satisfies 
$\mathcal{M}_{as}(x) = \{m(x)\}$ for every $x \in \mathcal{X}$.  
Since $\mathcal{X}$ is finite with $|\mathcal{X}|$ elements, a union bound gives
\begin{equation*}
    \PP\left(\cap_{x \in \mathcal{X}} 
    \left\{\phi\left(\hat{F}(Y| X = x)\right) \in C_{m(x)}\right\}\right) 
    \geq 1 - \sum_{x \in \mathcal{X}} 
    \PP\left(\phi\left(\hat{F}(Y | X = x)\right) \notin C_{m(x)}\right) \to 1.
\end{equation*}
On the event $\cap_{x \in \mathcal{X}} \{\phi(\hat{F}(Y| X = x)) \in C_{m(x)}\}$, 
which occurs with probability approaching $1$, each term in the sample objective 
function is maximized precisely by those $a$ satisfying 
$\psi(x, a) \in D^*_{m(x)}$, by properties~\eqref{eq:pointwise_separation1}--\eqref{eq:pointwise_separation2} of the 
score function. Since this holds simultaneously for all $x \in \mathcal{X}$, the argmax 
set of the sample objective equals
\begin{equation*}
    \mathcal{A}^* = \left\{a \in \mathcal A : \psi(x, a) \in D^*_{m(x)} 
    \text{ for all } x \in \mathcal{X}\right\}. 
\end{equation*}
Therefore $P(\widehat{A} = \mathcal{A}^*) \to 1$, 
which gives both $\widehat{A} \xrightarrow{p} \mathcal{A}^*$ and 
$\overline{\widehat{A}} \xrightarrow{W} \overline{\mathcal{A}^*}$. 

Note that continuity of $\phi$ and $\psi$ is not required for Case 2. $\blacksquare$

\vskip 0.1in 

\textbf{Proof of Proposition \ref{prop:bccross_probest1}.} 
(a) For each support point $x$, $\PP(X=x)>0$, so the probability
of an empty sample cell goes to zero and
$\widehat \PP(Y=1| x)\stackrel p\longrightarrow \PP(Y=1| x)$.
If $\PP(Y=1| x)\ne\gamma$, consistency preserves the sign
of this contrast with probability approaching one.

(b) If $\PP(Y=1| x)=\gamma$, the  CLT for the conditional frequency gives $\sqrt n\bigl(\widehat \PP(Y=1| x)-\gamma\bigr)
\stackrel d\longrightarrow
\mathcal N\!\left(0,\frac{\gamma(1-\gamma)}{\PP(X=x)}\right)$. The limit is continuous, symmetric, and nondegenerate, so each
strict sign has limiting probability $1/2$.
$\blacksquare$

\vskip 0.05in 

\textbf{Proof of Proposition \ref{prop:bccross_probest2}.} (a) A first-order expansion of the ratio in $\widehat \PP(Y=1| x_t)
=
\frac{
n^{-1}\sum_{i=1}^n
1\{X_i=x_t\}Y_i
}{
n^{-1}\sum_{i=1}^n
1\{X_i=x_t\}
}$ gives 
$\sqrt n\left(
\widehat \PP(Y=1| x_t)-\gamma
\right)
=
\frac{1}{\sqrt n}
\sum_{i=1}^n
\frac{
1\{X_i=x_t\}(Y_i-\gamma)
}{\PP(X=x_t)}
+o_p(1), \quad t=1, \ldots, T.$   Therefore,  the joint result follows from the multivariate central limit
theorem.

It remains to characterize the covariance matrix. For each $t$,
$\operatorname{Var}\left(
\frac{1\{X=x_t\}(Y-\gamma)}{\PP(X=x_t)}
\right)=
\frac{\gamma(1-\gamma)}{\PP(X=x_t)}$. For $t\neq s$, $1\{X=x_t\}1\{X=x_s\}=0$
a.s. 
and, therefore, $Cov\left(
\frac{1\{X=x_t\}(Y-\gamma)}{\PP(X=x_t)},
\frac{1\{X=x_s\}(Y-\gamma)}{\PP(X=x_s)}
\right)=0$, 
which implies  $\Sigma
=
\operatorname{diag}\left(
\frac{\gamma(1-\gamma)}{\PP(X=x_1)},
\ldots,
\frac{\gamma(1-\gamma)}{\PP(X=x_T)}
\right)$,
which is positive definite. This proves part (a).

(b) Since $\mathcal X$ is finite, Proposition
\ref{prop:bccross_probest1}(a) implies $\PP\left(
\bigcap_{x\in\mathcal X\setminus\mathcal X_0}
\left\{
\widehat \PP(Y=1| x)-\gamma\in C_{m(x)}
\right\}
\right)\to1$.

Thus it remains to study the joint regimes at
$x_1,\ldots,x_T$. By part (a), $\sqrt n
\left(
\widehat \PP(Y=1| x_t)-\gamma
\right)_{t=1}^T
\stackrel{d}{\longrightarrow}
\mathcal N(0,\Sigma)$, where $\Sigma$ is diagonal. Hence the coordinates of the limiting
normal vector are independent. Each coordinate is centered normal
and so has probability $1/2$ of being positive/negative. For any $(m_1,\ldots,m_T)\in\{1,2\}^T$, the corresponding event $\cap_{t=1}^T
\left\{
\widehat \PP(Y=1|x_t)-\gamma\in C_{m_t}
\right\}$
is one of the $2^T$ sign orthants. Since the limiting Gaussian distribution is continuous and assigns
probability zero to the boundary of every orthant, weak convergence
implies convergence of the corresponding orthant probabilities, or, in other words, 
$\PP\left(
\cap_{t=1}^T
\left\{
\widehat \PP(Y=1| x_t)-\gamma\in C_{m_t}
\right\}
\right)
\to
\left(\frac12\right)^T.$ 
Combining this with the probability one limiting classification of
the points in $\mathcal X\setminus\mathcal X_0$, we obtain 
\[
\PP\left(
\left[
\cap_{x\in\mathcal X\setminus\mathcal X_0}
\left\{
\widehat \PP(Y=1| x)-\gamma\in C_{m(x)}
\right\}
\right]
\cap
\left[
\cap_{t=1}^T
\left\{
\widehat \PP(Y=1| x_t)-\gamma\in C_{m_t}
\right\}
\right]
\right)
\to
\frac{1}{2^T}. \qquad \blacksquare
\]

\vskip 0.05in 

\noindent \textbf{Verification of conditions (C1)-(C11) for Theoretical example 1}

(C1)  holds because for each $x_t\in\mathcal X_0$, $m(x_t)=3$ and $\mathcal M_{as}(x_t)=\{1,2\}$.
Moreover, $D_3^*=\mathbb R \supseteq D_1^*=[0,\infty)$, $D_3^*=\mathbb R \supseteq D_2^*=(-\infty,0)$.

(C2) holds  because $D_1^* \cap D_2^* = D_1^* \cap D_2  =(-\infty,0) \cap [0,+\infty) = \varnothing.$

(C3) We first verify \eqref{C3part1a}--\eqref{C3part1b}.

If $x\notin\mathcal X_0$, then
$\mathcal M_{as}(x)=\{m(x)\}$. By the form of the maximum score
contribution in \eqref{scorefunction_bc}, the sample score is
constant over $D^*_{m(x)}$, so \eqref{C3part1a} holds immediately.
Moreover, with probability approaching one the sample regime
coincides with $m(x)$, and for every
$u\in D^*_{m(x)}$ and $u'\notin D^*_{m(x)}$, we have $$s\!\left(\phi(\widehat F(Y| X=x)),u\right)
-
s\!\left(\phi(\widehat F(Y| X=x)),u'\right)
=
\left|\widehat \PP(Y=1| x)-\gamma\right|.$$ 
Since $\left|\widehat \PP(Y=1| x)-\gamma\right|
\stackrel{p}{\longrightarrow}
\left|\PP(Y=1| x)-\gamma\right|>0$,
\eqref{C3part1b} holds, for example, with
$\delta_x
=
\frac{1}{2}
\left|\PP(Y=1|x)-\gamma\right|$.

Now let $x_t\in\mathcal X_0$. Then
$\mathcal M_{as}(x_t)=\{1,2\}$ and $D_1^*\cup D_2^*
=
[0,\infty)\cup(-\infty,0)
=
\mathbb R$. Hence, \eqref{C3part1b} is vacuous. From
\eqref{scorefunction_bc},
\[
\sup_{u,v\in D_1^*\cup D_2^*}
\left|
s\!\left(\phi(\widehat F(Y| X=x_t)),u\right)
-
s\!\left(\phi(\widehat F(Y| X=x_t)),v\right)
\right|
=
\left|\widehat \PP(Y=1| x_t)-\gamma\right|.
\]
Since $\PP(Y=1| x_t)=\gamma$, Proposition
\ref{prop:bccross_probest2}(a) implies $\left|\widehat \PP(Y=1| x_t)-\gamma\right|
=
O_p(n^{-1/2})
=
o_p(1)$, which verifies \eqref{C3part1a}.

Condition \eqref{C3part2} follows directly from
Proposition~\ref{prop:bccross_probest2}(b). 

Condition~\eqref{eq:region_constancy} is immediate, since for every $x\in\mathcal X$ the sign of $x'a$ is fixed throughout each candidate region $A_\ell$, and the score depends on $a$ only through this sign.

It remains to verify
\eqref{C3part3}--\eqref{C3part3unique}.
For $m,m'\in\{1,2\}$, write
\[
v(m,m')
=
\begin{cases}
 1, & m=1,\ m'=2,\\
-1, & m=2,\ m'=1,\\
 0, & m=m'.
\end{cases}
\]
With probability approaching one, all
$x\notin\mathcal X_0$ are classified according to their population
regimes. On this event their objective contributions are identical
across $A_1,\ldots,A_L$. Therefore, 
\[
\sqrt n(\widehat Q_\ell-\widehat Q_1)
=
\sum_{t=1}^T
w(x_t)\widehat \PP(X=x_t)
v(m_{\ell,t},m_{1,t})
\sqrt n\bigl(
\widehat \PP(Y=1\mid x_t)-\gamma
\bigr)
\]
for
$\ell=2,\ldots,L$. From Proposition~\ref{prop:bccross_probest2}(a),  
$\sqrt n
(
\widehat Q_2-\widehat Q_1, \ldots, 
\widehat Q_L-\widehat Q_1)^{\top}
\stackrel{d}{\longrightarrow}
(
V_2, \ldots 
V_L)^{\top}$
jointly for $\ell=2,\ldots,L$, where $V_1=0$ and $V_\ell
=
\sum_{t=1}^T
w(x_t)\PP(X=x_t)
v(m_{\ell,t},m_{1,t})Z_t$, $\ell=2,\ldots,L$.
Thus, \eqref{C3part3} holds with $r_n=\sqrt n$.

Finally, for any $\ell\neq\ell'$,
$V_\ell-V_{\ell'}
=
\sum_{t=1}^T
w(x_t)\PP(X=x_t)
v(m_{\ell,t},m_{\ell',t})Z_t$. For $\ell\neq\ell'$, choose the distinct profiles used in the
construction of $A_\ell$ and $A_{\ell'}$. They differ at at least
one $x_t$. Hence at least one coefficient in the
last display is nonzero. Since $\Sigma$ is positive definite, $\operatorname{Var}(V_\ell-V_{\ell'})>0$. Thus, $V_\ell-V_{\ell'}$ is a nondegenerate normal random variable and $\PP(V_\ell=V_{\ell'})=0$.

There are only finitely many pairs $\ell\neq\ell'$, so $\PP\left(
\left|
\arg\max_{\ell\in\{1,\ldots,L\}}V_\ell
\right|=1
\right)=1$. This completes verifying \eqref{C3part3}--\eqref{C3part3unique}.

(C4) Fix $\ell$ and let $a\in\mathcal A$ satisfy the closed
constraints defining $A_\ell$. Choose $b\in A_\ell$ and set
$a_\lambda=(1-\lambda)a+\lambda b$, $0<\lambda<1$.
Convexity gives $a_\lambda\in\mathcal A$. Every weak linear
inequality is preserved, and every inequality that is strict
at $b$ and weak at $a$ is strict at $a_\lambda$.
Thus $a_\lambda\in A_\ell$ and $a_\lambda\to a$, proving (C4).

(C5) First fix any $\alpha\in\mathcal A_0$ that is interior to the
normalized parameter space. Also fix  $\ell=1,\ldots,L$ (and the respective $A_{\ell}$), and any $x_s\in\mathcal X_{0,\ell}$.
By definition of $\mathcal A_0$, we have $x_s'\alpha=0$ for each $x_t\in\mathcal X_0$ and
strict inequalities for all $x\in\mathcal X\setminus\mathcal X_0$:
$$x'\alpha>0 \ \text{ if } m(x)=1,\qquad x'\alpha<0 \ \text{ if } m(x)=2.$$
Since $\mathcal X\setminus\mathcal X_0$ is finite, there exists $\delta>0$ such that at any $b$ such that 
$\|b-\alpha\|<\delta$ all these strict inequalities keep their strict sign. That is,  $x'b\in D_{m(x)}$ for all $x\notin\mathcal X_0$.

Now pick any $b_\ell\in A_\ell$.
By definition of $A_\ell$, $b_\ell$ satisfies
$$x'b_\ell\in D^*_{m(x)}\ \ \forall x\notin\mathcal X_0,
\qquad
x_t'b_\ell\in D^*_{m_{\ell,t}}\ \ \forall x_t\in\mathcal X_{0,\ell}.$$
Moreover, for $x_t\in\mathcal X_{0,\ell}$ the constraint is either $x_t'a\ge 0$ (if $m_{\ell,t}=1$)
or $x_t'a<0$ (if $m_{\ell,t}=2$).
Since $\alpha\in\mathcal A_0$ implies $x_t'\alpha=0$ for every $x_t\in\mathcal X_0$,
the line segment $a_{\lambda}=(1-\lambda)\alpha+\lambda b_\ell$ satisfies  $x_t'\bigl((1-\lambda)\alpha+\lambda b_\ell\bigr)=\lambda\, x_t'b_\ell$. Hence, for every $x_t\in\mathcal X_{0,\ell}$, the sign of $x_t'a_{\lambda}$ along this segment
is determined by $x_t'b_\ell$, and therefore the coarse constraint for that $x_t$ is satisfied for all $\lambda>0$.
(If the coarse constraint is $x_t'a\ge 0$, then $\lambda x_t'b_\ell\ge 0$; if it is $x_t'a<0$, then
$\lambda x_t'b_\ell<0$ as $x_t'b_\ell<0$ due to  $b_\ell\in A_\ell$.) Take $a^+_\varepsilon:=(1-\lambda)\alpha+\lambda b_\ell$ where  $\lambda>0$ is chosen small enough to ensure that $a^+_\varepsilon$ is in our chosen small neighborhood of $\alpha$. Then, as explained above,  $x'a^+_\varepsilon\in D^*_{m(x)}$ for all $x\notin\mathcal X_0$,
and by the argument above, $x_t'a^+_\varepsilon\in D^*_{m_{\ell,t}}\; \forall x_t\in\mathcal X_{0,\ell}$. Such $a^+_\varepsilon$ satisfies (C5). 

Pick any single $x^\star\in\mathcal X_{0,\ell}$.
Since $\alpha\in\mathcal A_0$ implies $x^{\star\prime}\alpha=0$, it suffices to perturb $\alpha$ so that
$x^{\star\prime}a^-_\varepsilon$ moves to the ``wrong'' side of the half-line defining $D^*_{m_{\ell,\star}}$. To do that, choose a direction $v$ with zero in the normalized coordinate
and $x^{\star\prime}v\neq0$,  Such a direction exists as 
$x^\star\neq0$ and $x^{\star\prime}\alpha=0$, while the normalized
coefficient equals one. Set $a^-_\varepsilon:=\alpha-\eta v$ for a very small  $\eta$. Since $\alpha$ is in interior, it ensures that
$a^-_\varepsilon\in\mathcal A$ for sufficiently small $\eta>0$. A sufficiently small  choice of $\eta$ also guarantees $x'a^-_\varepsilon\in D_{m(x)}$ for all $x\notin\mathcal X_0$ (here we also use the fact that $\alpha$ comes from the exhaustive characterization and, hence, satisfies \textit{strict} inequalities for all $x\notin\mathcal X_0$). 
Also, 
$x^{\star\prime}a^-_\varepsilon=-\eta\, x^{\star\prime}v.$ 
Choose the sign of $v$ so that this sign violates the coarse constraint at $x^\star$.  
Namely, if 
$m_{\ell,\star}=1$ and, hence, $D^*_{m_{\ell,\star}}=[0,\infty)$, then pick $x^{\star\prime}v>0$. This will guarantee $x^{\star\prime}a^{-}_{\varepsilon}<0$. If, on the other hand, $m_{\ell,\star}=2$ and, hence, $D^*_{m_{\ell,\star}}=(-\infty,0)$, then pick $x^{\star\prime}v<0$ which will guarantee $x^{\star\prime}a^-_\varepsilon>0$. Such a choice will ensure that $x^{\star\prime}a^-_\varepsilon\notin D^*_{m_{\ell,\star}}$ for this $x^\star\in\mathcal X_{0,\ell}$.

Finally, consider any $a_0\in\mathcal A_0$ on the boundary of $\mathcal A$. The inside approximation follows from the same interpolation with $b_\ell$ used above. For the outside approximation, let $c$ be a point constructed above near an interior identified parameter, satisfying all nonambiguous strict inequalities and
strictly violating the chosen ambiguous constraint. Then $(1-\lambda)a_0+\lambda c$, for $0<\lambda<1$, belongs to $\mathcal A$ by convexity, preserves the nonambiguous inequalities, and violates that constraint. These points converge to $a_0$ as $\lambda\downarrow0$, completing the verification of (C5).

(C6) For $x_s\in\mathcal X_0$, the exhaustive restriction is $x_s'a\in D_3=\{0\}$.
Moreover, $\partial D_1^*=\partial[0,\infty)=\{0\}$, $\partial D_2^*=\partial(-\infty,0)=\{0\}$. Hence $\{0\}\subseteq \partial D^*_{m_{\ell,t}}$ for either $m_{\ell,t}\in\{1,2\}$. 

(C7)  Take $a \in A^{*,\bullet}$ such that $\psi(x,a)=x'a \in \overline{D}^*_{m(x)}\setminus {D}^*_{m(x)}$ for some $x$. Necessary, such $x$ has $m(x)=2$ (since $D_1^*$ and $D_3^*$ are closed sets) and satisfies $x'a=0$ and $\PP(Y=1|x)-\gamma<0$. 

By the definition of $\mathcal{A}^*$, any $b \in \mathcal{A}^*$ satisfies $x'b <0$ for all $x$ with $m(x)=2$, and $x'b \geq 0$ for $m(x)=1$, and no restriction for $m(x)=3$. For $\lambda\in(0,1)$ define $a_\lambda:=(1-\lambda)a+\lambda b$. Then $a_\lambda\to a$ as $\lambda\downarrow 0$, and for all sufficiently small $\lambda$ we have $a_\lambda\in\mathcal A^*$.  Take $U$ to be any ball around $a$ small enough to contain $a_\lambda$ for some sufficiently small $\lambda$.
Then $U$ contains a point $a_\lambda$ such that $\psi(x,a_\lambda)\in D^*_{m(x)}$ for all $x$, which is exactly (C7).

(C8) follows from (C7). For any
$a\in\mathcal A^{*,\bullet}$, let $a_k\in\mathcal A^*$ satisfy
$a_k\to a$. Each $a_k$ induces one of finitely many sign profiles
at the ambiguous support points. Hence, along a subsequence, the
profile is constant and corresponds to a nonempty candidate region
$A_\ell$. Taking limits shows that $a$ satisfies the closed
constraints associated with this same profile. This verifies (C8).

(C9) consists of three parts. The first part follows by the same convex combination argument used to verify (C4), with $\alpha\in\mathcal A_0$ replacing an interior point of $A_\ell$. Condition \(\overline{D^*_{m(x)}}=\overline{D_{m(x)}}\) outside \(\mathcal X_0\) obviously holds. 

The third part verifies that majority intersection recovers $\overline{D}_{m(x_t)}$ for $x_t \in \mathcal{X}_0$.  For binary choice, this means  $m(x_t)=3$ and $ \mathcal{M}_{as}(x_t)=\{1,2\}$, so the only $\mathcal{J}_t$  we have to consider is $\mathcal{J}_t=\mathcal{M}_{as}(x_t)$. Thus, this part of (C9) reduces to  $\overline{D}_3=\{0\}= (-\infty,0]\cap [0,+\infty) = \overline{D}_2 \cap \overline{D}_1.$

The second part of (C9) requires each $\mathcal{F}_{\ell}$ is singleton. 
For each $x_t\in\mathcal X_0$ the two coarse constraints are
$H_t(1)=\{a:\ x_t'a\ge 0\}$ (when $m_t=1$) or $H_t(2)=\{a:\ x_t'a<0\}$ (when $m_t=2$),  
which are disjoint. Hence, if $m\neq m'$, there exists $t$ such that $m_t\neq m'_t$, implying that $A(m)\subseteq H_t(m_t)$, $A(m')\subseteq H_t(m'_t)$, 
with potentially one or both sets being empty. If both are nonempty, then necessarily $A(m)\cap A(m')=\varnothing$. Thus, no two distinct profiles can generate the same nonempty region, and each $\mathcal F_\ell$ is a singleton.

(C10) Suppose that vectors in $\mathcal{X}_0$  are linearly independent. Since $x_t'\alpha=0$ for each $t=1,\ldots,T$,  and the normalized coefficient equals one, these vectors remain linearly independent after deleting the normalized coordinate.

For any prescribed sign profile, linear independence  allows us to choose a coefficient change $v$, with zero in the normalized coordinate, such that $x_t'v=1$ or $-1$ according to the desired sign at each  $t$. At $\alpha+\varepsilon v$, these indices equal $\varepsilon x_t'v$, so every $\varepsilon>0$ produces the prescribed signs.

Since $\alpha$ is interior to the normalized parameter space,
a sufficiently small $\varepsilon>0$ keeps the perturbed parameter
in $\mathcal A$. It also preserves all nonambiguous restrictions,
which are finitely many strict inequalities at $\alpha$.
Thus every sign profile is feasible, establishing (C10).

When the regressors in $\mathcal{X}_0$ 
 are linearly dependent, the local flip feasibility required by (C10) may fail. In that case, the weaker pairing condition (C11) will hold, as shown below. 

(C11)  We now allow arbitrary linear dependence among the elements of
$\mathcal X_0$.

Fix a nonempty region $A(m)$ and choose
$a^\circ\in\operatorname{relint}A(m)$ sufficiently close to some
$\alpha\in\mathcal A_0$.  At each $x_t\in\mathcal X_0$,
$x_t'a^\circ$ is either strictly positive, strictly negative, or zero.
Moreover, if $x_t'a^\circ=0$, then $x_t'a=0$ throughout $A(m)$ as, 
otherwise, a relative interior point could be perturbed within
$A(m)$ to make the inequality strict. Reflect $a^\circ$ around $\alpha$ by taking  $a^- = 2\alpha-a^\circ$. Since $x_t'\alpha=0$ for every $x_t\in\mathcal X_0$, we have $x_t'a^-=-x_t'a^\circ$. Thus reflection reverses every strict ambiguous inequality and leaves
the inequalities that are identically zero unchanged.  Since 
$a^\circ$ can be chosen arbitrarily close to $\alpha$, all strict
restrictions outside $\mathcal X_0$ remain satisfied.  Hence $a^-$
belongs to another nonempty region, say $A(m^-)$. The indices that are identically zero on $A(m)$ are also identically
zero on $A(m^-)$.  Otherwise, a feasible direction in $A(m^-)$ making
one of these indices positive could be added, with sufficiently small
weight, to the direction $a^\circ-\alpha$.  This would preserve all
strict inequalities defining $A(m)$ while making that index strictly
positive, contradicting that it is identically zero on $A(m)$.

Therefore the nonempty regions can be paired by reversing all strict
ambiguous inequalities while leaving the structurally zero ones
unchanged.  For every such nontrivial pair, $\overline{A(m)}\cap\overline{A(m^-)}
=
\overline{\mathcal A}_0$,
as the two closures impose opposite weak inequalities at every
strictly varying ambiguous index and hence force all such indices to
zero. The remaining ambiguous indices are already identically zero.

If a region has no strictly varying ambiguous index, then all
ambiguous equalities hold throughout that region and $\overline{A(m)}=\overline{\mathcal A}_0$. $\blacksquare$

 \vskip 0.05in

\noindent\textbf{Proof of Theorem \ref{th:MSgeneral2}.} Suppose first that $\mathcal X_0\neq\varnothing$. The verification above
establishes the conditions required for part~1 of
Theorem~\ref{th:general5}, so its random set conclusions apply.

It remains to establish the probability statements. Let
$Z=(Z_1,\ldots,Z_T)'$
be the centered Gaussian limit in
Proposition~\ref{prop:bccross_probest2}(a). 

Since
$\overline{\mathcal A}_0\subseteq\overline A_\ell$
for every $\ell$, we immediately have
$p_\infty(u)=1$
 for $u\in\overline{\mathcal A}_0$.

Now take $u\notin\overline{\mathcal A}_0$. If $u$ belongs to none of
the $\overline A_\ell$, then $p_\infty(u)=0$. Otherwise, suppose
$u\in\overline A_\ell$, and choose
$a^\circ\in\operatorname{relint}A_\ell$ sufficiently close to
$\alpha$. For every $x_t\in\mathcal X_0$ such that $x_t'u\neq0$, the inequality
defining $A_\ell$ at $x_t$ must have the same sign as $x_t'u$.
We construct a feasible comparison region by reversing the inequality
at each such $x_t$, while keeping unchanged the inequalities at those
$x_t$ for which $x_t'u=0$. To see that this comparison region is feasible, for $c>0$ consider the
direction
$(a^\circ-\alpha)-c(u-\alpha)$.
Since $x_t'\alpha=0$ for every $x_t\in\mathcal X_0$,
\(
x_t'\bigl[(a^\circ-\alpha)-c(u-\alpha)\bigr]
=
x_t'a^\circ-cx_t'u.
\) 
 For sufficiently large $c$, this direction therefore reverses exactly
the ambiguous inequalities for which $x_t'u\neq0$. Taking a sufficiently
small step from $\alpha$ in this direction preserves the parameter-space
restriction and all strict restrictions outside $\mathcal X_0$. 
Hence, the resulting sign profile defines another nonempty region.

The limiting score difference between $A_\ell$ and this comparison
region is 
$\sum_{\{t:x_t'u\neq0\}}
\PP(X=x_t)Z_t\operatorname{sgn}(x_t'u)$. 
Importantly, this expression depends on $u$ but not on which region
$A_\ell$ with $u\in\overline A_\ell$ is considered. Hence, if the
selected limiting region has $u$ in its closure, its limiting score
must be at least as large as that of its comparison region, and, therefore, $\sum_{\{t:x_t'u\neq0\}}
\PP(X=x_t)Z_t\operatorname{sgn}(x_t'u)\geq0$.

Since $u\notin\overline{\mathcal A}_0$ and belongs to the closure of
at least one limiting region, there is at least one $t$ such that
$x_t'u\neq0$. Proposition~\ref{prop:bccross_probest2}(a) therefore
implies that the sum above is a nondegenerate centered normal random
variable. It is nonnegative with probability one half, and hence, 
$p_\infty(u)\leq\frac12$. This proves part~1(a).

For part~1(b), under (C10) all $2^T$ sign profiles are feasible.
By (C9), verified above for the binary choice model, distinct profiles
generate distinct nonempty regions. Hence, $L=2^T$.  Proposition~\ref{prop:bccross_probest2}(b) gives limiting probability
$2^{-T}$ to each profile. Therefore, $\pi_\ell=\frac1L$, 
$\ell=1,\ldots,L$.

Finally, suppose that $\mathcal X_0=\varnothing$. By
Proposition~\ref{prop:bccross_probest1}(a), all sample regimes coincide
with their population counterparts with probability approaching one.
Hence, 
$L=1$, $A_1=\mathcal A^*$, $\widehat A\stackrel{d}{\longrightarrow}\mathcal A^*$. Moreover, in this case the exhaustive and coarse characterizations
differ at most at boundary points, so $\overline{\mathcal A}_0=\overline{\mathcal A^*}$. It follows that $\overline{\widehat A}
\stackrel{W}{\longrightarrow}
\overline{\mathcal A}_0$. $\blacksquare$

 \vskip 0.05in  
\paragraph{Compatibility with the observed data.} Fix a support cell $\mathbf x=(x_1,x_2)$ and a coefficient vector $a$ satisfying the sign restrictions. These restrictions allow us to choose a common continuous error distribution with full support whose probability above $-x_s'a$ equals $\PP(Y_s=1|\mathbf X=\mathbf x)$ for each period.

Set the individual effect to zero. First draw an outcome pair from the observed conditional joint distribution of $(Y_1,Y_2)$. For each period, draw the error from the common distribution conditional on being above $-x_s'a$ if the drawn outcome is one, and at or below $-x_s'a$ otherwise. The resulting errors have the same marginal distribution in both periods, satisfying conditional stationarity, and the model reproduces the drawn outcome pair. Thus every coefficient vector satisfying the sign restrictions is compatible with the observed joint distribution, establishing that the characterization is exhaustive. 

\vskip 0.05in 

\noindent \textbf{Proof of Proposition  \ref{prop:bcpanel_probest1}.} For each $\mathbf x\in\mathcal X$,
$\PP(\mathbf X=\mathbf x)>0$.
Applying the multivariate central limit theorem and the
delta method to the empirical joint frequencies and
the empirical cell frequency gives
\(
\sqrt n\left(
\widehat{\PP}(Y_s=1|\mathbf x)
-\PP(Y_s=1|\mathbf x)
\right)_{s=1,2}
\stackrel d\longrightarrow
\mathcal N(0,\Sigma_{\mathbf x}),
\)
where the $(s,r)$ entry of $\Sigma_{\mathbf x}$ is
\(
\frac{
\operatorname{Cov}(Y_{is},Y_{ir}|\mathbf X_i=\mathbf x)
}{
\PP(\mathbf X=\mathbf x)
},
\) $s,r\in\{1,2\}.$ 
  The continuous mapping
$(\widehat{\PP}(Y_2=1| \mathbf x),\widehat{\PP}(Y_1=1| \mathbf x))\mapsto\Delta\widehat{\PP}(Y=1| \mathbf x)$
is linear, so $\sqrt{n}\bigl(\Delta\widehat{\PP}(Y=1| \mathbf x)-\Delta \PP(Y=1| \mathbf x)\bigr)
  \xrightarrow{d}\mathcal{N}(0,\sigma^2_{\mathbf x})$,
where $\sigma_{\mathbf x}^2
=
\frac{\operatorname{Var}(Y_2-Y_1|\mathbf X=\mathbf x)}{\PP(\mathbf X=\mathbf x)}$. When $\Delta \PP(Y=1| \mathbf x)=0$, we have 
$E[Y_2-Y_1| \mathbf X=\mathbf x]=0$, and therefore
$\operatorname{Var}(Y_2-Y_1| \mathbf X=\mathbf x)
=
\PP(Y_1\neq Y_2| \mathbf X=\mathbf x)>0$ 
by condition (iii) of Theoretical example 2. Hence the limiting normal distribution is
nondegenerate.

Part (b), in addition, uses the symmetry of the limiting normal distribution.

\noindent \textbf{Proof of Proposition  \ref{prop:bcpanel_probest2}.} For $t=1,\ldots,T$, a first order expansion gives
\[
\sqrt n\left(
\Delta\widehat \PP(Y=1| \mathbf x_t)
-
\Delta \PP(Y=1| \mathbf x_t)
\right)
=
\frac{1}{\sqrt n}\sum_{i=1}^n
\frac{
1\{\mathbf X_i=\mathbf x_t\}
\left[
(Y_{i2}-Y_{i1})
-\Delta \PP(Y=1| \mathbf x_t)
\right]
}{
\PP(\mathbf X_i=\mathbf x_t)
}
+o_p(1).
\]
Hence the joint central limit theorem gives a centered multivariate
normal limit. 

For $\mathbf x_t\in\mathcal X_0$,
$\Delta \PP(Y=1| \mathbf x_t)=0$, and
$\sigma_t^2
=
\frac{
\operatorname{Var}(Y_{i2}-Y_{i1}|\mathbf X_i=\mathbf x_t)
}{
P(\mathbf X_i=\mathbf x_t)
}
=
\frac{
\PP(Y_{i1}\neq Y_{i2}|\mathbf X_i=\mathbf x_t)
}{
P(\mathbf X_i=\mathbf x_t)
}
>0$
by condition (iii) of Theoretical example 2. For $t\neq s$,
$1\{\mathbf X_i=\mathbf x_t\}1\{\mathbf X_i=\mathbf x_s\}=0
\qquad\text{a.s.}$, so the corresponding influence functions have zero covariance. Thus, 
$\Sigma=\operatorname{diag}(\sigma_1^2,\ldots,\sigma_T^2)$
and is positive definite. This proves part (a). 

For part (b), the points in
$\mathcal X\setminus\mathcal X_0$ are classified correctly with
probability approaching one by Proposition
\ref{prop:bcpanel_probest1}(a). By part (a), the limiting Gaussian
coordinates at $\mathbf x_1,\ldots,\mathbf x_T$ are independent. Each is centered
and nondegenerate normal, and therefore has probability $1/2$ of
being positive and probability $1/2$ of being negative. Hence, $\PP\left(
\bigcap_{t=1}^T
\{\Delta\widehat \PP(Y=1| \mathbf x_t)\in C_{m_t}\}
\right)
\longrightarrow 2^{-T}$ for
every $(m_1,\ldots,m_T)\in\{1,2\}^T$. 
Combining this with the probability 1 limiting classification of
the points in $\mathcal X\setminus\mathcal X_0$ proves part (b). $\blacksquare$

\vskip 0.05in 

\noindent \textbf{Proof of Proposition \ref{prop:multinom_probest1}.} Part (a) follows from consistency of the estimated conditional
choice probabilities. Namely, every population gap is strictly positive,
and therefore the strict population ordering is preserved with
probability approaching one.

For part (b), alternatives belonging to the same population tie
block have identical conditional choice probabilities. This means that conditional
on the observations with $\mathbf X=\mathbf x$, the multinomial
count vector is therefore invariant under arbitrary permutations
of the alternative labels within each tie block. There are $|B_r|!$ possible strict orderings of the alternatives in
a tie block of size $|B_r|$, so the total number of strict
refinements of the population ordering is $M(\mathbf x)=\prod_{r=1}^R |B_r|!$. By permutation symmetry within each population tie block, all
$M(\mathbf x)$ strict refinements have the same probability.

 For each pair tied in the population, the scaled empirical probability
difference converges to a nondegenerate centered normal variable.
Consequently, the probability of an exact empirical tie tends to zero.
A finite union bound and consistency for the strictly ordered
between-block comparisons imply that the probability of some strict
refinement approaches 1. Since the $M(\mathbf x)$ refinements are
equiprobable by permutation symmetry, each has limiting probability
$1/M(\mathbf x)$.$\blacksquare$

\noindent \textbf{Proof of Proposition \ref{prop:multinom_probest2}.} For each $t$, the influence function of
$\widehat{p}(\mathbf x_t)$ is $\frac{
\mathbf 1\{\mathbf X=\mathbf x_t\}
\bigl(e_Y-p(\mathbf x_t)\bigr)
}{\PP(\mathbf X=\mathbf x_t)}$, where $e_Y$ is the \(J\)-dimensional unit vector with a one in position \(Y\). Therefore, 
$\Sigma_t
=
\frac{1}{\PP(\mathbf X=\mathbf x_t)}
\left[
\operatorname{diag}(p(\mathbf x_t))-p(\mathbf x_t) p(\mathbf x_t)' 
\right]$.

For $s\neq t$, the corresponding influence functions have zero
cross-covariance because the support cells
$\{\mathbf X=\mathbf x_s\}$ and
$\{\mathbf X=\mathbf x_t\}$ are disjoint. This gives the stated
block-diagonal Gaussian limit.

Since every component of $p(\mathbf x_t)$ lies in $(0,1)$,
$\Sigma_t$ has rank $J-1$, with null space generated by
$\mathbf 1_J$, which implies 
$v'\Sigma_tv>0$
for every nonzero $v$ satisfying $\mathbf 1_J'v=0$.

For part (b), Proposition~\ref{prop:multinom_probest1} implies that
each strict refinement at $\mathbf x_t$ has asymptotic probability
$1/M_t$. The Gaussian limits associated with distinct support
points are independent by the block-diagonal covariance above.
The strict population orderings outside $\mathcal X_0$ are selected
correctly with probability approaching one. We can conclude that  the probability
of any complete strict refinement profile converges to
$\prod_{t=1}^T\frac{1}{M_t}$.
$\blacksquare$

\subsection{Verification of conditions (C1)-(C11) for multinomial choice. Application of general theorems to Illustrative design 3.}

Throughout this verification, \(\mathcal A_0\) denotes the identified set of the multinomial rank-order model in Section \ref{sec:theoreticalmultinomial}.

\noindent\textbf{Population maximizers of the block score.} We verify the characterization stated in
Section~\ref{sec:theoreticalmultinomial}. Fix a population
probability vector $p$ with ordered probability tie blocks
$\mathfrak B(p)=(B_1(p),\ldots,B_{R(p)}(p))$.

First, every maximizer of \eqref{eq:block_score} must completely
separate the distinct probability blocks. To see this, take any
ordering of the blocks that extends all strict between-block
comparisons already implied by $u$, and separate the blocks
completely according to that ordering. No existing strict
between-block comparison is lost, while at least one new comparison
is added whenever two blocks were not previously separated.
Consequently, every $N_r(u;p)$ weakly increases and at least one
strictly increases. Since $W$ is strictly increasing, the block
score strictly increases unless the blocks were already completely
separated.

It remains to determine their ordering. For completely separated
blocks, \eqref{eq:block_score} assigns consecutive rank weights to
each probability block. Any ordering that differs from the probability ordering contains two adjacent blocks with the lower-probability block ranked above the higher-probability block. Interchanging these adjacent blocks strictly increases the score. Repeating such interchanges yields the unique maximizing ordering of the distinct probability blocks. Hence, 
\[
\operatorname*{arg\,max}_{u\in\mathbb R^J}
s^{\mathrm{blk}}(p,u)
=
\left\{
u\in\mathbb R^J:
\min_{j\in B_r(p)}u_j>
\max_{k\in B_{r'}(p)}u_k
\quad\text{for every }r<r'
\right\}.
\]
There is no restriction on the ordering of alternatives within a
probability tie block. For
$p=\phi(F(Y|\mathbf x))$, this is precisely
$D^*_{m(\mathbf x)}$ in \eqref{eq:Dm_star_block}.
$\blacksquare$

\smallskip

\noindent \textbf{Verification of (C1)-(C11).}

(C1)  For any $\mathbf{x}_t \in \mathcal{X}_0$,  any coarse set $D^*_{m_t}$ with $m_t \in  M_{as}(\mathbf{x}_t)$ coincides with $D_{m_t}$ as it imposes a strict complete ordering within each tie block while preserving ordering between blocks. It  is contained in the tie-block coarse set $D^*_{m(\mathbf{x}_t)}$, which replaces each tie with $\mathbb{R}$ (``anything can happen instead of this tie''). $D^*_{m(\mathbf{x}_t)}$ obviously  contains each complete ordering within a tie-block as a special case: 
\(
D^*_{m_t}\subset  D^*_{m(\mathbf{x}_t)}\quad\forall \, m_t\in\mathcal M_{as}(\mathbf{x}_t).
\)

(C2) 
Under the block estimator, disjointness of $\{D^*_{m_t}\}_{m_t\in\mathcal M_{as}(\mathbf{x}_t)}$ follows from the fact that two distinct $m_t $ and $m_t'$ from $\mathcal M_{as}(\mathbf{x}_t)$ will impose  strict tie-break regimes within tie-blocks resolving all the ties. Hence, $D^*_{m_t}=D_{m_t}$ and $D^*_{m_t'}=D_{m_t'}$ and $m_t$, $m_t'$ will impose different strict orderings within at least one true tie block in $\mathfrak{B}(\phi(F(Y|\mathbf{x}_t)))$. Hence, $D_{m_t} \cap D_{m_t'}=\varnothing$. 

(C3)  We first verify
\eqref{C3part1a}--\eqref{C3part1b}.

Consider first $\mathbf x_t\in\mathcal X_0$. By
Proposition~\ref{prop:multinom_probest1}, with probability
approaching one, $\phi(\widehat F(Y\mid\mathbf X=\mathbf x_t))
\in
\bigcup_{m\in\mathcal M_{as}(\mathbf x_t)}C_m$.

On an event that the estimated choice probabilities at
$\mathbf x_t$ are strictly ordered and belong to one of the
asymptotically relevant regimes, the block score reduces to the
ordinary multinomial rank score. Hence, for any
$u,v\in
\bigcup_{m\in\mathcal M_{as}(\mathbf x_t)}D_m^*$, the orderings represented by $u$ and $v$ preserve the same strict
ordering between distinct population tie blocks and can differ only
in their ordering of alternatives within population tie blocks.
Therefore,
\begin{multline*}
s^{\mathrm{blk}}
\bigl(\phi(\widehat F(Y|\mathbf X=\mathbf x_t)),u\bigr)
-
s^{\mathrm{blk}}
\bigl(\phi(\widehat F(Y|\mathbf X=\mathbf x_t)),v\bigr)
=\\
\sum_{j=1}^J
\left[
\widehat \PP(Y=j|\mathbf X=\mathbf x_t)
-
\PP(Y=j|\mathbf X=\mathbf x_t)
\right]
\times
\left[
W\left(\sum_{k\neq j}1\{u_j>u_k\}\right)
-
W\left(\sum_{k\neq j}1\{v_j>v_k\}\right)
\right].
\end{multline*}

There are only finitely many strict refinements of the population tie
blocks. By Proposition~\ref{prop:multinom_probest2}(a), $\phi(\widehat F(Y\mid\mathbf X=\mathbf x_t))
-
\phi(F(Y\mid\mathbf X=\mathbf x_t))
=
O_p(n^{-1/2})$, which implies 
\[
\sup_{\substack{
u,v\in
\bigcup_{m\in\mathcal M_{as}(\mathbf x_t)}D_m^*
}}
\left|
s^{\mathrm{blk}}
\bigl(\phi(\widehat F(Y\mid\mathbf X=\mathbf x_t)),u\bigr)
-
s^{\mathrm{blk}}
\bigl(\phi(\widehat F(Y\mid\mathbf X=\mathbf x_t)),v\bigr)
\right|
=
O_p(n^{-1/2})
=
o_p(1).
\]
This verifies \eqref{C3part1a} at
$\mathbf x_t\in\mathcal X_0$.

We next verify \eqref{C3part1b}. Take
$u\in
\bigcup_{m\in\mathcal M_{as}(\mathbf x_t)}D_m^*
$, $u'\notin
\bigcup_{m\in\mathcal M_{as}(\mathbf x_t)}D_m^*$. Along any asymptotically relevant strict sample regime, the sample
block score converges, for a fixed ordering represented by $u$, to $\sum_{j=1}^J
\PP(Y=j|\mathbf X=\mathbf x_t)
W\left(
\sum_{k\neq j}1\{u_j>u_k\}
\right)$.
When the population choice probabilities contain ties,
this limiting value need not equal the block score evaluated directly at the
population probability vector. Every $u$ in the union of the asymptotically relevant regions
preserves the strict ordering between distinct population tie blocks
and imposes a complete strict ordering within each population tie
block. All such strict refinements give the same limiting score,
because alternatives within a population tie block have equal
population choice probabilities.

An ordering represented by $u'$ outside this union either fails to
preserve the strict ordering between two distinct population tie
blocks or fails to impose a complete strict ordering within some
population tie block. In the first case, the rearrangement argument,
together with the strict difference between the corresponding
population choice probabilities and the strict monotonicity of $W$,
shows that restoring the population ordering strictly increases the
limiting score. In the second case, breaking an equality within a
population tie block strictly increases the limiting score. Indeed,
the alternatives involved have the same strictly positive population
choice probability, while resolving the equality gives at least one
of them a strictly higher rank and hence a strictly larger value of
$W$. Thus, every weak ordering outside $\bigcup_{m\in\mathcal M_{as}(\mathbf x_t)}D_m^*$ has a strictly smaller limiting score than every asymptotically
relevant strict refinement. Since there are only finitely many weak
orderings of $J$ alternatives, the minimum of these limiting score
gaps is strictly positive. Consistency of the estimated choice
probabilities therefore implies that there exists $\delta_t>0$ such
that
\[
\PP\left(
\inf_{\substack{
u\in\bigcup_{m\in\mathcal M_{as}(\mathbf x_t)}D_m^*\\
u'\notin\bigcup_{m\in\mathcal M_{as}(\mathbf x_t)}D_m^*
}}
\left[
s^{\mathrm{blk}}
\bigl(\phi(\widehat F(Y|\mathbf X=\mathbf x_t)),u\bigr)
-
s^{\mathrm{blk}}
\bigl(\phi(\widehat F(Y|\mathbf X=\mathbf x_t)),u'\bigr)
\right]
>\delta_t
\right)\to1.
\]
This verifies \eqref{C3part1b} at
$\mathbf x_t\in\mathcal X_0$.

For $\mathbf x\notin\mathcal X_0$,
$\mathcal M_{as}(\mathbf x)=\{m(\mathbf x)\}$ and all population
choice probabilities are strictly ordered. The score is constant on
$D^*_{m(\mathbf x)}$, so \eqref{C3part1a} holds immediately.
Moreover, the rearrangement argument implies that every weak ordering
outside $D^*_{m(\mathbf x)}$ has a strictly smaller population score.
Since there are finitely many weak orderings, the minimum gap is
strictly positive, and consistency of the estimated choice
probabilities implies \eqref{C3part1b}.

Condition \eqref{C3part2} follows directly from
Proposition~\ref{prop:multinom_probest2}(b). Condition~\eqref{eq:region_constancy} holds immediately, since each candidate region $A_\ell$ fixes a strict ordering of the utility indices at every support point. It remains to verify
\eqref{C3part3}--\eqref{C3part3unique}.

For $\ell,\ell'\in\{1,\ldots,L\}$ and
$t\in\{1,\ldots,T\}$, let the $j$th component of
$v(m_{\ell,t},m_{\ell',t},t)$ be
\[
W\left(
\sum_{k\neq j}
1\{(u_{\ell,t})_j>(u_{\ell,t})_k\}
\right)
-
W\left(
\sum_{k\neq j}
1\{(u_{\ell',t})_j>(u_{\ell',t})_k\}
\right),
\]
where
$u_{\ell,t}\in D^*_{m_{\ell,t}}=D_{m_{\ell,t}}$
and
$u_{\ell',t}\in D^*_{m_{\ell',t}}=D_{m_{\ell',t}}$.
This vector does not depend on the particular representatives chosen
within these regions. Since $m_{\ell,t}$ and $m_{\ell',t}$ differ only by permutations
within population tie blocks, $v(m_{\ell,t},m_{\ell',t},t)'
\phi(F(Y|\mathbf X=\mathbf x_t))
=
0$. With probability approaching one, every support cell is
observed and the empirical choice probabilities at every
support point are strictly ordered and refine their
population ordering. This follows from
Proposition~\ref{prop:multinom_probest1} and finiteness of
the support. On this event, the block score coincides
with the ordinary rank score, and therefore 
\begin{multline*}
s^{\mathrm{blk}}
\bigl(\phi(\widehat F(Y|\mathbf X=\mathbf x_t)),u_{\ell,t}\bigr)
-
s^{\mathrm{blk}}
\bigl(\phi(\widehat F(Y|\mathbf X=\mathbf x_t)),u_{\ell',t}\bigr)
\\
=
v(m_{\ell,t},m_{\ell',t},t)'
\left[
\phi(\widehat F(Y|\mathbf X=\mathbf x_t))
-
\phi(F(Y|\mathbf X=\mathbf x_t))
\right].
\end{multline*}

On the same event, the empirical ordering at every
$\mathbf x\notin\mathcal X_0$ agrees with its strict
population ordering. The objective contributions from
these support points are therefore identical across
$A_1,\ldots,A_L$. Hence, on this event, for every
$\ell=2,\ldots,L$, 
\[
\sqrt n(\widehat Q_\ell-\widehat Q_1)
=
\sum_{t=1}^T
\widehat \PP(\mathbf X=\mathbf x_t)
v(m_{\ell,t},m_{1,t},t)'
\sqrt n
\left[
\phi(\widehat F(Y|\mathbf X=\mathbf x_t))
-
\phi(F(Y|\mathbf X=\mathbf x_t))
\right].
\]

Proposition~\ref{prop:multinom_probest2}(a) implies that $\sqrt n
(
\widehat Q_2-\widehat Q_1, \ldots, 
\widehat Q_L-\widehat Q_1
)^{\top}
\stackrel{d}{\longrightarrow}
(
V_2, \ldots 
V_L)^{\top}$ for a centered multivariate Gaussian vector
$(V_2,\ldots,V_L)'$. Thus, \eqref{C3part3} holds with
$r_n=\sqrt n$ and, as in Theorem~\ref{th:general1}, we set $V_1=0$.

It remains to verify \eqref{C3part3unique}. For any
$\ell\neq\ell'$, the limiting difference
$V_\ell-V_{\ell'}$ is centered normal with variance $\sum_{t=1}^T
P(\mathbf X=\mathbf x_t)^2
v(m_{\ell,t},m_{\ell',t},t)'
\Sigma_t
v(m_{\ell,t},m_{\ell',t},t)$. The profiles used in constructing the distinct regions
$A_\ell$ and $A_{\ell'}$ differ at at least one
$\mathbf x_t$. For such a $t$, strict monotonicity of $W$ implies $v(m_{\ell,t},m_{\ell',t},t)\neq0$. Moreover,
$\mathbf 1_J'
v(m_{\ell,t},m_{\ell',t},t)=0$,
as the vector is the difference between two permutations of the
same collection of rank weights.

By Proposition~\ref{prop:multinom_probest2}(a), $\Sigma_t$ is
positive definite on the zero-sum subspace. Therefore, at least one term
in the variance above is strictly positive, and therefore $\operatorname{Var}(V_\ell-V_{\ell'})>0$. Since $V_\ell-V_{\ell'}$ is Gaussian, this fact implies $\PP(V_\ell=V_{\ell'})=0$
for every $\ell\neq\ell'$. Since there are only finitely many pairs,
$\PP\left(
\left|
\arg\max_{\ell\in\{1,\ldots,L\}}V_\ell
\right|=1
\right)=1$. This verifies
\eqref{C3part3}--\eqref{C3part3unique}.

(C4)  Fix $\ell$ and suppose $a\in\mathcal A$ satisfies the closed constraints:
\[
\psi(\mathbf{x},a)\in \overline D_{m(\mathbf{x})}\ \ \forall \mathbf{x}\in\mathcal X\setminus\mathcal X_0,
\qquad
\psi(\mathbf{x}_t,a)\in \overline D^*_{m_{\ell,t}}\ \ \forall \mathbf{x}_t\in\mathcal X_{0,\ell}.
\]
Since $A_{\ell}$ we consider is non-empty, pick any $b\in A_\ell$  and define
$a_\lambda=(1-\lambda)a+\lambda b$. 
Since $b$ satisfies the open (strict inequalities) constraints defining $A_{\ell}$ and $a$ satisfies analogous constraints but with some of them weak, we conclude that  $a_{\lambda}$ satisfies all strict inequality constraints defining $A_{\ell}$. We can choose $\lambda$ to make $a_{\lambda}$ arbitrarily close to $a$. 

(C5) Fix $a_0\in\mathcal A_0$ and $b\in A_\ell$. Convex interpolation $(1-\lambda)a_0+\lambda b$ belongs to $A_\ell$ for $0<\lambda<1$ and converges to $a_0$.
For the outside approximation, take the interior identified parameter
$\alpha$ and perturb it in a  direction preserving the normalization such that it
reverses one within-block contrast imposed by $A_\ell$.
Such a direction exists as $A_\ell$ is a feasible strict
refinement. For a sufficiently small perturbation, the resulting point
$c$ remains in $\mathcal A$ and preserves all population between-block
strict inequalities. Interpolating from $a_0$ toward $c$ gives points
outside $A_\ell$ arbitrarily close to $a_0$, while preserving the
nonambiguous restrictions. Thus,  (C5) holds.

(C6) obviously holds. For $\ell$ and $\mathbf{x}_t\in\mathcal X_0$,
The sharp identified restrictions at $\mathbf{x}_t$ include equalities within at least one
population tie block. For illustrational convenience, consider the case of just one tie block with elements $j_1$,...,$j_d$ in it: that is, $\psi_{j_1}(\mathbf{x}_t,a)=...=\psi_{j_d}(\mathbf{x}_t,a)$ in that block. Without a loss of generality  this equality is resolved in $D_{m_{\ell,t}}=D_{m_{\ell,t}}^*$ as  $\psi_{j_1}(\mathbf{x}_t,a)>...>\psi_{j_d}(\mathbf{x}_t,a)$ holding
all between-block inequalities fixed. Since the set of equalities of $\psi_{j_1}(\mathbf{x}_t,a)=...=\psi_{j_d}(\mathbf{x}_t,a)$ is clearly at the  boundary of the polyhedron $\psi_{j_1}(\mathbf{x}_t,a)>...>\psi_{j_d}(\mathbf{x}_t,a)$ (each point with equalities will contain points with strict inequalities as well as points with such strict inequalities violated in its arbitrary neighborhood) and all the between-block inequalities are fixed, this guarantees that $D_{m(\mathbf{x}_t)} \subseteq \partial D^*_{m_{\ell,t}} $.

(C7) The difference between $ D^*_{m(\mathbf{x})}$ and $\overline D^*_{m(\mathbf{x})}$ for any $\mathbf{x} \in \mathcal{X}$ is that in the latter the inequalities between tie-blocks are now weak  while in the former they are strict. 

Take  $a\in\mathcal A^{*,\bullet}$ such that $\psi(\mathbf{x},a)\in\overline D^*_{m(\mathbf{x})}\setminus D^*_{m(\mathbf{x})}$ for some $x \in \mathcal{X}$. This means  
that at least one between-block inequality holds with equality:
$\min_{j\in B_r(\mathbf{x})}\psi_j(\mathbf{x},a)=\max_{k\in B_{r'}(\mathbf{x})}\psi_k(\mathbf{x},a)$ for some $r<r'$.
Since $\mathcal A^*$ is nonempty, pick $b\in\mathcal A^*$ so that all between-block gaps are strict at
$b$. Then $a_\lambda=(1-\lambda)a+\lambda b$ satisfies the strict between-block inequalities for all
sufficiently small $\lambda>0$, and $a_\lambda\to a$. Hence any neighborhood of $a$ intersects
$\{a:\psi(\mathbf{x},a)\in D^*_{m(\mathbf{x})}\ \forall \mathbf{x}\}$, proving (C7).

(C8) Suppose \(\mathcal X_0\ne\varnothing\). By Assumption \ref{assn:assnexitense}, choose \(b\in A_{\ell_0}\) for some nonempty candidate region. Thus, \(b\) satisfies all strict between-block inequalities and has distinct utility indices within each population tie block at every support point.
Fix \(a\in\mathcal A^{*,\bullet}\), and define $a_\lambda=(1-\lambda)a+\lambda b$, $0<\lambda<1$. Convexity ensures \(a_\lambda\in\mathcal A\). Every between-block contrast is weakly positive at \(a\) and strictly positive at \(b\), so it is strictly positive at \(a_\lambda\).
For each within-block contrast $g_{x,jk}(c)=\psi_j(x,c)-\psi_k(x,c)$,
linearity gives $g_{x,jk}(a_\lambda)
=(1-\lambda)g_{x,jk}(a)+\lambda g_{x,jk}(b)$.
If \(g_{x,jk}(a)=0\), this contrast is nonzero for every \(\lambda>0\), since \(g_{x,jk}(b)\ne0\). If \(g_{x,jk}(a)\ne0\), it retains its sign for all sufficiently small \(\lambda>0\).
There are finitely many contrasts. Consequently, for all sufficiently small positive \(\lambda\), \(a_\lambda\) belongs to one fixed strict refinement profile \(\mathbf m\). This profile is asymptotically relevant, and its solution set is nonempty because it contains these \(a_\lambda\). Taking \(\lambda\downarrow0\) shows that \(a\) satisfies its closed constraints.

When \(\mathcal X_0=\varnothing\), the empty profile has solution set \(\mathcal A^*\), and the condition follows directly.

(C9) The first part follows by the same
convex combination argument as in the verification of (C4), using
the true parameter $\alpha\in\mathcal A_0$. Condition \(\overline{D^*_{m(x)}}=\overline{D_{m(x)}}\) outside \(\mathcal X_0\) obviously holds. 

The second part of (C9) requires each $\mathcal{F}_{\ell}$ is singleton. Suppose $m\neq m'$ for different profile $m=(m_1,...,m_T)$ and $m'=(m_1',...,m_T')$ with $m_t, m_t' \in \mathcal{M}_{as}(\mathbf{x}_t)$, $t=1,\ldots,T$. Since $m\neq m'$, there exists $t$ such that $m_t\neq m'_t$, implying that for some pair $j,k$ from the same tie block ($j\neq k)$ their  tie  is resolved in $A(m)$ as $\{\psi_j>\psi_k\}$ whereas the same tie is resolved in  $A(m')$ as $\{\psi_j<\psi_k\}$. 
Potentially one or both sets being empty. If both are nonempty, then necessarily $A(m)\cap A(m')=\emptyset$. Thus, no two distinct profiles can generate the same nonempty region, and each $\mathcal F_\ell$ is a singleton.

For each $\mathbf{x}_t\in\mathcal X_0$, the size of $\mathcal M_{as}(\mathbf{x}_t)$ is the product of factorials of the number of ties within each block induced by the population partition $\mathfrak B(\phi(F(Y|\mathbf{x}_t)))$. E.g. in the case of just one tie block $B_r(\mathbf{x}_t)$ with at least two elements  the size of  $\mathcal M_{as}(\mathbf{x}_t)$ is $|B_r(\mathbf{x}_t)|!$. Take any $j,k \in B_r(\mathbf{x}_t) $ for $j\neq k$. If in the resolution of ties we want to keep  $\psi_j(\mathbf{x}_t,a)>\psi_k(\mathbf{x}_t,a)$, then this would leave us with  $|B_r(\mathbf{x}_t)|!/2$ strict orderings. Thus,  each subset $\mathcal J_t\subseteq\mathcal M_{as}(\mathbf{x}_t)$ such that
$|\mathcal J_t|\ge \lfloor |\mathcal M_{as}(\mathbf{x}_t)|/2\rfloor+1$, will necessarily include cases of \textit{both}  $\psi_j(\mathbf{x}_t,a)>\psi_k(\mathbf{x}_t,a)$ and $\psi_j(\mathbf{x}_t,a)<\psi_k(\mathbf{x}_t,a)$ for any such $j,k$, $j \neq k$. The intersection of the closures therefore imposes $\psi_j(\mathbf{x}_t,a)=\psi_k(\mathbf{x}_t,a)$. This is  true for any pair $j,k$ from a tie block. 

\noindent{(C10).} Recall that condition (C10) requires that, for every $\ell$ and every
$t$, if $m=(m_t,m_{-t})\in\mathcal F_\ell$, then every alternative way
$\widetilde m_t\in\mathcal M_{as}(\mathbf x_t)$,
$\widetilde m_t\neq m_t$, of resolving the population ties at
$\mathbf x_t$, while keeping the resolutions $m_{-t}$ at all other
points in $\mathcal X_0$ fixed, generates a distinct nonempty region: $A(\widetilde m_t,m_{-t})\neq\varnothing$, 
$A(\widetilde m_t,m_{-t})\neq A(m_t,m_{-t})$.

We need to show that every possible way of breaking the population ties can be generated by a small change in the coefficients, while respecting the normalization. 

Within each population tie block, choose a reference alternative and express each other alternative’s index difference from the reference in terms of the free coefficients after normalization. Stack the coefficient vectors of these differences as rows of a matrix, across all blocks and covariate values. Full row rank means that these vectors are linearly independent. It ensures that a change in the free coefficients can produce any prescribed values of the differences, including values that realize any desired ordering within each block. Starting from \(\alpha\), where these differences are zero, a sufficiently small step in that direction breaks the ties as desired while remaining in the parameter space and preserving all strict between-block comparisons. Hence every complete refinement profile is feasible, establishing (C10).

(C11) and global reversal. When the full-rank condition above fails, some complete
strict refinement profiles may be infeasible. There is nevertheless
a useful symmetry of the  collection of feasible profiles.

Let $\mathcal F
=
\{m:A(m)\neq\varnothing\}$. Fix $m\in\mathcal F$ and choose $a\in A(m)$. Let
$\alpha_0\in\mathcal A_0$ lie in the relative interior of the
normalized parameter space and satisfy all strict between-block
restrictions. For sufficiently small $\lambda>0$, $a_\lambda
=
\alpha_0+\lambda(a-\alpha_0)$ belongs to $A(m)$. Consider its reflection $\widetilde a_\lambda
=
2\alpha_0-a_\lambda
=
\alpha_0-\lambda(a-\alpha_0)$. All strict between-block restrictions remain satisfied for
sufficiently small $\lambda$. Since every within-block utility
difference is zero at $\alpha_0$, reflection changes the sign of
every such difference. Thus, within each population tie block of size $d$, $j_1>\cdots>j_d$ is replaced by $j_d>\cdots>j_1$. Denote the resulting globally reversed profile by $fl(m)$. Then $m\in\mathcal F
\; \Longrightarrow\; 
fl(m)\in\mathcal F$, and $fl(fl(m))=m$. Thus, global reversal pairs the feasible profiles. By (C2), distinct feasible profiles generate distinct nonempty
regions, so global reversal also pairs the corresponding feasible
regions. Moreover, for every such pair,
$\overline A_m\cap\overline A_{fl(m)}
=
\overline{\mathcal A}_0$.
Indeed, within a population tie block, if $\overline A_m$ imposes
$\psi_{j_1}\geq\psi_{j_2}\geq\cdots\geq\psi_{j_d}$,
then the closure of the reversed region imposes
$\psi_{j_d}\geq\psi_{j_{d-1}}\geq\cdots\geq\psi_{j_1}$.
Their intersection therefore requires
$\psi_{j_1}=\cdots=\psi_{j_d}$.
This applies to every population tie block at every
$\mathbf x_t\in\mathcal X_0$. The restrictions between distinct population blocks are the same in the two regions and coincide, after taking closures, with the corresponding restrictions in $\overline{\mathcal A}_0$. Hence, global reversal verifies (C11), including when a population tie
block contains three or more alternatives.

\vskip 0.05in

\noindent \textbf{Illustrative Design 3 (continued).} 

For the DGP considered above, $\mathcal X_0
=
\{\mathbf x_1,\mathbf x_2,\mathbf x_3,\mathbf x_4\}$,
where
$\mathbf x_1=(1,0,0)^\top$, 
$\mathbf x_2=(1,0,1)^\top$, 
$\mathbf x_3=(0,1,0)^\top$, 
$\mathbf x_4=(0,1,1)^\top$. At each of these four support points, alternatives 1 and 2 have equal
population choice probabilities.  The two remaining support points,
$(1,1,0)^\top$ and $(1,1,1)^\top$, have strictly ordered population
choice probabilities.

The block maximum score criterion drops the within-block restriction
between alternatives 1 and 2 at the four ambiguous support points,
while retaining the strict ordering between the tie block and
alternative 3. Eliminating redundancies in strict inequalities gives 
\[
\mathcal A^*
=
\left\{
(a_{11},a_{22},a_{31},a_{32})\in\mathcal A:
\begin{array}{l}
0<a_{31}, \; \; 
a_{31}+a_{32}<0, \; \; a_{22}>a_{31}\\
a_{31}<a_{11}<a_{31}+1, \; \; 
a_{31}+a_{32}<a_{11}-1.
\end{array}
\right\},
\]
which has nonempty four-dimensional relative interior
when $\mathcal A$ is sufficiently large, whereas the identified set
$\mathcal A_0$ has affine dimension two.

There are $L=4$ nonempty regions (if $\mathcal A$ is sufficiently large) generated by the
possible resolutions of the population ties: 
\[
A_{++}
=
\mathcal A^*
\cap
\{a_{11}>1,\;a_{11}>a_{22}\},
\quad 
A_{+-}
=
\mathcal A^*
\cap
\{a_{11}>1,\;a_{11}<a_{22}\},
\]
\[
A_{-+}
=
\mathcal A^*
\cap
\{a_{11}<1,\;a_{11}>a_{22}\},
\quad 
A_{--}
=
\mathcal A^*
\cap
\{a_{11}<1,\;a_{11}<a_{22}\}.
\]
The regions have a natural reversal pairing as described in condition (C11) and verification of this condition for multinomial choice. Pair
$A_{++}$ with $A_{--}$, and  $A_{+-}$ with $A_{-+}$:   $\overline A_{++}\cap\overline A_{--}
=\overline{\mathcal A}_0$, \quad 
$\overline A_{+-}\cap\overline A_{-+}
=
\overline{\mathcal A}_0$. This verifies condition~(C11).

Condition~(C10), by contrast, does not hold in this design.  For
example, changing the tie resolution at $\mathbf x_1$ while holding
the resolution at $\mathbf x_2$ fixed would require simultaneously
$a_{11}>1$ and $a_{11}<1$, and therefore generates an empty region. Thus, Illustrative Design~3 provides a concrete case in which the
local profile condition~(C10) fails as some formal
strict refinement profiles are infeasible, while the global pairing
condition~(C11) continues to hold.

In Illustrative Design~3, all four limiting regions are
equiprobable for every strictly increasing $W$, even
though (C10) fails. To see this, consider the event that
all support cells are observed and all empirical
probability orderings strictly refine their population
orderings. Its probability tends to one, and on this
event the block score coincides with the ordinary
rank score.

Changing the first sign from $-$ to $+$, while holding
the second sign fixed, changes the sample objective by
\begin{align*}
&\widehat\PP(\mathbf X=\mathbf x_1)
[W(1)-W(0)]
\left[
\widehat\PP(Y=1\mid\mathbf x_1)
-\widehat\PP(Y=2\mid\mathbf x_1)
\right]\\
&\quad+
\widehat\PP(\mathbf X=\mathbf x_2)
[W(2)-W(1)]
\left[
\widehat\PP(Y=1\mid\mathbf x_2)
-\widehat\PP(Y=2\mid\mathbf x_2)
\right].
\end{align*}
The different rank weight differences arise because
alternatives 1 and 2 occupy the two lowest rank positions
at $\mathbf x_1$ and the two highest at $\mathbf x_2$.
Similarly, changing the second sign from $-$ to $+$,
while holding the first sign fixed, changes the
sample objective by
\[
[W(2)-W(1)]
\sum_{t=3}^{4}
\widehat\PP(\mathbf X=\mathbf x_t)
\left[
\widehat\PP(Y=1\mid\mathbf x_t)
-\widehat\PP(Y=2\mid\mathbf x_t)
\right].
\]
The contributions from the two nonambiguous support
points are identical across all four regions.

By Proposition~\ref{prop:multinom_probest2}, multiplying
these two score differences by $\sqrt n$ gives a joint
limit consisting of independent, nondegenerate centered
normal variables. Independence follows because the two
differences involve disjoint groups of support cells.
Nondegeneracy follows from positive cell and choice
probabilities and strict monotonicity of $W$.

Since all four regions are nonempty, the maximizing
region selects each sign according to whether its
corresponding score difference is positive or negative.
Each of the four sign combinations therefore has
limiting probability one quarter:
\[
\pi_{++}=\pi_{+-}=\pi_{-+}=\pi_{--}=\frac14.
\]

Every point of $\overline{\mathcal A}_0$ belongs to the
closures of all four regions. By the reversal pairing
established above, a point outside
$\overline{\mathcal A}_0$ belongs to at most one closure
from each pair, and hence to at most two of the four
closures. Consequently,
\[
p_\infty(u)=1
\quad\text{for }u\in\overline{\mathcal A}_0,
\qquad
p_\infty(u)\leq\frac12
\quad\text{for }u\notin\overline{\mathcal A}_0.
\]
Thus condition~\eqref{eq:membership_separation} holds in
Illustrative Design~3 for every strictly increasing $W$.

\subsection{Proof of Theorem \ref{th:estimatorafterQRSE} and discussion of why consistency holds in Theoretical examples 2 and 3.} 
\label{app:consistency}

\textbf{Proof of Theorem \ref{th:estimatorafterQRSE}.}
Suppose first that $\mathcal X_0\neq\varnothing$. Let
$
p_n(u)=\Pr\{u\in\overline{\widehat A}\}$, 
$
\pi_{\ell,n}=\Pr(\widehat A=A_\ell),
$
and let
$
\varepsilon_n
=
\Pr\!\left(
\widehat A\notin\{A_1,\ldots,A_L\}
\right).
$ By \eqref{eq:finite_region_approx} and
\eqref{eq:region_prob_convergence},
$\varepsilon_n\to0$ and
$\pi_{\ell,n}\to\pi_\ell$ for every $\ell=1,\ldots,L$.

The representation through a finite number of regions  implies $\sup_{u\in\mathcal A}
\left|
p_n(u)-p_\infty(u)
\right|
\leq
\sum_{\ell=1}^L
\left|
\pi_{\ell,n}-\pi_\ell
\right|
+\varepsilon_n
\rightarrow0$.

Fix $\tau\in(1/2,1)$ and choose
$0<\varepsilon<
\min\left\{
1-\tau,\,
\tau-\frac12
\right\}.$ 
By the uniform convergence above, for all sufficiently large $n$,
$\sup_{u\in\mathcal A}
|p_n(u)-p_\infty(u)|<\varepsilon$. Hence, by \eqref{eq:membership_separation}, $p_n(u)>1-\varepsilon>\tau$ for every  $u\in\overline{\mathcal A_0}$, whereas $p_n(u)<\frac12+\varepsilon<\tau$ for every  $u\notin\overline{\mathcal A_0}$. These inequalities imply $q_\tau(\overline{\widehat A})=\overline{\mathcal A_0}$ for all sufficiently large $n$. In particular, $d_H\!\left(
q_\tau(\overline{\widehat A}),
\overline{\mathcal A_0}
\right)\rightarrow0$.

When $\mathcal X_0=\varnothing$, 
$\sup_{u\in\mathcal A}
\left|
p_n(u)-\mathbf 1\{u\in\overline{\mathcal A_0}\}
\right|
\leq
\Pr\!\left(
\overline{\widehat A}\neq\overline{\mathcal A_0}
\right)
\rightarrow0. $
The same threshold argument therefore gives the stated Hausdorff
convergence for every $\tau\in(0,1)$. $\blacksquare$

\bigskip 

\textbf{Discussion of consistency for Theoretical examples 2 and 3} 

\textbf{Theoretical example 2.} The same probability
argument as in Theorem~\ref{th:MSgeneral2} applies after replacing
$x_t$ by the index difference $\Delta x_t=x_{t2}-x_{t1}$ and replacing $\widehat \PP(Y=1|x_t)-\gamma$ by $\Delta\widehat \PP(Y=1|\mathbf x_t)
=
\widehat \PP(Y_2=1|\mathbf x_t)
-
\widehat \PP(Y_1=1|\mathbf x_t)$.

Let $Z=(Z_1,\ldots,Z_T)'$ denote the centered Gaussian limit in
Proposition~\ref{prop:bcpanel_probest2}(a). Take
$u\notin\overline{\mathcal A}_0$. If $u$ belongs to none of the
closures of the limiting regions, then $p_\infty(u)=0$. Otherwise, suppose that $u$ belongs to the closure of a limiting
region. For every such region, the same feasible counterpart
construction as in the proof of Theorem~\ref{th:MSgeneral2}
reverses precisely the ambiguous inequalities for which
$\Delta x_t'u\neq0$ and leaves unchanged those for which
$\Delta x_t'u=0$. The corresponding limiting score difference is
\[
\sum_{\{t:\Delta x_t'u\neq0\}}
P(\mathbf X=\mathbf x_t)
Z_t\operatorname{sgn}(\Delta x_t'u).
\]
This expression depends on $u$ but not on which limiting region whose
closure contains $u$ is considered. Hence, if the selected limiting
region has $u$ in its closure, the displayed sum must be nonnegative.

Since $u\notin\overline{\mathcal A}_0$ and belongs to the closure of
at least one limiting region, there is at least one $t$ such that
$\Delta x_t'u\neq0$. By
Proposition~\ref{prop:bcpanel_probest2}(a), the displayed sum is
therefore a nondegenerate centered normal random variable. Hence, $p_\infty(u)\leq\frac12$
 for $u\notin\overline{\mathcal A}_0$. Since every $u\in\overline{\mathcal A}_0$ belongs to the closure of
every limiting region, $p_\infty(u)=1$
 for $u\in\overline{\mathcal A}_0$.

If (C10) additionally holds, all $2^T$ sign profiles are feasible.
As in the cross-sectional binary-choice model, distinct profiles
generate distinct nonempty regions, and
Proposition~\ref{prop:bcpanel_probest2}(b) gives probability
$2^{-T}$ to every profile. Hence $L=2^T$ and the limiting regions
are equiprobable.

When $\mathcal X_0=\varnothing$, the coarse and exhaustive sets
differ at most at boundary points, so $\overline{\mathcal A^*}
=
\overline{\mathcal A_0}$, and the second part of
Theorem~\ref{th:estimatorafterQRSE} applies.

\textbf{Theoretical example 3.} The multinomial case differs from the binary examples because a
population tie block may contain more than two alternatives.  Let $\mathfrak B(\phi(F(Y\mid\mathbf x_t)))
=
\bigl(B_{t1},\ldots,B_{tR_t}\bigr)$ denote the population tie blocks at
$\mathbf x_t\in\mathcal X_0$, and write
$d_{tr}=|B_{tr}|$.  The number of strict refinements of the
population ordering at $\mathbf x_t$ is $M_t
=
\prod_{r=1}^{R_t}d_{tr}!$.

By permutation symmetry of the multinomial cell counts within each
population tie block, every specified strict refinement has
asymptotic probability $1/M_t$.  Consequently, every complete
profile $(m_1,\ldots,m_T)
\in
\times_{t=1}^T\mathcal M_{as}(\mathbf x_t)$ has asymptotic probability $\prod_{t=1}^T\frac{1}{M_t}$.

Suppose first that (C10) holds.  Then every such profile produces
a nonempty distinct region.  Hence, all limiting regions have probability $1/(\prod_{t=1}^T M_t)$. Now take
$u\notin\overline{\mathcal A_0}$.  If
$u\notin\overline{\mathcal A^*}$, then
$p_\infty(u)=0$.  Otherwise, $u$ satisfies the between-block
restrictions but violates at least one equality within a population
tie block.  Thus, for some $\mathbf x_t$ and some
$j,k\in B_{tr}$, $\psi_j(\mathbf x_t,u)
\neq
\psi_k(\mathbf x_t,u)$. Among the $d_{tr}!$ strict rankings of this block, exactly one half
rank $j$ above $k$ and one half rank $k$ above $j$.  It follows
that at least one half of all complete profiles impose an ordering
incompatible with $u$.  Therefore, $p_\infty(u)\leq\frac12$. As before, $p_\infty(u)=1$
for 
$u\in\overline{\mathcal A_0}$. Hence,  \eqref{eq:membership_separation} holds under (C10), and
Theorem~\ref{th:estimatorafterQRSE} gives consistency for every
$\tau\in(1/2,1)$.

When (C10) fails, some complete strict refinement profiles are infeasible. The empirical profiles remain asymptotically equiprobable by
Proposition~\ref{prop:multinom_probest2}, but the selected feasible
regions do not necessarily have equal probabilities. Global reflection pairs feasible regions (for $m$ and $fl(m)$) without
by itself establishing equality of their selection probabilities.

For a general rank function $W$, no further probability symmetry
follows from this geometric fact alone. For the special case
$W(q)=q$, however, the limiting score is equivariant under global
reversal, which is sufficient to obtain pairwise equality of the
limiting region probabilities. To see this, let $Z$ denote the centered Gaussian
limit of the empirical probability deviations at the population
ties and let $V_\ell(Z)$ denote the corresponding first-order score
for region $A_\ell$.  Within a population tie block occupying rank
positions $q=\underline q,\ldots,\overline q$, reversal replaces
the rank $q$ by $\underline{q}+\overline q-q$. Since $W(q)=q$, replacing $Z$ by $-Z$ and simultaneously reversing
the ranking gives $V_{fl(\ell)}(-Z)
=
V_\ell(Z)+C(Z)$, where $C(Z)$ does not depend on $\ell$.  Hence the maximizing region
under $-Z$ is the global reversal of the maximizing region under
$Z$.  Since $Z\stackrel{d}{=}-Z$, the two regions in every reversal pair have the same limiting
selection probability: $\pi_{fl(\ell)}=\pi_\ell$. Linearity of $W$ is sufficient but not necessary for this
reversal identity. If every population tie block contains
at most two alternatives, the same identity holds for
every strictly increasing $W$, as reversing a
two-element block simply exchanges its two rank weights. We now show directly how this pairwise probability symmetry implies
the required membership probability separation. Recall that
$p_\infty(u)
=
\sum_{\ell=1}^L
\pi_\ell\mathbf 1\{u\in\overline A_\ell\}$. If $u\notin\overline{\mathcal A^*}$, then $u$ belongs to none of the
limiting regions and hence $p_\infty(u)=0$. Now suppose that $u\in
\overline{\mathcal A^*}
\setminus
\overline{\mathcal A}_0$.
By (C11), $\overline A_\ell\cap\overline A_{fl(\ell)}
=
\overline{\mathcal A}_0$, so for every reversal pair,
$\mathbf 1\{u\in\overline A_\ell\}
+
\mathbf 1\{u\in\overline A_{fl(\ell)}\}
\leq1$. Using also $\pi_{fl(\ell)}=\pi_\ell$, we obtain
\begin{multline*}
p_\infty(u)
=
\sum_{\{\ell,fl(\ell)\}}
\left[
\pi_\ell\mathbf 1\{u\in\overline A_\ell\}
+
\pi_{fl(\ell)}
\mathbf 1\{u\in\overline A_{fl(\ell)}\}
\right]
\leq
\sum_{\{\ell,fl(\ell)\}}\pi_\ell
=
\frac12
\sum_{\{\ell,fl(\ell)\}}
\left(\pi_\ell+\pi_{fl(\ell)}\right)
=
\frac12.
\end{multline*}

Finally, every $u\in\overline{\mathcal A}_0$ belongs to the closure
of every limiting region, so $p_\infty(u)=1$. Therefore,
\[
p_\infty(u)=1
\quad\text{for }u\in\overline{\mathcal A}_0,
\qquad
p_\infty(u)\leq\frac12
\quad\text{for }u\notin\overline{\mathcal A}_0.
\]
Hence \eqref{eq:membership_separation} holds.

Thus, for the multinomial block score estimator,
\eqref{eq:membership_separation} holds for a general strictly
increasing $W$ under (C10). When (C10) fails, it continues to hold
for $W(q)=q$, because global reversal supplies the geometric pairing
in (C11), while score equivariance yields equal limiting selection
probabilities within each reversal pair.

\subsection{Proof of Lemma \ref{lemmaLdriftasymp}, Theorem  \ref{th:MSdrift} and Theorem \ref{th:QRSEdrift}}

\textbf{Proof of Lemma \ref{lemmaLdriftasymp}.} Fix an admissible $h$ and $x_t\in\mathcal X_0$.
Since $x_t'\alpha=0$ and $F_{u|X}(0|x_t)=\gamma$,
differentiability at zero gives
\[
\PP_{\alpha_n(h)}(Y=1|x_t)-\gamma
=
F_{u|X}(x_t'h/\sqrt n|x_t)-F_{u|X}(0|x_t)
=
\frac{\mu_t(h)}{\sqrt n}+o(n^{-1/2}).
\] 

To establish the joint convergence, fix an admissible \(h\), and for each 
\(t=1,\ldots,T\), define the empirical cell frequencies and conditional probabilities by
$\frac1n\sum_{i=1}^n1\{X_i=x_t\}$ and 
$\widehat \PP(Y=1| x_t)$, respectively. Since the covariate distribution is fixed along the local path,
\(\frac1n\sum_{i=1}^n1\{X_i=x_t\} \xrightarrow{p} \PP(X=x_t)\), jointly over the finitely many support points. Moreover, the probability that any of these cells is empty tends to zero. Apply the conditional frequency expansion used in the proof
of Proposition~\ref{prop:bccross_probest2}(a), now centering
at $\PP_{\alpha_n(h)}(Y=1|x_t)$.
The centered summands are uniformly bounded, the cell
probabilities of support points remain fixed and positive, and their covariance
matrices converge to the baseline covariance matrix.
The multivariate triangular array CLT and Slutsky's theorem
therefore give
\[
\sqrt n\bigl(
\widehat \PP(Y=1|x_t)-\PP_{\alpha_n(h)}(Y=1|x_t)
\bigr)_{t=1}^T
\stackrel d\longrightarrow \mathcal N(0,\Sigma).
\]
Adding the deterministic mean shift from the preceding
expansion proves the result.$\blacksquare$ 

\smallskip

\noindent \textbf{Proof of Theorem  \ref{th:MSdrift}.}
Suppose $\mathcal X_0\neq\varnothing$ and fix an admissible
$h$. Lemma~\ref{lemmaLdriftasymp} gives the joint limit
$Z+\mu(h)$ for the scaled empirical contrasts.
The strict baseline restrictions outside $\mathcal X_0$ remain
separated from zero along $\alpha_n(h)$ and are therefore
classified correctly with probability tending to one. The
 reduction to a finite number of regions in the proof of
Theorem~\ref{th:MSgeneral2} therefore  continues to apply.

The joint limits of the scaled score differences are
$V_\ell(h,Z)-V_k(h,Z)$. Define $c_{\ell k}
=
\frac12
\bigl(
\PP(X=x_t)(s_{\ell t}-s_{kt})
\bigr)_{t=1}^T$. Distinct candidate regions have distinct sign profiles,
so $c_{\ell k}\neq0$ for $\ell\neq k$. Since $\Sigma$
is positive definite,
$\operatorname{Var}
\{V_\ell(h,Z)-V_k(h,Z)\}
=
c_{\ell k}'\Sigma c_{\ell k}>0$. The local drift changes the mean but not this variance.
Hence each limiting score difference has a continuous,
nondegenerate Gaussian distribution, and the limiting
maximizer is unique almost surely. The probability of
a sample tie among  scores of candidate regions tends to zero.
The continuous mapping theorem then gives
$\PP_{\alpha_n(h)}\{\widehat A_n=A_\ell\}
\rightarrow\pi_\ell(h)$, as well as the stated finite region approximation and
weak convergence.

For any $\ell$, the event
$\{M(h,Z)=\ell\}$ is described by finitely many strict linear
inequalities in $Z+\mu(h)$.  Its boundary has Gaussian probability
zero.  Since $\mu(h)$ is continuous in $h$, it follows that
$\pi_\ell(h)$ is continuous.  Hence
$\pi_\ell(h)\to\pi_\ell(0)$ for every $\ell$ as $h\to0$,
which is equivalent to
$\mathscr R_h \stackrel{W}{\rightarrow} \mathscr  R_0$.

If $\mathcal X_0=\varnothing$, all relevant population inequalities
remain strict in a neighborhood of $\alpha$, and the result is
immediate.
\hfill$\blacksquare$

\noindent\textbf{Proof of Theorem \ref{th:QRSEdrift}.}
Suppose first that $\mathcal X_0\neq\varnothing$.
The baseline membership function $p_0=p_\infty$ satisfies
\eqref{eq:membership_separation}. By
Theorem~\ref{th:MSdrift},
$\sup_{u\in\mathcal A}|p_h(u)-p_0(u)|
\leq
\sum_{\ell=1}^L|\pi_\ell(h)-\pi_\ell(0)|
\rightarrow0$ as $h\to0$. Fix $\tau\in(1/2,1)$ and choose
$0<\varepsilon<\min\{1-\tau,\tau-1/2\}$.
There is a neighborhood $U_\tau$ of zero such that the
preceding supremum is less than $\varepsilon$ for every
admissible $h\in U_\tau$. Hence $p_h$ remains uniformly
above $\tau$ on $\overline{\mathcal A_0(\alpha)}$
and uniformly below $\tau$ on its complement, so
$Q_\tau(h)=\overline{\mathcal A_0(\alpha)}$.

For each fixed admissible $h\in U_\tau$, the approximation by a finite number of regions 
and convergence of the 
probabilities for region selections in Theorem~\ref{th:MSdrift} imply  $\sup_{u\in\mathcal A}|p_{n,h}(u)-p_h(u)|\rightarrow0$.
This follows from the bound used in the proof of Theorem~\ref{th:estimatorafterQRSE}. Since $p_h$ is uniformly separated from $\tau$ on both
$\overline{\mathcal A_0(\alpha)}$ and its complement,  $\exists \; N(h,\tau)<\infty$ such that $Q_{n,\tau}(h)=\overline{\mathcal A_0(\alpha)}$ for all $
n\geq N(h,\tau)$. In particular, the asserted Hausdorff convergence holds. The required sample size may depend on $h$ and $\tau$ and we do not assert any 
 uniformity over local directions.

If $\mathcal X_0=\varnothing$, Theorem~\ref{th:MSdrift}
gives, for every fixed admissible $h$,
$\sup_{u\in\mathcal A}
\left|p_{n,h}(u)
-\mathbf 1\{u\in\overline{\mathcal A_0(\alpha)}\}\right|
\rightarrow0$. The same threshold argument completes the proof.
$\blacksquare$

\subsection{Proof of Theorem \ref{th:bootstrap_consistency}} 
Suppose first that $\mathcal X_0\neq\varnothing$.  Let $\pi_{\ell,n}^*
=
\PP^*\!\left(
\overline{\widehat A^{*}}=\overline A_\ell
\right)$, $\ell=1,\ldots,L$, where $\PP^*$ denotes probability conditional on the original sample.

At the ambiguous support points,
\[
\sqrt m\bigl(
\widehat \PP^*(Y=1\mid x_t)-\gamma
\bigr)
=
\sqrt m\bigl(
\widehat \PP^*(Y=1\mid x_t)
-\widehat \PP(Y=1\mid x_t)
\bigr)
+
\sqrt{\frac mn}\,
\sqrt n\bigl(
\widehat \PP(Y=1\mid x_t)-\gamma
\bigr).
\]
Under the baseline DGP, $\sqrt n\bigl(
\widehat \PP(Y=1\mid x_t)-\gamma
\bigr)_{t=1}^T
=O_p(1)$, and hence, since $m/n\to0$, the second term is $o_p(1)$.
The conditional bootstrap CLT therefore gives $\sqrt m\bigl(
\widehat \PP^*(Y=1\mid x_t)-\gamma
\bigr)_{t=1}^T
\stackrel{d}{\longrightarrow^*}
Z$ in probability, where $Z$ is the same centered Gaussian vector
that determines the limiting region probabilities in
Theorem~\ref{th:MSgeneral2}.

At $x\notin\mathcal X_0$ the population inequalities are strict
and the corresponding bootstrap regimes are selected correctly
with probability approaching one.  Since the limiting maximum of
the score over the finite collection
$\overline A_1,\ldots,\overline A_L$ is unique almost surely,
the continuous mapping theorem gives $\pi_{\ell,n}^*
\stackrel p\longrightarrow
\pi_\ell$, $\ell=1,\ldots,L$,
where $\pi_\ell$ are the limiting probabilities in
Theorem~\ref{th:MSgeneral2}.  

The same argument about a finite number of regions gives $\PP^*\{\overline{\widehat A^*}
\notin\{\overline A_1,\ldots,\overline A_L\}\}
\stackrel p\longrightarrow0$.
Let
$\widehat\pi_{\ell,B}
=
\frac1B\sum_{b=1}^B
\mathbf 1\{
\overline{\widehat A^{(b)}}=\overline A_\ell
\}$, and let $\widehat r_B$ denote the fraction of bootstrap estimates
whose closures are outside
$\{\overline A_1,\ldots,\overline A_L\}$.
Conditional independence of the bootstrap replications, $B\to\infty$, and the preceding vanishing probability of estimates outside the candidate collection imply 
$\sum_{\ell=1}^L
|\widehat\pi_{\ell,B}-\pi_{\ell,n}^*|
+
\widehat r_B
\stackrel p\longrightarrow0$, and, hence, 
$\sum_{\ell=1}^L
|\widehat\pi_{\ell,B}-\pi_\ell|
+
\widehat r_B
\stackrel p\longrightarrow0$.

The bootstrap coverage function can be written as $\widehat p(u)
=
\sum_{\ell=1}^L
\widehat\pi_{\ell,B}
\mathbf 1\{u\in\overline A_\ell\}
+
R_{n,B}(u)$, $0\leq R_{n,B}(u)\leq\widehat r_B$,
and $\sup_{u\in\mathcal A}
|\widehat p(u)-p_\infty(u)|
\leq
\sum_{\ell=1}^L
|\widehat\pi_{\ell,B}-\pi_\ell|
+\widehat r_B\stackrel p\longrightarrow0$. For fixed $\tau\in(1/2,1)$, choose $0<\varepsilon<\min\{1-\tau,\tau-1/2\}$. On the event $\sup_{u\in\mathcal A}|\widehat p(u)-p_\infty(u)|<\varepsilon$, 
the membership separation condition \eqref{eq:membership_separation} implies
$\widehat A^{boot}_{RSQ,\tau}=\overline{\mathcal A_0}$.
This event has probability tending to one. Thus
$\PP\!\left(\widehat A^{boot}_{RSQ,\tau}=\overline{\mathcal A_0}\right)\rightarrow1$, 
which implies the stated Hausdorff convergence. 

If $\mathcal X_0=\varnothing$, the strict population inequalities
are selected correctly by both the original sample and the
bootstrap with probability approaching one.  Thus
$\PP^*\!\left(
\overline{\widehat A^*}
=
\overline{\mathcal A_0}
\right)
\stackrel p\longrightarrow1,$ and the same conclusion follows for any $\tau\in(0,1)$. $\blacksquare$

\subsection{Proof of Theorem~\ref{th:bootstrap_robustness}}

Suppose first that $\mathcal X_0\neq\varnothing$.
For every fixed admissible $h$,
Lemma~\ref{lemmaLdriftasymp} gives
$\sqrt n(\widehat{\mathbf p}-\gamma\mathbf 1_T)
=O_p(1)$
under $\PP_{\alpha_n(h)}$, where $\widehat{\mathbf p}$ collects the estimated
conditional choice probabilities at the baseline ties.

Use the bootstrap decomposition in the proof of
Theorem~\ref{th:bootstrap_consistency}. Its empirical
centering term is
$\sqrt{m/n}\sqrt n(\widehat{\mathbf p}-\gamma\mathbf 1_T)
=o_p(1)$.
The conditional bootstrap CLT continues to have covariance
$\Sigma$, since the conditional choice probabilities at
these cells converge to $\gamma$. Hence, 
$\sqrt m(\widehat{\mathbf p}^{\,*}-\gamma\mathbf 1_T)
\Rightarrow^* Z$ in probability under $\PP_{\alpha_n(h)}$.

The strict baseline inequalities remain separated from zero
along the local path. The argument with a finite number of regions then
gives conditional selection probabilities converging to
$\pi_\ell(0)$, rather than $\pi_\ell(h)$. The bootstrap
aggregation argument in the proof of
Theorem~\ref{th:bootstrap_consistency} consequently gives  $\sup_{u\in\mathcal A}
|\widehat p(u)-p_0(u)|
\stackrel p\longrightarrow0$ under $\PP_{\alpha_n(h)}$. Fix $\tau\in(1/2,1)$ and choose
$0<\varepsilon<\min\{1-\tau,\tau-1/2\}$.
On the event that the preceding supremum is less than
$\varepsilon$, membership separation implies $\widehat A^{boot}_{RSQ,\tau}
=
\overline{\mathcal A_0(\alpha)}$. The probability of this event tends to one, which proves
the asserted Hausdorff convergence.

If $\mathcal X_0=\varnothing$, all baseline inequalities
remain strict for sufficiently large $n$. The bootstrap
coverage function then converges uniformly in probability
to the indicator of $\overline{\mathcal A_0(\alpha)}$,
and the same threshold argument applies.
$\blacksquare$

\clearpage 
\appendix
\setcounter{page}{1}
\numberwithin{equation}{section}

\section*{Computational Supplement}
\label{app:computation}

We compute the maximum score and RSQ estimates using
650 constituencies, of which 561 retained the same
winning party between 2015 and 2019. The empirical
section describes the covariates and their construction.
We use the normalization
$a=(a_0,a_1,-1,a_3,a_4,a_5)'$ and restrict each of
the five free coefficients to $[-5,5]$ for the
RSQ computation. 

\subsection*{Sample support and the maximum score estimate}

There are 23 observed discrete covariate values, including the
intercept. Figure~\ref{fig:Empiriccs_preliminary} displays the
sample conditional choice probabilities
\[
\widehat \PP(Y=1|x)
=\frac{\sum_i Y_i1(x_i=x)}{\sum_i1(x_i=x)}
\]
at these values. Each support point contributes once to the
histogram, regardless of the number of constituencies it contains.
The second panel plots these probabilities against the number
of constituencies at each support point. 

\begin{figure}[!ht]
\centering
\includegraphics[width=\linewidth]{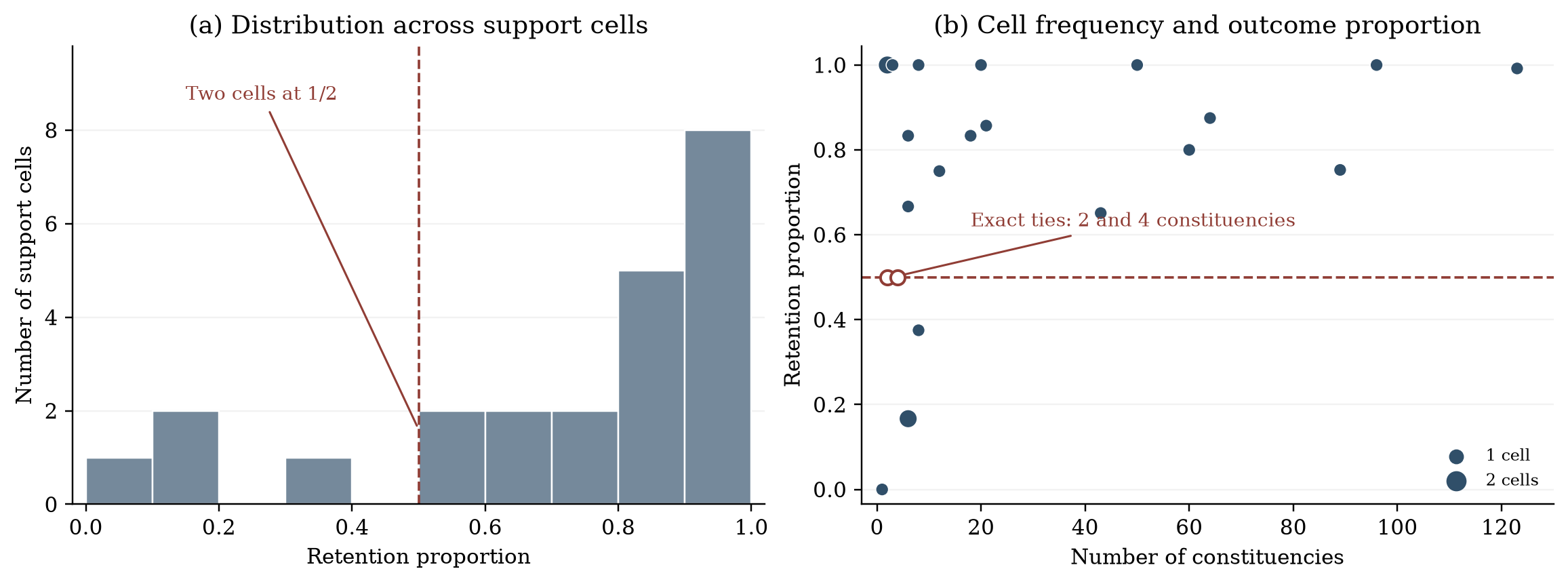}
\caption{Empirical outcome proportions at the 23 covariate support
points. The left panel gives an unweighted histogram of retention
proportions, using bins of width 0.1. Bins include their left
endpoint, and the final bin includes one. The right panel plots
each retention proportion against the number of constituencies
at that support point. Marker area indicates the number of
support points at coincident coordinates. Dashed lines mark
one half. The two exact ties occur in cells containing two and
four constituencies.}
\label{fig:Empiriccs_preliminary}
\end{figure}

There are 17 cells with $\widehat \PP(Y=1|x)>1/2$, four with
$\widehat \PP(Y=1|x)<1/2$, and two with $\widehat \PP(Y=1|x)=1/2$.
The latter contain one success in two observations and two successes in 
four observations.

Table~\ref{tab:empirical-support} reports the full empirical distribution.
The intercept is one in every cell and is omitted from the table.
The five displayed covariates are the Leave indicator, the indicator for a close previous election, income growth category, Labour indicator and Leave--Labour
interaction $X_5=X_1X_4$, respectively.

\begin{table}[!ht]
\centering
\caption{Covariate support, cell frequencies and party-retention counts.}
\label{tab:empirical-support}
\begin{tabular}{rrrrrrrrrr}
\hline
Cell & $X_1$ & $X_2$ & $X_3$ & $X_4$ & $X_5$
& \shortstack{Number of\\constituencies}
& \shortstack{Same\\party}
& \shortstack{Changed\\party}
& \shortstack{Retention\\proportion}\\ 
\hline
1 & 0 & 0 & 0 & 0 & 0 & 21 & 18 & 3 & 0.8571\\
2 & 0 & 0 & 0 & 1 & 0 & 3 & 3 & 0 & 1.0000\\
3 & 0 & 0 & 1 & 0 & 0 & 64 & 56 & 8 & 0.8750\\
4 & 0 & 0 & 1 & 1 & 0 & 20 & 20 & 0 & 1.0000\\
5 & 0 & 0 & 2 & 0 & 0 & 60 & 48 & 12 & 0.8000\\
6 & 0 & 0 & 2 & 1 & 0 & 50 & 50 & 0 & 1.0000\\
7 & 0 & 1 & 0 & 0 & 0 & 2 & 1 & 1 & 0.5000\\
8 & 0 & 1 & 1 & 0 & 0 & 8 & 3 & 5 & 0.3750\\
9 & 0 & 1 & 1 & 1 & 0 & 2 & 2 & 0 & 1.0000\\
10 & 0 & 1 & 2 & 0 & 0 & 6 & 1 & 5 & 0.1667\\
11 & 0 & 1 & 2 & 1 & 0 & 8 & 8 & 0 & 1.0000\\
12 & 1 & 0 & 0 & 0 & 0 & 18 & 15 & 3 & 0.8333\\
13 & 1 & 0 & 0 & 1 & 1 & 6 & 4 & 2 & 0.6667\\
14 & 1 & 0 & 1 & 0 & 0 & 123 & 122 & 1 & 0.9919\\
15 & 1 & 0 & 1 & 1 & 1 & 43 & 28 & 15 & 0.6512\\
16 & 1 & 0 & 2 & 0 & 0 & 96 & 96 & 0 & 1.0000\\
17 & 1 & 0 & 2 & 1 & 1 & 89 & 67 & 22 & 0.7528\\
18 & 1 & 1 & 0 & 0 & 0 & 2 & 2 & 0 & 1.0000\\
19 & 1 & 1 & 0 & 1 & 1 & 1 & 0 & 1 & 0.0000\\
20 & 1 & 1 & 1 & 0 & 0 & 6 & 5 & 1 & 0.8333\\
21 & 1 & 1 & 1 & 1 & 1 & 6 & 1 & 5 & 0.1667\\
22 & 1 & 1 & 2 & 0 & 0 & 12 & 9 & 3 & 0.7500\\
23 & 1 & 1 & 2 & 1 & 1 & 4 & 2 & 2 & 0.5000\\
\hline
\multicolumn{6}{l}{Total} & 650 & 561 & 89 & 0.8631\\
\hline
\end{tabular}
\par\smallskip
\begin{minipage}{0.96\linewidth}\footnotesize
\textit{Notes:} ``Same party'' counts constituencies retaining
the same winning party. ``Changed party'' counts constituencies
changing party. Counts are exact, proportions are rounded to
four decimal places. Relative cell frequencies are obtained by
dividing the number of constituencies in each cell by 650.
The final proportion is the overall retention rate, not the
unweighted mean of the cell proportions. 
\end{minipage}
\end{table}

The table specifies the empirical distribution used
for resampling. The replication package preserves
the observation ordering and random-number settings
needed to reproduce the saved draws. Rebuilding
the analysis matrix reproduces the table entries;
the checks for the five estimation runs are
reported below. 

For this sample, all 21 preferred index signs can be satisfied
simultaneously. Hence the full sample maximum score estimate
consists exactly of the normalized parameter values satisfying
\[
\begin{aligned}
x'a&\geq0
&&\text{at every support point with }\widehat \PP(Y=1|x)>1/2,\\
x'a&<0
&&\text{at every support point with }\widehat \PP(Y=1|x)<1/2.
\end{aligned}
\]
Table~\ref{tab:empirical-support} gives the covariate values
defining all 17 weak and four strict inequalities.
The weak and strict inequalities follow from assigning a zero
index to outcome one. The two support points with
$\widehat \PP(Y=1|x)=1/2$ impose no restriction on the maximizing set.

The inequalities can all be satisfied with a positive
margin, so the full sample estimate has nonempty
interior in the five free coordinates. They also
allow the index comparisons in the empirical section
to be checked directly. The bootstrap optimization
uses the coefficient box without imposing these
full sample sign restrictions.

\subsection*{Bootstrap computation}

The main calculation uses $m=550$, $B=19{,}999$ and MATLAB's twister seed
1234567. Each draw samples $m$ observations independently with replacement
from the 650 rows. This resampling scheme treats constituencies as i.i.d. and does not represent dependence from common or geographic electoral shocks.
In each draw, we maximize the binary choice maximum score criterion
from Section~\ref{sec:theoreticalMSbinary}, with $\gamma=1/2$,
over the coefficient box. For computation, each support point
contributes its number of observations with outcome 1 minus its
number of observations with outcome 0, multiplied by
$\operatorname{sgn}(x'a)$. This is a positive rescaling of the
criterion and permits comparison using integer scores.
Empty cells and cells with equal outcome counts contribute zero.
Every feasible region attaining the largest score is retained,
including all regions tied for that value. 

We enumerate strict sign patterns within the
coefficient box by maximizing a common index margin
subject to the box constraints. Retaining patterns
with margin above $10^{-8}$ gives 8,754 regions.
The computation stops if the solver fails; patterns
reported as infeasible are discarded.

These strict sign regions suffice to recover the
closure of the maximizing set. To see this, move
a boundary maximizer a sufficiently small distance
toward $(2,0,-1,0,0,0)'$. At this reference vector,
every index equals $2-X_2>0$. The move stays inside
the box, makes zero indices positive, and preserves
all other signs. Since zero indices are assigned
to outcome one, the score is unchanged. Thus every
boundary maximizer is a limit of maximizers in the
enumerated regions.

For each draw, we take the union of the closures of
all maximizing regions to obtain
$\overline{\widehat A^{(b)}}$ in
\eqref{eq:bootstrap_coverage}. Each parameter value
counts once per draw, even if it belongs to several
maximizing regions. The RSQ estimate retains points
in $\mathcal A_5$ whose membership frequency is at
least $\tau$, as in \eqref{eq:feasible_RSQ}.
Taking closures allows boundary points to enter
the estimate. 

To localize the search, we use the vertices of the
index-hyperplane arrangement intersected with the
coefficient box, including vertices on the box facets.
Membership is constant on each relatively open
stratum. Any bootstrap estimate containing a stratum
also contains its closure, so membership at each
boundary vertex is at least as large as on the
stratum. Every retained stratum therefore lies in
the convex hull of vertices meeting the threshold.
We restrict the search to an outer set containing
this hull, then compute the RSQ estimate as the
union of retained strata.

\subsection*{The threshold sequence and its geometry}

Table~\ref{tab:complete-rsq-plateaus} lists every change
in the estimate over the computed threshold range.
The reported pieces are maximal closed polyhedra,
which may intersect. We report both the largest
piece dimension and the affine dimension of their
union.

\begin{table}[!ht]
\centering
\caption{Complete finite bootstrap RSQ threshold sequence.}
\label{tab:complete-rsq-plateaus}
\begin{tabular}{rrccc}
\hline
Lower count & Upper count & Maximal & Largest piece & Union affine\\
 & & pieces & dimension & dimension\\
\hline
-- & 10,853 & 2 & 4 & 5 \\
10,853 & 10,856 & 4 & 3 & 5 \\
10,856 & 11,184 & 3 & 3 & 5 \\
11,184 & 11,187 & 4 & 3 & 5 \\
11,187 & 11,219 & 3 & 3 & 5 \\
11,219 & 11,561 & 3 & 3 & 5 \\
11,561 & 11,564 & 3 & 3 & 5 \\
11,564 & 11,701 & 2 & 3 & 4 \\
11,701 & 12,319 & 4 & 2 & 4 \\
12,319 & 13,148 & 3 & 2 & 4 \\
13,148 & 13,743 & 2 & 2 & 3 \\
13,743 & 14,366 & 2 & 2 & 3 \\
14,366 & 17,352 & 1 & 2 & 2 \\
17,352 & 17,978 & 1 & 1 & 1 \\
17,978 & -- & 0 & -- & -- \\
\hline
\end{tabular}
\par\smallskip
\begin{minipage}{0.96\linewidth}\footnotesize
\textit{Notes:} $B=19{,}999$, $m=550$. For interior rows, $\tau$
is strictly greater than the lower count divided by $B$ and no
greater than the upper count divided by $B$. The first row applies
to $0.5000001\leq\tau\leq10853/B$ and the last to
$17978/B<\tau<1$. Integer counts give exact endpoints.
Dimensions are measured in the five free coordinates.
The last row is the empty set. 
\end{minipage}
\end{table}

The geometry can be described using the three triangles defined in
the empirical section:
\begin{align*}
\mathcal C=\{a\in\mathcal A_5:\;&
a_0=1,\ a_3=0,\ a_5=-a_1-a_4,\ 
a_1\geq0,\ a_4\geq0,\ a_1+a_4\leq5\},\\
\mathcal T_1=\{a\in\mathcal A_5:\;&
a_0=0,\ a_3=\tfrac12,\ a_5=-a_1-a_4,\ 
a_1\geq1,\ a_4\geq\tfrac12,\ a_1+a_4\leq5\},\\
\mathcal T_2=\{a\in\mathcal A_5:\;&
a_0=1,\ a_3=0,\ a_5=-1-a_1-a_4,\ 
a_1\geq0,\ a_4\geq0,\ a_1+a_4\leq4\}.
\end{align*}
The estimate is $\mathcal C\cup\mathcal T_1\cup\mathcal T_2$ for
$12319/B<\tau\leq13148/B$, and $\mathcal C\cup\mathcal T_2$ for
$13148/B<\tau\leq13743/B$. Each triangle has dimension two; the
respective unions have affine dimensions four and three.
For $14366/B<\tau\leq17352/B$, the entire estimate is $\mathcal C$.
The bound $a_1+a_4\leq5$ on this core follows from the imposed
lower bound $a_5\geq-5$. 

The core interval follows directly from the membership
counts over all strata. The minimum count on
$\mathcal C$ is 17,352, and the maximum count outside
$\mathcal C$ is 14,366. Hence the estimate equals
$\mathcal C$ precisely for
\[
\frac{14366}{19999}<\tau\leq\frac{17352}{19999},
\qquad\text{approximately }(0.718336,0.867643].
\]
For $17352/19999<\tau\leq17978/19999$, only the edge
$\mathcal C\cap\{a:a_1=0\}$ remains. Above this range,
the estimate is empty.

\subsection*{Sensitivity and numerical verification}

We repeat the calculation for
$m\in\{325,450,550,580\}$, using $B=5{,}000$ and seed
1234567. The sample, normalization, coefficient box,
and numerical tolerance are unchanged. We reuse the
same arrangement after checking the support and
bounds, and recompute the bootstrap draws and
membership counts for each $m$.

For each run, the core interval has lower endpoint
equal to the largest membership frequency outside
$\mathcal C$ and upper endpoint equal to the smallest
membership frequency on $\mathcal C$. When the lower
endpoint is smaller than the upper endpoint, the
estimate equals $\mathcal C$ throughout this
left-open, right-closed interval.
Table~\ref{tab:rsq-m-sensitivity} reports the intervals
and the maximum membership frequency in each run.

\begin{table}[!ht]
\centering
\caption{Sensitivity of the core interval to the resample size.}
\label{tab:rsq-m-sensitivity}
\begin{tabular}{rrcrr}
\hline
$m$ & $B$ & Core interval & Width & Maximum\\
 & & & & membership frequency\\ 
\hline
325 & 5,000 & $(0.7374,0.7814]$ & 0.0440 & 0.8516\\
450 & 5,000 & $(0.7154,0.8442]$ & 0.1288 & 0.8840\\
550 & 5,000 & $(0.7212,0.8708]$ & 0.1496 & 0.9022\\
580 & 5,000 & $(0.7188,0.8752]$ & 0.1564 & 0.9030\\
\hline
550 & 19,999 & $(0.718336,0.867643]$ & 0.149307 & 0.898945\\
\hline
\end{tabular}
\par\smallskip
\begin{minipage}{0.96\linewidth}\footnotesize
\textit{Notes:} The first four rows use 5,000 draws; their displayed
endpoints are exact frequencies. The last row reports the main run with
rounded endpoints. Intervals are open on the left and closed on the right.
The estimate is empty above its maximum coverage.
\end{minipage}
\end{table}

The core triangle $\mathcal C$ appears in all four
sensitivity runs. Its threshold interval is narrower
at $m=325$, but all four intervals contain
$(0.7374,0.7814]$. Thus $\tau=0.75$ gives the same
triangle throughout the grid. At $\tau=0.80$, the
estimate for $m=325$ is an edge, whereas the other
three estimates remain equal to $\mathcal C$.

These comparisons concern the reported resample sizes
and the fixed coefficient box $\mathcal A_5$. In the
main run, $m/n=550/650\simeq0.846$; the asymptotic
condition $m/n\to0$ does not give a finite-sample
accuracy guarantee for this choice. The threshold
$\tau$ is a bootstrap membership threshold, not a
confidence level.

All five runs use MATLAB R2019b with Optimization
Toolbox and numerical tolerance $10^{-8}$. Each
calculation gives 18 localized vertices and 114 strata.
We checked union membership at 14 thresholds.
Separate Python calculations reproduce the maximizing
regions for all 39,999 draws, the vertex and stratum
counts, and the checked membership values. A separate
face enumeration also gives 114 strata.

\end{document}